\documentclass[acmsmall,screen,nonacm,final]{acmart} 
\pdfoutput=1
\usepackage[T1]{fontenc}
\usepackage{graphicx}
\usepackage{tikz}
\usetikzlibrary{calc}
\usetikzlibrary{positioning}
\usetikzlibrary{arrows,automata}
\usetikzlibrary{patterns,patterns.meta}
\usepackage{xspace}
\usepackage{xcolor}
\usepackage{amsmath}
\usepackage{thmtools}
\usepackage{thm-restate}
\usepackage{mathrsfs}
\usepackage{ifdraft}
\usepackage{wrapfig}
\usepackage{subcaption}
\usepackage{verbatim}
\usepackage{booktabs}
\usepackage{multirow}
\usepackage{paralist}
\usepackage{lipsum}
\usepackage{amsbsy}
\usepackage[normalem]{ulem}
\usepackage{textcomp}
\usepackage{stmaryrd}
\usepackage{mathpartir}
\usepackage{keyval}
\usepackage{algorithm}
\usepackage{algorithmicx}
\usepackage[noend]{algpseudocode}
\usepackage{makecell}
\RequirePackage{centernot}
\RequirePackage[Symbol]{upgreek}
\RequirePackage{etoolbox}
\RequirePackage{mathtools}
\newcommand{\formulas}{\mathcal{F}}
\newcommand{\SymProt}{\mathbb{S}}
\newcommand{\unreach}{\mathsf{unreach}}
\newcommand{\prodreach}{\mathsf{prodreach}}
\newcommand{\avail}{\mathsf{avail}}
\newcommand{\Characterization}{\emph{CC}\xspace}
\newcommand{\Kleenestar}{\mathrel{\vphantom{\to}^*}}
\let\originalleft\left
\let\originalright\right
\renewcommand{\left}{\mathopen{}\mathclose\bgroup\originalleft}
\renewcommand{\right}{\aftergroup\egroup\originalright}
\newcommand{\pvec}[1]{\vec{#1}\mkern2mu\vphantom{#1}}

\newcommand{\Procs}{\ensuremath{\mathcal{P}}}
\definecolor{roleColor}{rgb}{0.1, 0.3, 0.1}
\newcommand{\roleCol}[1]{{\color{roleColor}#1}}

\newcommand{\roleFmt}[1]{\roleCol{\mathtt{#1}}}
\newcommand{\roleP}[1][]{
	\ifempty{#1}{{\color{roleColor}\roleFmt{p}}}{{\color{roleColor}\roleFmt{p}_{#1}}}
}

\newcommand{\roleVar}[1][]{
	\ifempty{#1}{{\color{roleColor}\roleFmt{x}}}{{\color{roleColor}\roleFmt{x}_{#1}}}
}

\newcommand{\procA}{{\roleFmt{p}}}
\newcommand{\procB}{{\roleFmt{q}}}
\newcommand{\procC}{{\color{roleColor}\roleFmt{r}}}
\newcommand{\procD}{{\color{roleColor}\roleFmt{s}}}
\newcommand{\procE}{{\color{roleColor}\roleFmt{t}}}
\newcommand{\run}{\rho}
\newcommand{\val}{\ensuremath{m}}
\newcommand{\MsgVals}{\ensuremath{\mathcal{V}}}
\newcommand{\lblmsgO}{\textcolor{orange}{\mathsf{o}}}
\newcommand{\lblmsgB}{\textcolor{blue}{\mathsf{b}}}
\newcommand{\lblmsgM}{\textcolor{magenta}{\mathsf{m}}}

\newcommand{\msgO}{\lblmsgO}
\newcommand{\msgB}{\lblmsgB}
\newcommand{\msgM}{\lblmsgM}

\newcommand{\CSM}[1]{\ensuremath{\{\!\!\{#1_\procA\}\!\!\}_{\procA \in \Procs}}}
\newcommand{\CSMl}[1]{\ensuremath{\{\!\!\{{#1}\}\!\!\}_{\procA \in \Procs}}}
\newcommand{\CLTS}[1]{\ensuremath{\{\!\!\{#1_\procA\}\!\!\}_{\procA \in \Procs}}}

\newcommand{\emptystring}{\varepsilon}

\newcommand{\set}[1]{\{#1\}}
\newcommand{\cond}[1]{\{#1\}}
\newcommand{\initUpdLoop}[2]{\ensuremath{\langle #1 \rangle \is \langle #2 \rangle}}

\newcommand{\lang}{\mathcal{L}}
\newcommand{\langasync}{\mathcal{L}_{async}}
\newcommand{\interswaplang}{\mathcal{C}^{\interswap}}
\newcommand{\SyncToAsync}{\ensuremath{\texttt{\upshape split}}}
\newcommand{\channels}{\ensuremath{\mathsf{Chan}}}
\newcommand{\channel}[2]{\ensuremath{#1,#2}}

\newcommand{\trace}{\texttt{\upshape trace}}

\newcommand{\GG}{\mathbf{G}}
\newcommand{\getMu}{\mathit{get\mu}}

\newcommand{\semglobal}{\ensuremath{\mathsf{GAut}}}
\newcommand{\semglobalsync}{\ensuremath{\mathsf{GAut}}}
\newcommand{\projerasuresymb}{\downarrow}

\newcommand{\AlphSync}{\ensuremath{\Gamma}}
\newcommand{\AlphSyncSubscript}{\ensuremath{\Gamma_{\negthinspace\mathit{sync}}}} 
\newcommand{\Alphabet}{\AlphAsync} 
\newcommand{\AlphAsync}{\ensuremath{Σ}}
\newcommand{\AlphAsyncSubscript}{\ensuremath{Σ_{\mathit{async}}}}

\newcommand{\fin}{\operatorname{fin}}

\newcommand{\infi}{\operatorname{inf}}

\newcommand{\interswap}{\ensuremath{\sim}}

\def \ifempty#1{\def\temp{#1} \ifx\temp\empty }

\newcommand{\snd}[3]{#1\triangleright#2!#3}
\newcommand{\rcv}[3]{#2\triangleleft#1?#3}

\newcommand{\msgFromTo}[3]{#1\!\to\!#2\!:\!#3}
\newcommand{\msgFromToNS}[3]{#1\to#2:#3}

\newcommand{\pref}{\operatorname{pref}}
\newcommand{\preforder}{\ensuremath{\leq}}

\newcommand{\myparagraph}[1]{\smallskip\noindent\textbf{#1}}
\newcommand{\proj}{{\ensuremath{\upharpoonright}}}

\newcommand{\wproj}{{\ensuremath{\Downarrow}}}

\def\mmerge{\mathrel{\ThisStyle{\stretchrel*{\ooalign{
					\raise0.2\LMex\hbox{$\SavedStyle\sqcap$}\cr
					\raise-0.2\LMex\hbox{$\SavedStyle\sqcap$}}}{\sqcap}}}}
\def\mmmerge{\mathrel{\ThisStyle{\stretchrel*{\ooalign{
					\raise0.6\LMex\hbox{$\SavedStyle\sqcap$}\cr
					\raise0.2\LMex\hbox{$\SavedStyle\sqcap$}\cr
					\raise-0.2\LMex\hbox{$\SavedStyle\sqcap$}}}{\sqcap}}}}

\newcommand{\blockedset}{\ensuremath{\mathcal{B}}}
\newcommand{\semavail}{\ensuremath{M}}

\newcommand{\union}{\cup}
\newcommand{\inters}{\cap}
\newcommand{\Union}{\bigcup}
\newcommand{\Inters}{\bigcap}
\newcommand{\dunion}{\uplus}

\DeclarePairedDelimiter\card{\lvert}{\rvert}

\providecommand{\implies}{\Rightarrow}

\providecommand{\coloneq}{\mathrel{\mathop:}=} 
\providecommand{\Coloneqq}{\mathrel{\mathop{::}}=} 
\newcommand{\is}{\coloneq}

\newcommand{\from}{\colon}

\def\st{\colon}  

\newcommand{\ie}{i.e.~}
\newcommand{\eg}{e.g.~}
\def\grammOr{\hspace{3pt}\mid\hspace{3pt}}
\def\grammIs{\Coloneqq}
\makeatletter
\begingroup
\catcode`\|=\active
\gdef\@grammar@bar{
	\catcode`\|=\active
	\def|{\grammOr}
}
\endgroup
\newcommand{\gramm}[1]{
	\begingroup
	\def\is{\grammIs}
	\@grammar@bar
	#1
	\endgroup
}
\newenvironment{grammar}{
	\begin{equation*}
		\def\is{& \grammIs }
		
		\@grammar@bar
		\aligned
	}
	{
		\endaligned
	\end{equation*}
	\aftergroup\ignorespaces
}
\makeatother

\newcommand{\hole}{\hbox{-}}

\newcommand{\subsetproj}[2]{\ensuremath{\mathscr{P}(#1,#2)}}

\newcommand{\globcomplocal}[3]{\ensuremath{\operatorname{R}^#1_{#2}(#3)}}
\newcommand{\denotations}[1]{\llbracket {#1} \rrbracket}
\newcommand{\domain}{\mathfrak{D}}

\newcommand{\CommentLine}[1]{
	\State \textcolor{blue}{$\triangleright$ \textit{#1}}
}
\newcommand{\seller}{\ensuremath{\roleFmt{S}}}
\newcommand{\buyerOne}{\ensuremath{\roleFmt{B_1}}}
\newcommand{\buyerTwo}{\ensuremath{\roleFmt{B_2}}}
\newcommand{\isbn}{\ensuremath{\mathsf{ISBN}}}
\newcommand{\succMsg}{\ensuremath{\mathsf{succ}}}
\newcommand{\contMsg}{\ensuremath{\mathsf{cont}}}
\newcommand{\quitMsg}{\ensuremath{\mathsf{quit}}}
\usetikzlibrary{arrows.meta}
\tikzset{
  trans/.style={
    draw,-{Stealth[round]}, semithick, shorten >= 1pt,
  },
  init/.style={initial by arrow},
  final/.style={accepting},
  initial text={},
  initial distance=2ex,
  every initial by arrow/.style={trans},
}
\tikzset{
  sem/.style={
     every state/.style = {semnode},
	 every edge/.style = {semarrow},
	 every loop/.style = {semarrow},
  }
}
\tikzset{
  semnode/.style={
    thick,
    draw,
    minimum size=1ex,
    shape=circle,
    font=\scriptsize,
    inner sep=2pt,
  },
  semarrow/.style={
    trans,
    font=\scriptsize,
    draw,
    pos=.4,
  }
}
 \usepackage{hmsc-tikz}
\usepackage{newunicodechar}
\newunicodechar{∃}{\ensuremath{\exists}}
\newunicodechar{∀}{\ensuremath{\forall}}
\newunicodechar{θ}{\ensuremath{\theta}}
\newunicodechar{τ}{\ensuremath{\tau}}
\newunicodechar{φ}{\ensuremath{\varphi}}
\newunicodechar{ξ}{\ensuremath{\xi}}
\newunicodechar{ζ}{\ensuremath{\zeta}}
\newunicodechar{ψ}{\ensuremath{\psi}}
\newunicodechar{π}{\ensuremath{\pi}}
\newunicodechar{α}{\ensuremath{\alpha}}
\newunicodechar{β}{\ensuremath{\beta}}
\newunicodechar{γ}{\ensuremath{\gamma}}
\newunicodechar{δ}{\ensuremath{\delta}}
\newunicodechar{ε}{\ensuremath{\varepsilon}}
\newunicodechar{κ}{\ensuremath{\kappa}}
\newunicodechar{λ}{\ensuremath{\lambda}}
\newunicodechar{μ}{\ensuremath{\mu}}
\newunicodechar{ρ}{\ensuremath{\rho}}
\newunicodechar{σ}{\ensuremath{\sigma}}
\newunicodechar{ω}{\ensuremath{\omega}}
\newunicodechar{Γ}{\ensuremath{\Gamma}}
\newunicodechar{Φ}{\ensuremath{\Phi}}
\newunicodechar{Δ}{\ensuremath{\Delta}}
\newunicodechar{Σ}{\ensuremath{\Sigma}}
\newunicodechar{Π}{\ensuremath{\Pi}}
\newunicodechar{∑}{\ensuremath{\Sigma}}
\newunicodechar{∏}{\ensuremath{\Pi}}
\newunicodechar{Θ}{\ensuremath{\Theta}}
\newunicodechar{Ω}{\ensuremath{\Omega}}
\newunicodechar{⇒}{\ensuremath{\Rightarrow}}
\newunicodechar{⇐}{\ensuremath{\Leftarrow}}
\newunicodechar{⇔}{\ensuremath{\Leftrightarrow}}
\newunicodechar{→}{\ensuremath{\rightarrow}}
\newunicodechar{←}{\ensuremath{\leftarrow}}
\newunicodechar{↔}{\ensuremath{\leftrightarrow}}
\newunicodechar{¬}{\ensuremath{\neg}}
\newunicodechar{∧}{\ensuremath{\land}}
\newunicodechar{∨}{\ensuremath{\lor}}
\newunicodechar{≠}{\ensuremath{\neq}}
\newunicodechar{≡}{\ensuremath{\equiv}}
\newunicodechar{∼}{\ensuremath{\sim}}
\newunicodechar{≈}{\ensuremath{\approx}}
\newunicodechar{≥}{\ensuremath{\geq}}
\newunicodechar{≤}{\ensuremath{\leq}}
\newunicodechar{≫}{\ensuremath{\gg}}
\newunicodechar{≪}{\ensuremath{\ll}}
\newunicodechar{∅}{\ensuremath{\emptyset}}
\newunicodechar{⊆}{\ensuremath{\subseteq}}
\newunicodechar{⊂}{\ensuremath{\subset}}
\newunicodechar{∩}{\ensuremath{\cap}}
\newunicodechar{⋂}{\ensuremath{\cap}}
\newunicodechar{∪}{\ensuremath{\cup}}
\newunicodechar{⋃}{\ensuremath{\cup}}
\newunicodechar{⊎}{\ensuremath{\uplus}}
\newunicodechar{∈}{\ensuremath{\in}}
\newunicodechar{∉}{\ensuremath{\not\in}}
\newunicodechar{⊤}{\ensuremath{\top}}
\newunicodechar{⊥}{\ensuremath{\bot}}
\newunicodechar{₀}{\ensuremath{_0}}
\newunicodechar{₁}{\ensuremath{_1}}
\newunicodechar{₂}{\ensuremath{_2}}
\newunicodechar{₃}{\ensuremath{_3}}
\newunicodechar{₄}{\ensuremath{_4}}
\newunicodechar{₅}{\ensuremath{_5}}
\newunicodechar{₆}{\ensuremath{_6}}
\newunicodechar{₇}{\ensuremath{_7}}
\newunicodechar{₈}{\ensuremath{_8}}
\newunicodechar{₉}{\ensuremath{_9}}
\newunicodechar{⁰}{\ensuremath{^0}}
\newunicodechar{¹}{\ensuremath{^1}}
\newunicodechar{²}{\ensuremath{^2}}
\newunicodechar{³}{\ensuremath{^3}}
\newunicodechar{⁴}{\ensuremath{^4}}
\newunicodechar{⁵}{\ensuremath{^5}}
\newunicodechar{⁶}{\ensuremath{^6}}
\newunicodechar{⁷}{\ensuremath{^7}}
\newunicodechar{⁸}{\ensuremath{^8}}
\newunicodechar{⁹}{\ensuremath{^9}}
\newunicodechar{𝔹}{\ensuremath{\mathbb{B}}}
\newunicodechar{ℝ}{\ensuremath{\mathbb{R}}}
\newunicodechar{ℕ}{\ensuremath{\mathbb{N}}}
\newunicodechar{ℂ}{\ensuremath{\mathbb{C}}}
\newunicodechar{ℚ}{\ensuremath{\mathbb{Q}}}
\newunicodechar{𝕋}{\ensuremath{\mathbb{T}}}
\newunicodechar{𝕏}{\ensuremath{\mathbb{X}}}
\newunicodechar{ℤ}{\ensuremath{\mathbb{Z}}}
\newunicodechar{✓}{\ding{51}}
\newunicodechar{✗}{\ding{55}}
\newunicodechar{◊}{\ensuremath{\lozenge}}
\newunicodechar{□}{\ensuremath{\square}}
\newunicodechar{𝓐}{\ensuremath{\mathcal{A}}}
\newunicodechar{𝓑}{\ensuremath{\mathcal{B}}}
\newunicodechar{𝓒}{\ensuremath{\mathcal{C}}}
\newunicodechar{𝓓}{\ensuremath{\mathcal{D}}}
\newunicodechar{𝓔}{\ensuremath{\mathcal{E}}}
\newunicodechar{𝓕}{\ensuremath{\mathcal{F}}}
\newunicodechar{𝓖}{\ensuremath{\mathcal{G}}}
\newunicodechar{𝓗}{\ensuremath{\mathcal{H}}}
\newunicodechar{𝓘}{\ensuremath{\mathcal{I}}}
\newunicodechar{𝓙}{\ensuremath{\mathcal{J}}}
\newunicodechar{𝓚}{\ensuremath{\mathcal{K}}}
\newunicodechar{𝓛}{\ensuremath{\mathcal{L}}}
\newunicodechar{𝓜}{\ensuremath{\mathcal{M}}}
\newunicodechar{𝓝}{\ensuremath{\mathcal{N}}}
\newunicodechar{𝓞}{\ensuremath{\mathcal{O}}}
\newunicodechar{𝓟}{\ensuremath{\mathcal{P}}}
\newunicodechar{𝓠}{\ensuremath{\mathcal{Q}}}
\newunicodechar{𝓡}{\ensuremath{\mathcal{R}}}
\newunicodechar{𝓢}{\ensuremath{\mathcal{S}}}
\newunicodechar{𝓣}{\ensuremath{\mathcal{T}}}
\newunicodechar{𝓤}{\ensuremath{\mathcal{U}}}
\newunicodechar{𝓥}{\ensuremath{\mathcal{V}}}
\newunicodechar{𝓦}{\ensuremath{\mathcal{W}}}
\newunicodechar{𝓧}{\ensuremath{\mathcal{X}}}
\newunicodechar{𝓨}{\ensuremath{\mathcal{Y}}}
\newunicodechar{𝓩}{\ensuremath{\mathcal{Z}}}
\newunicodechar{…}{\ensuremath{\ldots}}
\newunicodechar{∗}{\ensuremath{\ast}}
\newunicodechar{⊢}{\ensuremath{\vdash}}
\newunicodechar{⊧}{\ensuremath{\models}}
\newunicodechar{′}{\ensuremath{'}}
\newunicodechar{″}{\ensuremath{''}}
\newunicodechar{‴}{\ensuremath{'''}}
\newunicodechar{∥}{\ensuremath{\|}}
\newunicodechar{⊕}{\ensuremath{\oplus}}
\newunicodechar{⁺}{\ensuremath{^+}}
\newunicodechar{⊇}{\ensuremath{\supseteq}}
\newunicodechar{∘}{\ensuremath{\circ}}
\newunicodechar{∙}{\ensuremath{\cdot}}
\newunicodechar{⋅}{\ensuremath{\cdot}}
\newunicodechar{≈}{\ensuremath{\approx}}
\newunicodechar{×}{\ensuremath{\times}}
\newunicodechar{∞}{\ensuremath{\infty}}
\newunicodechar{⊑}{\ensuremath{\sqsubseteq}}
 \ifoptionfinal
{\usepackage[disable]{todonotes}}
{\usepackage{todonotes}}
\newcommand\fs[1]{\todo[color=blue!20!white]{FS, #1}}

\newcommand\efl[1]{\todo[color=purple!20!white]{EFL, #1}}

\newcommand\FS[1]{\todo[color=white,inline]{\textcolor{blue}{FS, #1}}}

\newcommand\EFL[1]{\todo[color=white,inline]{\textcolor{purple}{EFL, #1}}}

\ifoptionfinal
{}
{\paperwidth=\dimexpr \paperwidth + 6cm\relax
\oddsidemargin=\dimexpr\oddsidemargin + 3cm\relax
\evensidemargin=\dimexpr\evensidemargin + 3cm\relax
\marginparwidth=\dimexpr \marginparwidth + 3cm\relax}
\usepackage{hyperref}
\usepackage[capitalise]{cleveref}
\usepackage[shortlabels]{enumitem}
\crefformat{section}{\S#2#1#3} 
\crefmultiformat{section}{\S#2#1#3}{ and~\S#2#1#3}{, \S#2#1#3}{, and~\S#2#1#3} 
\crefrangeformat{section}{\S#3#1#4 to~\S#5#2#6} 
 
\AtBeginDocument{
  }
\begin{document}
\title
[Characterizing Implementability of Global Protocols with Infinite States and Data]
{Characterizing Implementability of Global Protocols \\ with Infinite States and Data}
\author{Elaine Li}
\authornote{corresponding author}
\affiliation{\institution{New York University}
	\country{USA}}
\author{Felix Stutz}
\affiliation{
  \institution{University of Luxembourg}
  \country{Luxembourg}
}
\author{Thomas Wies}
\affiliation{\institution{New York University}
	\country{USA}}
\author{Damien Zufferey}
\affiliation{
  \institution{SonarSource}
  \country{Switzerland}
}
\renewcommand{\shortauthors}{E.\ Li, F.\ Stutz, T.\ Wies, D.\ Zufferey}
\begin{abstract}
We study the implementability problem for an expressive class of symbolic communication protocols involving multiple participants. 
Our symbolic protocols describe infinite states and data values using dependent refinement predicates. 
Implementability asks whether a global protocol specification admits a distributed, asynchronous implementation, namely one for each participant, that is deadlock-free and exhibits the same behavior as the specification. 
We provide a unified explanation of seemingly disparate sources of non-implementability through a precise semantic characterization of implementability for infinite protocols. 
Our characterization reduces the problem of implementability to (co)reachability in the global protocol restricted to each participant. 
This compositional reduction yields the first sound and relatively complete algorithm for checking implementability of symbolic protocols. 
We use our characterization to show that for finite protocols, implementability is co-NP-complete for explicit representations and PSPACE-complete for symbolic representations. 
The finite, explicit fragment subsumes a previously studied fragment of multiparty session types for which our characterization yields a co-NP decision procedure, tightening a prior PSPACE upper bound.
 \end{abstract}
        \maketitle
	%
\section{Introduction}
\label{sec:intro}
Concurrency is ubiquitous in modern computing, message-passing is a major concurrency paradigm, and communication protocols are therefore a key target for formal verification. 
Communication protocols specify distributed, message-passing behaviors from a global point of view, altogether describing the interactions between all participants in the protocol. 
Implementability and synthesis are two central questions to the verification of communication protocols. 
Implementability asks whether a protocol admits a distributed implementation, and synthesis in turn computes an admissible one. 
A distributed implementation is considered admissible if it is deadlock-free and exhibits exactly the same communication behaviors described by the specification. 
We refer to the latter property as \emph{protocol fidelity}. 
The implementability question precedes the synthesis question in importance: synthesizing implementations for unrealizable protocols is a fruitless endeavor. 

Global protocol specifications find industry applications in the form of UML's high-level message sequence charts and the Web Service Choreography Description Language, and are widely studied in academia in the form of multiparty session types and choreographic programming. 
Multiparty session types (MSTs) have been implemented in at least 16 programming languages 
including Python \cite{DBLP:journals/fmsd/DemangeonHHNY15, DBLP:journals/corr/NeykovaY16, DBLP:journals/fac/NeykovaBY17}, Java \cite{DBLP:conf/fase/HuY16, DBLP:conf/fase/HuY17}, C \cite{DBLP:conf/tools/NgYH12}, Go \cite{DBLP:conf/icse/LangeNTY18, DBLP:journals/pacmpl/CastroHJNY19}, Scala \cite{DBLP:conf/ecoop/Castro-PerezY23}, Rust \cite{DBLP:conf/ppopp/CutnerYV22,DBLP:journals/darts/LagaillardieNY22}, OCaml \cite{DBLP:conf/ecoop/ImaiNYY19}, F\# \cite{DBLP:conf/cc/NeykovaHYA18}, 
and applied to 
    operating systems~\cite{DBLP:conf/eurosys/FahndrichAHHHLL06},
    high performance computing~\cite{DBLP:conf/pvm/HondaMMNVY12, DBLP:conf/fpl/NiuNYWYL16,demuijnckhughes_et_al:LIPIcs.ECOOP.2019.6},
    cyber-physical systems~\cite{DBLP:conf/ecoop/MajumdarPYZ19, DBLP:journals/pacmpl/MajumdarYZ20}, and
    web services~\cite{DBLP:conf/tgc/YoshidaHNN13}.
We refer the reader to  \cite{DBLP:books/sp/24/Yoshida24} and \cite{M23} for a comprehensive survey of MST and choreography applicability respectively.

To model real-world verification targets, we desire for our protocol specifications to be as expressive as possible.
Various dimensions of expressivity have been explored in the literature, such as arbitrary message payloads, non-deterministic choice, unrestricted recursion and parametricity. 
Formalisms such as choreography automata~\cite{DBLP:conf/ecoop/GheriLSTY22}, high-level message sequence charts~\cite{DBLP:conf/sdl/MauwR97,
	DBLP:conf/ac/GenestMP03,
	DBLP:conf/acsd/GenestM05,
	DBLP:conf/concur/GazagnaireGHTY07,
	DBLP:journals/tosem/RoychoudhuryGS12, 
	DBLP:journals/tse/AlurEY03, 
	DBLP:journals/tcs/Lohrey03, 
	DBLP:conf/concur/AlurY99,DBLP:conf/mfcs/MuschollP99, 
	DBLP:conf/stacs/Morin02,DBLP:journals/jcss/GenestMSZ06} and session types~\cite{DBLP:conf/popl/HondaYC08, DBLP:conf/concur/BocchiHTY10, DBLP:conf/tgc/BocchiDY12, DBLP:journals/jlp/ToninhoY17, DBLP:journals/pacmpl/00020HNY20, DBLP:conf/cav/LiSWZ23} correspond to syntactically-defined fragments that incorporate a selection of these features.
	
In this paper, we study the implementability problem for a semantically-defined class of communication protocols, which we call \emph{global communicating labeled transition systems (GCLTS)}.
GCLTS impose only modest syntactic restrictions and subsume many existing fragments of asynchronous multiparty session types and choreography automata. 
GCLTS capture the following important features: 
\begin{itemize}
	\item Asynchrony: the semantics are interpreted over a peer-to-peer, asynchronous network, with FIFO channels connecting each pair of protocol participants. 
	\item Generalized sender-driven choice: the only notable syntactic restriction imposed by our formalism is that at each branching point of the protocol's control flow, a single participant chooses a branch. 
	In other words, the first message that is sent in each branch of a choice must come from the same sender. However, we impose no restrictions on the recipient or the message payload other than that no two branches share the same recipient and message.
	\item Infinite protocol state: protocol states contain registers that take values from an infinite domain. This allows loops to carry memory across iterations, and allows the protocol to be specified in terms of dependent refinement predicates. 
	\item Infinite message payloads: messages can carry values drawn from an infinite data domain. 
\end{itemize}
Implementability is undecidable for this general class of protocols.
The presence of and interaction between the aforementioned features means that even soundly approximating implementability is challenging. 
Existing work is either comparable in expressivity but does not solve the implementability problem, or solves the implementability problem but is incomparably restricted in expressivity. 
\citet{DBLP:journals/pacmpl/00020HNY20} present a framework for synchronous, refined multiparty session types that soundly approximates implementability through its endpoint projection, but that may yield local specifications that are not implementable.
Several works \cite{DBLP:journals/tcs/AlurEY05, DBLP:journals/tcs/Lohrey03, DBLP:conf/ecoop/Stutz23, DBLP:conf/cav/LiSWZ23} precisely characterize implementability for finite protocol specifications. However, the implementability check in \cite{DBLP:journals/tcs/AlurEY05, DBLP:conf/cav/LiSWZ23} relies on synthesizing an implementation upfront, which is not possible for infinite state protocols.
\citet{DBLP:conf/concur/DasP20}~study local session types with arithmetic refinements in a binary~setting.

We address these challenges by decomposing the implementability problem into two steps. 
First, we give a precise, semantic characterization of implementability for GCLTS that we prove sound and complete once and for all.
Our characterization is defined directly on the global specification, and thus forgoes the need to first synthesize a candidate implementation. 
Moreover, our characterization gives a unified semantic explanation to disparate causes of non-implementability that arise from the expressivity of our protocol fragment.
We encapsulate the complexities introduced by communication-specific features such as asynchrony and partial information in the first step. 
Our semantic characterization reduces implementability to (co)reachability in the GCLTS. 
Specifically, we provide a sound and complete reduction to the first-order fixpoint logic $\mu$CLP~\cite{DBLP:journals/pacmpl/UnnoTGK23}. 
The $\mu$CLP calculus can express recursive predicates with least and greatest fixpoint semantics where the predicate body is constrained by a first-order logic formula over a background theory. 
Our implementability characterization can therefore be checked by existing $\mu$CLP solvers. 
Second, we use this reduction to obtain a
blueprint for solving implementability algorithmically.
Our reduction yields algorithms that are sound and complete relative to an assumed oracle for solving $\mu$CLP validity, in addition to decision procedures with optimal complexity for various decidable classes. 
\clearpage
\paragraph{Contributions.}
In summary, our contributions are: 
\begin{itemize}
	\item Global communicating labeled transition systems (GCLTS): a semantically-defined class of asynchronous communication protocols that subsumes most formalisms in the literature. 
	\item A precise characterization of implementability for GCLTS. 
	\item The first symbolic algorithm for checking implementability of infinite, symbolic protocols that is sound and relatively complete. 
	\item Optimal decision procedures for checking implementability of finite protocols. In particular, we show that for explicit protocol representations that enumerate all states and transitions, the problem is co-NP-complete, and for symbolic protocol representations that encode states and transitions using predicates and variables, the problem is PSPACE-complete. 
	        \item As a corollary of the previous result, we obtain a co-NP decision procedure for implementability of global types, tightening a prior PSPACE upper bound~\cite{DBLP:conf/cav/LiSWZ23,DBLP:journals/corr/abs-2305-17079}. 
\end{itemize}
\section{Overview} 
\label{sec:overview}
We motivate our work using an infinite state version of a two-bidder protocol, depicted as a high-level message sequence chart (HMSC) in \cref{fig:two-bidder-protocol}. The protocol specifies the behavior of two bidders, $\buyerOne$ and $\buyerTwo$, who negotiate to split the purchase of a book from seller $\seller$.

The protocol begins with $\buyerOne$ announcing to $\seller$ and $\buyerTwo$ the book $y$ it proposes to buy. The protocol requires that $y$ signifies a valid ISBN number, which we abstract with the predicate $\isbn(y)$. The seller $\seller$ then informs $\buyerOne$ the requested book's price $z$. After this, $\buyerOne$ and $\buyerTwo$ enter a bidding phase in which they negotiate the split of their respective contributions $b_1$ and $b_2$ towards the purchase. In each round of the bidding phase, $\buyerOne$ proposes its contribution $b_1$ to $\buyerTwo$. Bidder $\buyerTwo$ then decides to either abort the protocol by sending a $\quitMsg$ message to $\seller$, or respond to $\buyerOne$ with its own bid $b_2$. In case $\buyerTwo$ aborts, $\seller$ echoes the abort message to $\buyerOne$ and the protocol terminates. In case $\buyerTwo$ continues bidding, if the sum of the proposed bids exceeds the book's price,  $\buyerOne$ informs $\seller$ of the successful negotiation. Seller $\seller$ in turn relays the message to $\buyerTwo$. Otherwise, $\buyerOne$ sends a $\contMsg$ message to $\buyerTwo$, informing them that they need to enter another bidding round. Throughout the bidding phase, $\buyerOne$ and $\buyerTwo$ track the values of their latest bids in the registers $z_1$ and $z_2$. The refinements ensure that the proposed bids are strictly increasing from one round to the next, thus enforcing termination.
\begin{figure}[b]
  \begin{minipage}[b]{.59\textwidth}
	\centering
	
	\resizebox{1.01\textwidth}{!}{
\begin{tikzpicture}[hmsc,baseline,msg/.append style={tight,above=2pt,pos=.2},node distance=15ex]
\begin{scope}[msc=s1,init]
  \begin{scope}[participant=s]
    \node[head];
    \node[event] {};
    \node[no event] {};
    \node[event] {};
  \end{scope}
  \node[yshift=.7pt] at (s-1) {$\seller$};
  \begin{scope}[participant=ba]
    \node[head];
    \node[event] {};
    \node[event] {};
    \node[event] {};
  \end{scope}
  \node at (ba-1) {$\buyerOne$};
  \begin{scope}[participant=bb]
    \node[head];
    \node[no event] {};
    \node[event] {};
    \node[no event] {};
  \end{scope}
  \node at (bb-1) {$\buyerTwo$};
  \draw[messages]
    (ba-2) edge node[msg,midway]{$y \cond{\isbn(y)}$} (s-2)
    (ba-3) edge node[msg,midway]{$y \cond{\isbn(y)}$} (bb-3)
    (s-4) edge node[msg,midway]{$z \cond{z > 0}$} (ba-4)
  ;
\end{scope}
\begin{scope}[msc=s2, yshift=-16ex, xshift=13ex]
  \begin{scope}[participant=s]
    \node[head];
    \node[no event] {};
  \end{scope}
  \node[yshift=.7pt] at (s-1) {$\seller$};
  \begin{scope}[participant=ba]
    \node[head];
    \node[event] {};
  \end{scope}
  \node at (ba-1) {$\buyerOne$};
  \begin{scope}[participant=bb]
    \node[head];
    \node[event] {};
  \end{scope}
  \node at (bb-1) {$\buyerTwo$};
  \draw[messages]
    (ba-2) edge node[msg,midway]{$b_1 \cond{b_1 > z_1}$} (bb-2)
  ;
\end{scope}
\begin{scope}[msc=s3, yshift=-27ex, final]
  \begin{scope}[participant=s]
    \node[head];
    \node[event] {};
    \node[event] {};
  \end{scope}
  \node[yshift=.7pt] at (s-1) {$\seller$};
  \begin{scope}[participant=ba]
    \node[head];
    \node[no event] {};
    \node[event] {};
  \end{scope}
  \node at (ba-1) {$\buyerOne$};
  \begin{scope}[participant=bb]
    \node[head,xshift=-5ex];
    \node[event] {};
    \node[no event] {};
  \end{scope}
  \node at (bb-1) {$\buyerTwo$};
  \draw[messages]
    (bb-2) edge node[msg,pos=0.2]{$\quitMsg$} (s-2)
    (s-3) edge node[msg,pos=0.51]{$\quitMsg$} (ba-3)
  ;
\end{scope}
\begin{scope}[msc=s4, yshift=-40ex, final]
  \begin{scope}[participant=s]
    \node[head];
    \node[event] {};
    \node[event] {};
  \end{scope}
  \node[yshift=.7pt] at (s-1) {$\seller$};
  \begin{scope}[participant=ba]
    \node[head];
    \node[event] {};
    \node[no event] {};
  \end{scope}
  \node at (ba-1) {$\buyerOne$};
  \begin{scope}[participant=bb]
    \node[head,xshift=-5ex];
    \node[no event] {};
    \node[event] {};
  \end{scope}
  \node at (bb-1) {$\buyerTwo$};
  \draw[messages]
    (ba-2) edge node[msg,midway]{$\succMsg \cond{b_1 + b_2 \geq z}$} (s-2)
    (s-3) edge node[msg,pos=0.8]{$\succMsg$} (bb-3)
  ;
\end{scope}
\begin{scope}[msc=s5, yshift=-27ex, xshift=31ex]
  \begin{scope}[participant=s]
    \node[head];
    \node[no event] {};
  \end{scope}
  \node[yshift=.7pt] at (s-1) {$\seller$};
  \begin{scope}[participant=ba]
    \node[head,xshift=-5ex];
    \node[event] {};
  \end{scope}
  \node at (ba-1) {$\buyerOne$};
  \begin{scope}[participant=bb]
    \node[head];
    \node[event] {};
  \end{scope}
  \node at (bb-1) {$\buyerTwo$};
  \draw[messages]
    (bb-2) edge node[msg,midway]{$b_2 \cond{b_2 > z_2}$} (ba-2)
  ;
\end{scope}
\begin{scope}[msc=s6, yshift=-40ex, xshift=31ex]
  \begin{scope}[participant=s]
    \node[head];
    \node[no event] {};
  \end{scope}
  \node[yshift=.7pt] at (s-1) {$\seller$};
  \begin{scope}[participant=ba]
    \node[head,xshift=-5ex];
    \node[event] {};
  \end{scope}
  \node at (ba-1) {$\buyerOne$};
  \begin{scope}[participant=bb]
    \node[head];
    \node[event] {};
  \end{scope}
  \node at (bb-1) {$\buyerTwo$};
  \draw[messages]
    (ba-2) edge node[msg,midway]{$\contMsg \cond{b_1 + b_2 < z}$} (bb-2)
  ;
\end{scope}
\draw[trans]
  (s1) edge node[left,xshift=-1ex]{\scriptsize$\initUpdLoop{z_1, z_2}{0, 0}$} (s2.north)
  (s2) edge (s3)
  (s5) edge (s4)
  (s2) edge (s5)
  (s5) edge (s6)
  ;
  \draw[trans]
  (s6.south) -- node {} ++(0,-0.15) |- ++(2.40,0) |- 
  node[yshift=1.5ex,xshift=-7ex]{\scriptsize$\initUpdLoop{z_1,z_2}{b_1,b_2}$} (s2.east)
  ;
\end{tikzpicture} 	}
	\caption{Two-bidder protocol.}
		\label{fig:two-bidder-protocol}
  \end{minipage}
  \hfill
  \begin{minipage}[b]{.4\textwidth}
	\centering
	\resizebox{.98\textwidth}{!}{
\begin{tikzpicture}[sem, node distance=2em and 1em]
    \node[state, initial above, initial text = $\seller$] (q0) {$q_{0, \seller}$};
    \node[state, below = of q0] (q1) {$q_{1, \seller}$};
    \node[state, below = of q1] (q2) {$q_{2, \seller}$};
    \node[state, below left = of q2] (q3) {$q_{3, \seller}$};
    \node[state, accepting, below = of q3] (q4) {$q_{4, \seller}$};
    \node[state, below right = of q2] (q5) {$q_{5, \seller}$};
    \node[state, accepting, below = of q5] (q6) {$q_{6, \seller}$};
    
    \path (q0) edge node[right] {$\rcv{\buyerOne}{\seller}{y \cond{\isbn(y)}}$} (q1);
    \path (q1) edge node[right] {$\snd{\seller}{\buyerOne}{z \cond{z > 0}}$} (q2);
    \path (q2) edge node[left] {$\rcv{\buyerTwo}{\seller}{\quitMsg}$} (q3);
    \path (q3) edge node[left] {$\snd{\seller}{\buyerOne}{\quitMsg}$} (q4);
    \path (q2) edge node[right] {$\rcv{\buyerOne}{\seller}{\succMsg \cond{b_1 + b_2 \geq z}}$} (q5);
    \path (q5) edge node[right] {$\snd{\seller}{\buyerTwo}{\succMsg}$} (q6);
\end{tikzpicture}
 	}
	\caption{State machine for seller $\seller$ for
		\cref{fig:two-bidder-protocol}.}
   \label{fig:two-bidder-implementation-1}
  \end{minipage}
\end{figure}
\begin{figure}[t]
  \begin{minipage}[b]{.52\textwidth}
	\centering
	\resizebox{0.88\textwidth}{!}{
\begin{tikzpicture}[sem, node distance=1.7em and .8em]
    \node[state, initial above, initial text = $\buyerOne$] (q0) {$q_{0, \buyerOne}$};
    \node[state, right = 6.5em of q0] (q1) {$q_{1, \buyerOne}$};
    \node[state, below = of q1] (q2) {$q_{2, \buyerOne}$};
    \node[state, left = 6.5em of q2] (q3) {$q_{3, \buyerOne}$};
    \node[state, below = of q3] (q4) {$q_{4, \buyerOne}$};
    \node[state, below = of q4] (q5) {$q_{5, \buyerOne}$};
    \node[state, accepting, below left = of q5] (q6) {$q_{6, \buyerOne}$};
    \node[state, below right = of q5] (q7) {$q_{7, \buyerOne}$};
    \node[state, accepting, below = 2.1em of q7] (q8) {$q_{8, \buyerOne}$};
    \node[state, below = 2.1em of q7, xshift=3em] (q9) {$q_{9, \buyerOne}$};
    
    \path (q0) edge node[above,xshift=0.5em] {$\snd{\buyerOne}{\seller}{y \cond{\isbn(y)}}$} (q1);
    \path (q1) edge node[right] {$\snd{\buyerOne}{\buyerTwo}{y \cond{\isbn(y)}}$} (q2);
    \path (q2) edge node[above,xshift=-.3em] {$\rcv{\seller}{\buyerOne}{z \cond{z > 0}}$} (q3);
    \path (q3) edge node[right] {$\initUpdLoop{z_1, z_2}{0, 0}$} (q4);
    \path (q4) edge node[right] {$\snd{\buyerOne}{\buyerTwo}{b_1 \cond{b_1 > z_1}}$} (q5);
    \path (q5) edge node[left] {$\rcv{\seller}{\buyerOne}{\quitMsg}$} (q6);
    \path (q5) edge node[right] {$\rcv{\buyerTwo}{\buyerOne}{b_2 \cond{b_2 > z_2}}$} (q7);
    \path (q7) edge node[left] {$\snd{\buyerOne}{\buyerTwo}{\succMsg \cond{b_1 + b_2 \geq z}}$} (q8);
    \path (q7) edge node[right] {$\snd{\buyerOne}{\buyerTwo}{\contMsg \cond{b_1 + b_2 < z}}$} (q9);
    \draw[semarrow] (q9.south) -- node {} ++(0,-0.15) |- ++(2.3cm,0) |- (q4) node[above,xshift=3cm] {$\initUpdLoop{z_1,z_2}{b_1,b_2}$} node[] {};
\end{tikzpicture}
 	}
	\caption{State machine for bidder $\buyerOne$ for
		\cref{fig:two-bidder-protocol}.}
   \label{fig:two-bidder-implementation-2}
  \end{minipage}
  \begin{minipage}[b]{.47\textwidth}
	\centering
	\resizebox{.80\textwidth}{!}{
\begin{tikzpicture}[sem, node distance=1.7em and .8em]
    \node[state, initial above, initial text = $\buyerTwo$] (q0) {$q_{0, \buyerTwo}$};
    \node[state, below = of q0] (q1) {$q_{1, \buyerTwo}$};
    \node[state, below = of q1] (q2) {$q_{2, \buyerTwo}$};
    \node[state, below = of q2] (q3) {$q_{3, \buyerTwo}$};
    \node[state, accepting, below left = of q3] (q4) {$q_{4, \buyerTwo}$};
    \node[state, below right = of q3] (q5) {$q_{5, \buyerTwo}$};
    \node[state, accepting, below = 2.1em of q5] (q6) {$q_{6, \buyerTwo}$};
    \node[state, below = 2.1em of q5, xshift=3em] (q7) {$q_{7, \buyerTwo}$};
    
    \path (q0) edge node[right] {$\rcv{\buyerOne}{\buyerTwo}{y \cond{\isbn(y)}}$} (q1);
    \path (q1) edge node[right] {$\initUpdLoop{z_1, z_2}{0, 0}$} (q2);
    \path (q2) edge node[right] {$\rcv{\buyerOne}{\buyerTwo}{b_1 \cond{b_1 > z_1}}$} (q3);
    \path (q3) edge node[left] {$\snd{\buyerTwo}{\seller}{\quitMsg}$} (q4);
    \path (q3) edge node[right] {$\snd{\buyerTwo}{\buyerOne}{b_2 \cond{b_2 > z_2}}$} (q5);
    \path (q5) edge node[left] {$\rcv{\seller}{\buyerTwo}{\succMsg}$} (q6);
    \path (q5) edge node[right] {$\rcv{\buyerOne}{\buyerTwo}{\contMsg \cond{b_1 + b_2 < z}}$} (q7);
    \draw[semarrow] (q7.south) -- node {} ++(0,-0.15) |- ++(2.3cm,0) |- (q2) node[above,xshift=3.2cm] {$\initUpdLoop{z_1,z_2}{b_1,b_2}$} node[] {};
\end{tikzpicture}
 	}
	\caption{State machine for bidder $\buyerTwo$ for
		\cref{fig:two-bidder-protocol}.}
   \label{fig:two-bidder-implementation-3}
  \end{minipage}
\end{figure}

\cref{fig:two-bidder-implementation-1,fig:two-bidder-implementation-2,fig:two-bidder-implementation-3} show an admissible implementation for the two-bidder protocol in \cref{fig:two-bidder-protocol}, consisting of a local implementation for each participant: $\seller$, $\buyerOne$ and $\buyerTwo$. 
The transition labels specify their local behaviors: 
$\snd{\buyerOne}{\seller}{y \cond{\isbn(y)}}$
specifies that $\buyerOne$ sends a message $y$ to $\seller$ such that $y$ satisfies $\isbn(y)$, i.e.\ $y$ is a valid ISBN number;
$\rcv{\buyerOne}{\seller}{y \cond{\isbn(y)}}$
specifies that $\seller$ receives $y$ from $\buyerOne$, and can assume $\isbn(y)$ holds of $y$.
We assume an asynchronous setting in which every pair of participants is connected by a FIFO channel. 
The implementability of \cref{fig:two-bidder-protocol} is witnessed by 
\cref{fig:two-bidder-implementation-1,fig:two-bidder-implementation-2,fig:two-bidder-implementation-3}, which 
together exhibit the same behaviors as the global protocol and is never stuck. 

To see that the implementability problem is non-trivial, consider a variant of the protocol in \cref{fig:two-bidder-protocol} where the $\succMsg$ message to $\seller$ is sent by $\buyerTwo$ instead of $\buyerOne$.
The resulting protocol is no longer implementable because $\buyerTwo$ never learns about the price $z$ of the book $y$ and is therefore unable to determine when the negotiation with $\buyerOne$ has succeeded.

Our example highlights several important expressive features of GCLTS: 
\begin{itemize}
\item Generalized sender-driven choice: after $\buyerTwo$ receives a bid from $\buyerOne$, it has the option to either send a bid back to $\buyerOne$ and continue the bidding process, or terminate the protocol by sending a $\quitMsg$ message to the bookseller, who then relays the termination message to the first bidder. Due to this choice interaction alone, the protocol is not expressible in \cite{DBLP:journals/pacmpl/00020HNY20}.
\item Infinite state: the protocol state contains registers that can be assigned values from an infinite domain. Registers are updated to store the last bid from each round $z_1$ and $z_2$, and to enforce that bidders make strictly increasing bids per round. 
\item Infinite message data: message payload values can be drawn from an infinite data domain, such as the book price $z$ and bids $b_1$ and $b_2$.
\item Dependent refinement predicates: message payloads are constrained by data refinements such as $z_1 < b_1$ and $z < b_1 + b_2$. The refinement predicates can refer to current register values in addition to data values sent in prior messages.
\item Partial information: each protocol participant only has a partial view of the global protocol state. For example, even though $\seller$ participates in the bidding phase of the protocol, it never learns about the bids $b_1$ and $b_2$ in each bidding round. In fact, the registers $z_1$ and $z_2$ that store the last bid are known only to the bidders.
\end{itemize}

The presence of these features in the class of communication protocols we consider makes checking implementability uniquely challenging.
For protocols with finite GCLTS specifications, deciding implementability in the presence of asynchrony and non-deterministic choice already presents a challenge. 
Note that finiteness here refers only to the specification, and does not mean that the underlying protocol is finite-state, nor that it contains only finite traces. 
Most existing work has therefore focused on developing projection operators that are sound but incomplete~\cite{DBLP:conf/popl/HondaYC08,DBLP:conf/sfm/CoppoDPY15,DBLP:journals/jlp/ToninhoY17,DBLP:conf/ecoop/ScalasDHY17}. 
These projection operators solve implementability and synthesis simultaneously by computing a candidate implementation, but often fail eagerly for protocols for which an implementation exists. 
\citet{DBLP:conf/cav/LiSWZ23} proposed the first sound and complete projection operator for finite, asynchronous, multiparty session types. 
The projection operator critically relies on the observation that if a global type is implementable, then a canonical implementation implements it. 
Thus, the implementability problem reduces to checking whether this canonical implementation indeed implements the global type, i.e.\ it recognizes the same set of behaviors and is deadlock-free.
Towards these ends, \citet{DBLP:conf/cav/LiSWZ23}~identify sound and complete conditions, referred to as \emph{Send Validity} and \emph{Receive Validity}, that are checked on the states of the canonical implementation. 
\begin{figure}[b]
\begin{tikzpicture}[sem] 
	\node[state, initial left=, initial text =](q0){};
	
	\node[state, right=0.4cm and 1.5cm of q0](q1){};
	\node[state, right=2.5cm of q1](q2){};
	
	\node[state, above right=0.2cm and 2.5cm of q2](q2l){};
	\node[state, below right=0.2cm and 2.5cm of q2](q2r){};
	\node[state, right=1.3cm of q2l](q3l){};
	\node[state, right=1.3cm of q3l](q4l){};
	\node[state, accepting, right=2.7cm of q4l](q5l){};
	\node[state, right=1.3cm of q2r](q3r){};
	\node[state, right=1.3cm of q3r](q4r){};
	\node[state, accepting, right=2.7cm of q4r](q5r){};
	\path(q0) edge node[sloped, pos=0.5, above] {$\msgFromTo{\procD}{\procA}{x \cond{\top}}$} (q1);
	\path(q1) edge node[sloped, pos=0.5, above] {\textcolor{gray}{\textcircled{a}} $\msgFromTo{\procA}{\procC}{y \cond{\textcolor{green}{y > x}}}$} (q2);
    \path(q1) edge node[sloped, pos=0.5, below] {\textcolor{gray}{\textcircled{\raisebox{-.8pt}{b}}} $\msgFromTo{\procA}{\procC}{y \cond{\textcolor{green}{y = x+2}}}$} (q2);
	\path(q2) edge node[sloped, pos=0.5, above] {$\msgFromTo{\procA}{\procB}{\lblmsgB \cond{even(x)}}$} (q2l);
	\path(q2) edge node[sloped, pos=0.5, below] {$\msgFromTo{\procA}{\procB}{\lblmsgM \cond{odd(x)}}$} (q2r);
	\path(q2l) edge node[sloped, pos=0.5, above] {$\msgFromTo{\procA}{\procC}{z_1}$} (q3l);
	\path(q3l) edge node[sloped, pos=0.5, above] {$\msgFromTo{\procB}{\procC}{z_2}$} (q4l);
	\path(q4l) edge node[sloped, pos=0.5, above] {$\msgFromTo{\procC}{\procA}{z \cond{z = z_1 - z_2}}$} (q5l);
	\path(q2r) edge node[sloped, pos=0.5, below] {$\msgFromTo{\procB}{\procC}{z_2}$} (q3r);
	\path(q3r) edge node[sloped, pos=0.5, below] {$\msgFromTo{\procA}{\procC}{z_1}$} (q4r);
	\path(q4r) edge node[sloped, pos=0.5, below] {$\msgFromTo{\procC}{\procA}{z \cond{z = z_2 - z_1}}$} (q5r);
\end{tikzpicture}
 	\vspace{-3ex}
	\caption{
	Two protocols:
	$\SymProt_1$ using \textcolor{gray}{\textcircled{a}} with receive order violation
	$\SymProt_1'$ using \textcolor{gray}{\textcircled{\raisebox{-.8pt}{b}}} without receive order violation.
	\label{fig:receive-violation-and-ok}
	}
\end{figure}

In the presence of dependent refinement predicates, checking these conditions is not straightforward.
Consider the examples $\SymProt_1$ (using \textcolor{gray}{\textcircled{a}}) and $\SymProt_1'$ (using \textcolor{gray}{\textcircled{\raisebox{-1pt}{b}}}) in \cref{fig:receive-violation-and-ok}, which are variations of the examples for Receive Validity~\cite{DBLP:conf/cav/LiSWZ23} featuring dependent predicates. 
A transition label $\msgFromTo{\procA}{\procC}{y \cond{y>x}}$, which is \textcolor{gray}{\textcircled{a}} for $\SymProt_1$, atomically specifies the send event by $\procA$ and the corresponding receive event by $\procC$, along with the constraint that $y$ satisfies $y>x$.
In $\SymProt_1$, participant $\procA$ chooses a branch without explicitly informing $\procC$ of their choice. 
In both branches, $\procC$ is required to subtract the second value that is sent from the first value that is sent, and send the result back to $\procA$. 
However, due to asynchrony, both messages can arrive in $\procC$'s message channels simultaneously, and $\procC$ cannot tell which value was sent first. 
Therefore, $\procC$ may subtract the values in the wrong order, rendering the protocol unimplementable. 

\citet{DBLP:conf/cav/LiSWZ23} propose one method of protocol repair: introducing a message sent by $\procC$ on each branch that creates a causal dependency between the messages from $\procA$ and $\procB$, so that $\procC$ can no longer receive them in either order. 
The incorporation of dependent refinements enables a new method of protocol repair; one that does not change the communication events among the participants. 
The newly repaired protocol is depicted in $\SymProt_1'$, in which 
the predicate on the second transition is changed from $y > x$ to $y = x+2$.
Despite the fact that $\procC$ is  still not informed of $\procA$'s choice, $\procC$ can \emph{infer} $\procA$'s choice through the parity of the first value it received from $\procA$ and thus correctly follow the protocol: if $y$ is even, $\procC$ receives from $\procA$ first, and if $y$ is odd, $\procC$ receives from $\procB$ first. 
\begin{figure}[t]
\begin{tikzpicture}[sem] 
    \node[state, initial left=, initial text =](q0){};
    
    \node[state, above right=0.3cm and 1.5cm of q0](q1){};
    \node[state, right=2.5cm of q1](q2){$q_1$};
    \node[state, right=1.5cm of q2](q3){};
    \node[state, accepting, right=1.5cm of q3](q4){};
    
    \node[state, below right=0.3cm and 1.5cm of q0](q5){};
    \node[state, right= 2.5cm of q5](q6){$q_2$};
    \node[state, accepting, above right=0.07cm and 2.5cm of q6](q6l){};
    \node[state, below right=0.07cm and 2.5cm of q6](q6r){};
    \node[state, accepting, right= 1.5cm of q6r](q7){};
    
    \path(q0) edge node[sloped, pos=0.5, above] {$\msgFromTo{\procD}{\procB}{\lblmsgB}$} (q1);
    \path(q1) edge node[sloped, pos=0.5, above] {$\msgFromTo{\procB}{\procA}{x_1 \cond{x_1 = 4}}$} (q2);
    \path(q2) edge node[sloped, pos=0.5, above] {$\msgFromTo{\procA}{\procB}{\lblmsgO}$} (q3);
    \path(q3) edge node[sloped, pos=0.5, above] {$\msgFromTo{\procB}{\procC}{\lblmsgM}$} (q4);
    
    \path(q0) edge node[sloped, pos=0.5, below] {$\msgFromTo{\procD}{\procB}{\lblmsgM}$} (q5);
    \path(q5) edge node[below] {$\msgFromTo{\procB}{\procA}{x_2 \cond{\top}}$} (q6);
    \path(q6) edge node[sloped, pos=0.5, above] {$\msgFromTo{\procB}{\procC}{\lblmsgB \cond{x_2 \neq 5}}$} (q6l);
    \path(q6) edge node[sloped, pos=0.5, below] {$\msgFromTo{\procB}{\procC}{\lblmsgM \cond{x_2 = 5}}$} (q6r);
    \path(q6r) edge node[below] {$\msgFromTo{\procA}{\procB}{\lblmsgO}$} (q7);
\end{tikzpicture}
 	\vspace{-2ex}
	\caption{
	$\SymProt_2$: An protocol with a send violation.
	\label{fig:send-violation}
	}
\end{figure}

We now turn our attention to send violations.
In the protocol shown in \cref{fig:send-violation}, $\procD$~chooses a branch and communicates its choice to $\procB$.
Participant $\procA$ is again not explicitly informed of the choice: in fact, $\procA$ can receive $4$ from $\procB$ on both branches.
At first glance, it appears as though it is safe for~$\procA$ to send $\lblmsgO$ to $\procB$ upon receiving 4 from $\procB$, because whilst $\procA$ cannot distinguish the two branches, both branches contain the transition $\msgFromTo{\procA}{\procB}{\lblmsgO}$.
Upon closer inspection, the predicate guarding the transition immediately preceding $\procA \xrightarrow{} \procB: \lblmsgO$ on the lower branch, $x_2 = 5$, is only satisfied when $\procA$ receives $5$ from $\procB$.
When $\procA$ receives $4$, the lower branch from $q_2$ is disabled, and since the upper branch from $q_2$ does not contain the transition $\msgFromTo{\procA}{\procB}{\lblmsgO}$, the protocol is not implementable.

The examples above exemplify the ways in which refinement predicates complicate implementability checking for symbolic protocols. We return to these examples, in addition to some others, in greater detail in \cref{sec:characterization} when we present our precise characterization of implementability. 
We structure the rest of the paper as follows. 
\S\ref{sec:prelim} presents relevant preliminary definitions and defines the class of communication protocols we consider. 
\Cref{sec:characterization} presents our semantic characterization of implementability for GCLTS in terms of (co)reachability, and proves that it is precise. 
\cref{sec:symbolic} describes our sound and complete reduction from the characterization in \cref{sec:characterization} to logical formulas in $\mu$CLP~\cite{DBLP:journals/pacmpl/UnnoTGK23}, and additionally presents improved complexity results under certain finiteness assumptions on the GCLTS.
\cref{sec:related-work} discusses related work and concludes. 
\section{Preliminaries}
\label{sec:prelim}
We introduce some basic concepts and notation before defining our class of protocols. 
\paragraph{Words.}
Let $\AlphAsync$ be an alphabet.
$\AlphAsync^*$ denotes the set of finite words over $\AlphAsync$, $\AlphAsync^\omega$ the set of infinite words, and $\AlphAsync^\infty$ their union $\AlphAsync^* \cup \AlphAsync^\omega$.
A word $u \in \AlphAsync^*$ is a \emph{prefix} of word $v \in \AlphAsync^\infty$, denoted $u \leq v$, if there exists $w \in \AlphAsync^\infty$ with $u \cdot w = v$;
we denote all prefixes of $u$ with $\pref(u)$. 
Given a word $w = w_0 \ldots w_n$, we use $w[i]$ to denote the i-th symbol $w_i \in \AlphAsync$, and $w[0..i]$ to denote the subword between and including $w_0$ and $w_i$, i.e.\ $w_0 \ldots w_i$.
\paragraph{Message Alphabets.}
Let $\Procs$ be a finite set of participants and $\MsgVals$ be a (possibly infinite) data domain. 
We define the set of \emph{synchronous events} $\AlphSyncSubscript \is \set{ \msgFromTo{\procA}{\procB}{\val} \mid \procA,\procB ∈ \Procs \text{ and } \val ∈ \MsgVals}$ where $\msgFromTo{\procA}{\procB}{\val}$ denotes a message exchange of $\val$ from sender $\procA$ to receiver $\procB$.
For a participant $\procA\in \Procs$, we define the alphabet $\AlphSync_\procA = \set{\msgFromTo{\procA}{\procB}{\val} \mid \procB \in \Procs,\; \val \in \MsgVals} \cup \set{\msgFromTo{\procB}{\procA}{\val} \mid \procB \in \Procs,\; \val \in \MsgVals}$, and a homomorphism~$\wproj_{\AlphSync_\procA}$, where $x \wproj_{\AlphSync_\procA} = x$ if $x \in \AlphSync_\procA$ and $\emptystring$ otherwise. 
A~synchronous event is split into a send and receive event for the respective participant, yielding \emph{asynchronous events}.
For a participant $\procA \in\Procs$, we define the alphabet
$\AlphAsync_{\procA,!} = \set{\snd{\procA}{\procB}{\val} \mid \procB \in \Procs,\; \val \in \MsgVals }$ of \emph{send} events
and the alphabet
$\AlphAsync_{\procA,?} = \set{\rcv{\procB}{\procA}{\val} \mid \procB \in \Procs,\; \val \in \MsgVals }$ of \emph{receive} events.
The event $\snd{\procA}{\procB}{\val}$ denotes participant $\procA$ sending a message $\val$ to $\procB$,
and $\rcv{\procB}{\procA}{\val}$ denotes participant $\procA$ receiving a message $\val$ from $\procB$.
We write $\AlphAsync_{\procA} = \AlphAsync_{\procA,!} \union \AlphAsync_{\procA,?}$,
$\AlphAsync_! = \Union_{\procA \in \Procs} \AlphAsync_{\procA,!}$, and
$\AlphAsync_? = \Union_{\procA \in \Procs} \AlphAsync_{\procA,?}$.
Finally, $\AlphAsyncSubscript = \AlphAsync_! \cup \AlphAsync_?$.
We define a homomorphism to map the synchronous alphabet to its asynchronous counterpart, defined as 
$
\SyncToAsync(\msgFromTo{\procA}{\procB}{\val})
\is
\snd{\procA}{\procB}{\val}. \,
\rcv{\procA}{\procB}{\val}
$.
Because $\SyncToAsync$ is injective, there exists a unique inverse, which we denote~$\SyncToAsync^{-1}$.
We say that $\procA$ is \emph{active} in $x \in \AlphAsyncSubscript$ if $x \in \AlphAsync_{\procA}$.
For each participant $\procA\in \Procs$, we define a homomorphism~$\wproj_{\AlphAsync_\procA}$, where $x \wproj_{\AlphAsync_\procA} = x$ if $x \in \AlphAsync_\procA$ and $\emptystring$ otherwise.
We write $\MsgVals(w)$ to project the send and receive events in $w$ onto their messages.
We fix $\Procs$ and~$\MsgVals$ in the remainder of the paper.
\paragraph{Labeled Transition Systems.}
A \emph{labeled transition system} (LTS) is a tuple
$\mathcal{S} = (S, \AlphSync, T, s_0, F)$
where
$S$ is a set of states, $\AlphSync$ is a set of labels,
$T$ is a set of transitions from $S \times \AlphSync \times S$,
$F \subseteq S$ is a set of final states,
and $s_0 \in S$ is the initial state. 
We use $p \xrightarrow{\alpha} q$ to denote the transition $(p, \alpha, q) \in T$.
Runs and traces of an LTS are defined in the expected way.
A run is $\emph{maximal}$ if it is either finite and ends in a final state, or is infinite.
The language of an LTS $\mathcal{S}$, denoted $\lang(\mathcal{S})$, is defined as the set of maximal traces.
A state $s \in S$ is a \emph{deadlock} if it is not final and has no outgoing transitions.
An LTS is \emph{deadlock-free} if no reachable state is a deadlock.
Given an LTS $\mathcal{S} = (S, \AlphSync, T, s_0, F)$ and a state $s \in S$, we use $\mathcal{S}_{s}$ to denote the LTS obtained by replacing $s_0$ with $s$ as the initial state: $(S, \AlphSync, T, s, F)$.
\subsection{Global Communicating Labeled Transition Systems (GCLTS)}
We use LTS over the synchronous alphabet~$\AlphSyncSubscript$ to model communication protocols from a global perspective. 
We impose three more conditions on the class of LTSs we consider: 
that final states do not contain outgoing transitions, that multiple outgoing transitions from a state share a sender, and that the LTS is deadlock-free.
\begin{definition}[Global communicating LTS]
	\label{def:gclts} 
An LTS $\mathcal{S} = (S, \AlphSyncSubscript, T, s_0, F)$ is a \emph{global communicating labeled transition system} (GCLTS) if the following conditions hold:
\begin{enumerate}[(1)]
\item \label{cond-gclts:1} \emph{sink-finality}:
for every final state $s \in F$, there does not exist $l \in \AlphSyncSubscript$ and $s' \in S$ with $s \xrightarrow{l} s' \in T$;
 \item \label{cond-gclts:2} \emph{sender-driven choice}:
 for all states $s, s_1, s_2 \in S$ and $l_1, l_2 \in \AlphSyncSubscript$ such that
 $s \xrightarrow{l_i} s_i \in T$ for $i \in \set{1,2}$,
 there is a participant $\procA \in \Procs$
 who is the sender for both, \ie
 $\SyncToAsync(l_i) \in \AlphAsync_{\procA,!}$ for $i \in \set{1,2}$, and furthermore $l_1 = l_2 \implies s_1 = s_2$;
 \item \label{cond-gclts:3} \emph{deadlock freedom}: $\mathcal{S}$ is deadlock-free.
\end{enumerate}
\end{definition}
Condition \ref{cond-gclts:1} is ubiquitous in the domain of multiparty session types and was also shown to require special treatment in the literature on high-level message sequence charts \cite{DBLP:conf/sefm/DanHC10}. 

Condition \ref{cond-gclts:2} is a generalisation of most multiparty session types fragments, which require not only a dedicated sender but also a dedicated receiver. 
This more restrictive condition is called \emph{directed choice}.
In contrast, \emph{mixed choice} lifts all restrictions on choice, and amounts to only requiring determinism. 
\citet{DBLP:journals/tcs/Lohrey03} showed that 
implementability
is undecidable for HMSCs satisfying determinism and Condition \ref{cond-gclts:3}. 
\citet{DBLP:phd/dnb/Stutz24} showed that implementability remains undecidable for mixed choice global 
multiparty session types satisfying determinism and Conditions \ref{cond-gclts:1} and \ref{cond-gclts:3}. 
Sender-driven choice thus represents a good middle ground, allowing to express interesting communication patterns while retaining decidability of implementability. 

Condition \ref{cond-gclts:3} simply requires that protocols do not specify deadlocking behaviors.

In the remainder of the paper we refer to a GCLTS simply as a \emph{protocol}.
\paragraph{Restricting Protocols to Participants.}
From a protocol $\mathcal{S}$, we can define a local protocol for each participant $\procA$ via domain restriction to $\AlphAsync_\procA$. 
Formally, given a protocol $\mathcal{S} = (S, \AlphSyncSubscript, T, s_0, F)$, 
we define
$\mathcal{S_\procA} \is
(S, \AlphSync_\procA \dunion \set{\emptystring}, T_\procA, s_0, F)$
where
$T_\procA \is
\set{s \xrightarrow{l \wproj_{\AlphSync_\procA}} s'
\mid s \xrightarrow{l} s' \in T}$
for a participant $\procA \in \Procs$.
\paragraph{Asynchronous Protocol Semantics.}
Note that a protocol is specified using the synchronous alphabet~$\AlphSyncSubscript$. 
To define the \emph{asynchronous} semantics of a protocol $\mathcal{S}$ we first map finite and infinite words of $\mathcal{S}$ onto their asynchronous counterpart using $\SyncToAsync$, thus obtaining a set of asynchronous words in which matching send and receive events are adjacent to each other. 
In an asynchronous, FIFO network, two events are independent if they are not related by the \emph{happened-before} relation~\cite{DBLP:journals/cacm/Lamport78}. 
For example, any two send events from distinct senders are independent. 
Consequently, two words are \emph{indistinguishable} if any asynchronous, FIFO implementation that recognizes one word must recognize the other, e.g. 
$w\cdot\snd{\procA}{\procB}{\val}\cdot\snd{\procC}{\procD}{\val'}\cdot u$ 
and 
$w\cdot\snd{\procC}{\procD}{\val'}\cdot\snd{\procA}{\procB}{\val}\cdot u$, where 
$\procA ≠ \procC$.  
We define our protocol semantics as the set of \emph{channel-compliant}~\cite{DBLP:conf/concur/MajumdarMSZ21} words that are closed under this notion of indistinguishability. 
Channel compliance characterizes words that respect FIFO order, i.e. receive events appear after their matching send event, and the order of receive events follows that of send events in each~channel. 
\begin{definition}[Channel compliance] 
	Let $w \in \AlphAsyncSubscript^\infty$. We say that $w$ is \emph{channel-compliant} if for all prefixes $w' \leq w$, for all $\procA\neq \procB \in \Procs$, 
	$\MsgVals(w' \wproj_{\rcv{\procA}{\procB}{\_}}) \leq 
	\MsgVals(w' \wproj_{\snd{\procA}{\procB}{\_}})$.  
\end{definition} 
The asynchronous semantics of a protocol is defined as follows: 
\begin{align*}
\small
\interswaplang(\mathcal{S}) = \; & \{ w' \!\in\! \AlphAsyncSubscript^* \mid \exists w \in \AlphAsyncSubscript^*. \, w \!\in\! \SyncToAsync(\lang(\mathcal{S})) \land w' \text{ is channel-compliant } 
\\ & \hspace{2.5cm}
\land \forall\procA \!\in\! \Procs.~w' \wproj_{\AlphAsync_\procA} \!=\! w \wproj_{\AlphAsync_\procA}\} \\
	\cup {}\; & \{ w' \!\in\! \AlphAsyncSubscript^\omega\! \mid \exists w \in \AlphAsyncSubscript^\omega. \, w \!\in\! \SyncToAsync(\lang(\mathcal{S})) \land \forall v' \leq w'.~\exists u, u' \in \AlphAsyncSubscript^*.~
	\\ & 	\hspace{2.5cm} v' \cdot u' \text{ is channel-compliant}~\land~
	u \leq w~\land \forall \procA \in \Procs. ~(v' \cdot u') \wproj_{\AlphAsync_\procA} = u \wproj_{\AlphAsync_\procA}
 \}
\enspace . 
\end{align*}
Membership of infinite words is defined in terms of their prefixes: every prefix $v'$ must be channel-compliant, and moreover extensible to a word $v'u'$ that is indistinguishable from another prefix already in the language. 
Since we do not make any fairness assumptions on scheduling, the semantics of infinite words includes traces such as 
$(\snd{\procA}{\procB}{\val})^ω$
for the protocol 
$(\snd{\procA}{\procB}{\val}.\rcv{\procA}{\procB}{\val})^ω$, where only the sender is scheduled. 
Membership of finite words follows standard MSC~semantics. 

In the remainder of the paper, we overload notation and use $\lang(\mathcal{S})$ to denote $\interswaplang(\mathcal{S})$. 
\paragraph{Communicating LTS}
We use communicating LTS to model the local behaviors of participants:
$\mathcal{T} = \CLTS{T}$ is a \emph{communicating labeled transition system} (CLTS) over $\Procs$ and~$\MsgVals$ if
${T}_\procA$
is a deterministic LTS
over~$\AlphAsync_\procA$ for every $\procA\in\Procs$, denoted by
$(Q_\procA, \AlphAsync_\procA, \delta_\procA, q_{0, \procA}, F_\procA)$.
Let 
$\prod_{\procA \in \Procs} Q_\procA$ 
denote the set of global states and
\mbox{$\channels = \set{(\channel{\procA}{\procB}) \mid \procA,\procB\in \Procs, \procA\neq \procB}$}
denote the set of channels. 
A~\emph{configuration} of $\mathcal{A}$ is a pair $(\vec{s}, \xi)$, where $\vec{s}\,$ is a global state and
$\xi : \channels \rightarrow \MsgVals^*$
is a mapping from each channel to a sequence of messages.
We use $\vec{s}_\procA$ to denote the state of $\procA$ in $\vec{s}$.
The CLTS transition relation, denoted $\rightarrow$, is defined as follows. 
\begin{itemize}
	\item
	$(\vec{s},\xi) \xrightarrow{\snd{\procA}{\procB}{\val}} (\pvec{s}',\xi')$ if
	$(\vec{s}_\procA, \snd{\procA}{\procB}{\val}, \pvec{s}'_\procA)\in\delta_\procA$,
	$\vec{s}_\procC = \pvec{s}'_\procC$ for every participant $\procC \neq \procA$,
	$\xi'(\channel{\procA}{\procB}) =  \xi(\channel{\procA}{\procB})\cdot\val$ and $\xi'(c) = \xi(c)$ for every other channel $c\in \channels$.
	
	\item
	$(\vec{s},\xi) \xrightarrow{\rcv{\procA}{\procB}{\val}} (\pvec{s}',\xi')$ if
	$(\vec{s}_\procB, \rcv{\procA}{\procB}{\val}, \pvec{s}'_\procB)\in\delta_\procB$,
	$\vec{s}_\procC = \pvec{s}'_\procC$ for every participant $\procC \neq \procB$,
	$\xi(\channel{\procA}{\procB}) = \val\cdot \xi'(\channel{\procA}{\procB})$
	and $\xi'(c) = \xi(c)$ for every other channel $c\in \channels$.
\end{itemize}
In the initial configuration $(\vec{s}_0, \xi_0)$, each participant's state in $\vec{s}_0$ is the initial state $q_{0,\procA}$ of $A_\procA$, and $\xi_0$ maps each channel to $\emptystring$.
A configuration $(\vec{s}, \xi)$ is \emph{final} iff $\vec{s}_\procA$ is final for every $\procA$ and $\xi$ maps each channel to~$\emptystring$.
Runs and traces are defined in the expected way. 
A run is \emph{maximal} if either it is finite and ends in a final configuration, or it is infinite. 
The language $\lang(\mathcal{T})$ of the CLTS $\mathcal{T}$ is defined as the set of maximal traces.
A configuration $(\vec{s}, \xi)$ is a \emph{deadlock} if it is not final and has no outgoing transitions.
A CLTS is \emph{deadlock-free} if no reachable configuration is a~deadlock.

Observe that in a CLTS, send transitions are always enabled, whereas receive transitions are only enabled if the message exists at the head of its corresponding channel.
Communicating state machines \cite{DBLP:journals/jacm/BrandZ83} are a special case of CLTS where the LTS for each participant $\procA \in \Procs$ is a deterministic finite state machine.
Note that CLTS describe asynchronous communication with message channels of unbounded size. Thus, they differ from Zielonka's \emph{asynchronous automata}~\cite{DBLP:journals/ita/Zielonka87}, which actually describe \emph{synchronously communicating systems}~\cite{Mukund2002}. 
We refer the reader to \cite{DBLP:books/ws/95/DR1995} for further details. \fs{FIX: citet?}

Finally, we define protocol implementability. 
\begin{definition}[Protocol Implementability]
	\label{def:lts-implementability}
	A protocol $\mathcal{S}$ is \emph{implementable} if there exists a CLTS $\CLTS{T}$ such that the following two properties hold:
	\begin{inparaenum}[(i)]
		\item \label{def:lts-implementability-protocol-fidelity}
		\emph{protocol fidelity}: $\lang(\CLTS{T}) = \lang(\mathcal{S})$, and
		\item \label{def:lts-implementability-deadlock-freedom}
		\emph{deadlock freedom}: $\CLTS{T}$ is deadlock-free.
	\end{inparaenum}
	We say that $\CLTS{T}$ implements $\mathcal{S}$.
\end{definition}
A notion of implementability that relaxes language equality to language inclusion has been studied as protocol refinement~\cite{DBLP:conf/esop/LiSW24}. 
Alternatively, one can expand the set of protocol behaviors to include deadlocking behaviors, resulting in a notion of implementability that replaces language equality with prefix set equality, and waives the requirement of deadlock freedom. 
\subsection{Symbolic protocols with dependent refinements}
We now introduce our model for finitely representing infinite state protocols.
We refer to these representations simply as \emph{symbolic protocols}.
\Cref{fig:two-bidder-symbolic} shows the two-bidder protocol from \cref{fig:two-bidder-protocol} expressed as a symbolic protocol. 

\begin{figure}
	\begin{minipage}[b]{.82\textwidth}
		\centering
		\resizebox{1.03\textwidth}{!}{
			\begin{tikzpicture}[sem, node distance=2em and 1em]
				
				\node[state, initial left, initial text = ] (q0) {$q_{0}$};
				\node[state, right = 7.7em of q0] (q1) {$q_{1}$};
				\node[state, right = 7.7em of q1] (q2) {$q_{2}$};
				\node[state, right = 7.7em of q2] (q3) {$q_{3}$};
				\node[state, below = 3em of q3] (q5) {$q_{4}$};
				\node[state, below left = 2em and 3.5em of q5] (q6) {$q_{5}$};
				\node[state, accepting, below = of q6] (qf) {$q_{6}$};
				
				\node[state, below right = 2em and 3.5em of q5] (q7) {$q_{7}$};
				\node[state, below = of q7] (q8) {$q_{8}$};
				\node[below = 0em of q0] {
					{\scriptsize 
						$
						\left\{ \hspace{-1.5ex}
						\begin{array}{c}
							r_y = 0 
							\, \land \, r_z = 0 \phantom{.} \\
							\,\, \land \, r_{z_1} = 0 
							\, \land \, r_{z_2} = 0 \phantom{.} \\
						\end{array}
						\hspace{-1.5ex} \right\}
						$
					}
				};
				\path (q0) edge node[above,pos=.5] {$\msgFromTo{\buyerOne}{\seller}{y 
						\left\{ \hspace{-1.5ex}
						\begin{array}{c}
							\isbn(y) \\
							\, \land \, r_y' = y \phantom{.} 
						\end{array}
						\hspace{-1.5ex} \right\}
					}$} (q1);
				\path (q1) edge node[above,pos=.5] {$\msgFromTo{\buyerOne}{\buyerTwo}{y 
						\cond{y = r_y}
					}$} (q2);
				\path (q2) edge node[above,pos=.5] {$\msgFromTo{\seller}{\buyerOne}{z \left\{ \hspace{-1.5ex} 
						\begin{array}{c}
							z > 0 \\ 
							\, \land \, r_z' = z \phantom{.} 
						\end{array} \hspace{-1.5ex} \right\}}$} (q3);
				\path (q3) edge node[right] {$\msgFromTo{\buyerOne}{\buyerTwo}{b_1
						\left\{ \hspace{-1.5ex}
						\begin{array}{c}
							b_1 > r_{z_1} \\
							\, \land \, r'_{z_1} = b_1 \phantom{.} 
						\end{array}
						\hspace{-1.5ex} \right\}
					}$} (q5);
				\path (q5) edge node[left,yshift=.5ex] {$\msgFromTo{\buyerTwo}{\seller}{x \cond{x=\quitMsg}}$} (q6);
				\path (q6) edge node[left] {$\msgFromTo{\seller}{\buyerOne}{x \cond{x=\quitMsg}}$} (qf);
				\path (q5) edge node[right,yshift=.5ex] {$\msgFromTo{\buyerTwo}{\buyerOne}{b_2
						\left\{ \hspace{-1.5ex}
						\begin{array}{c}
							\,\, b_2 > r_{z_2} \\ 
							\, \land \, r'_{z_2} = b_2 
						\end{array}
						\hspace{-1.5ex} \right\}
					}$} (q7);
				\path (q7) edge node[right] {$\msgFromTo{\buyerOne}{\seller}{x
						\left\{ \hspace{-1.5ex}
						\begin{array}{c}
							\, x = \succMsg \phantom{.} \\ 
							\, \land \, r_{z_1} + r_{z_2} \geq r_z \phantom{.} \\
						\end{array}
						\hspace{-1.5ex} \right\}
					}$} (q8);
				\path (q8) edge node[below,pos=.45] {$\msgFromTo{\seller}{\buyerTwo}{x \cond{x=\succMsg}}$} (qf);
				\draw[semarrow] (q7.east) -- node {} ++(3,0) |- (q3) node[above,xshift=2.8cm]
				{$\msgFromTo{\buyerOne}{\buyerTwo}{
						x
						\left\{ \hspace{-1.5ex}
						\begin{array}{c}
							x = \contMsg \\
							\, \land \, r_{z_1} + r_{z_2} < r_z \phantom{.}
						\end{array}
						\hspace{-1.5ex} \right\}
					}$} node[] {};
				
			\end{tikzpicture}
			
		}
		\caption{The two-bidder protocol from \cref{fig:two-bidder-protocol} as a symbolic protocol with registers $r_z$, $r_y$, $r_{z_1}$, and $r_{z_2}$.}
		\label{fig:two-bidder-symbolic}
	\end{minipage}
\end{figure}

The formal definition of this class of symbolic protocols is given below.
In this definition, we assume a fixed but unspecified first-order background theory of message values (\eg linear integer arithmetic).
We assume standard syntax and semantics of first-order formulas and denote by $\formulas$ the set of first-order formulas with free variables drawn from an infinite set $X$.
We assume that these variables are interpreted over the set of message values $\MsgVals$.
For a valuation $\rho \in X \to \MsgVals$ and $\varphi \in \formulas(X)$, we write $\rho \models \phi$ to indicate that $\varphi$ evaluates to true under $\rho$ in the underlying theory.
 
\begin{definition}[Symbolic protocol]
	A symbolic protocol is a tuple 
	$\SymProt = (S, R, \Delta, s_{0}, \rho_0, F)$ 
        where
        \begin{itemize}
         \item $S$ is a finite set of control states,
         \item $R$ is a finite set of register variables,
         \item $\Delta \subseteq S \times \Procs \times X \times \Procs \times \formulas \times S$ is a finite set that consists of symbolic transitions of the form $s \xrightarrow{\msgFromToNS{\procA}{\procB}{x \set{\varphi}}} s'$ where the formula $\varphi$ with free variables $R \uplus R' \uplus \set{x}$ expresses a transition constraint that relates the old and new register values ($R$ and $R'$), and the sent value $x$,
		 \item $s_0 \in S$ is the initial control state,
		 \item $\rho_0 \from R \rightarrow \MsgVals$ is the initial register assignment, and
		 \item $F \subseteq S$ is a set of final states.
        \end{itemize}
\end{definition}
To streamline our definition, we choose not to separate 
register update expressions from predicates describing the communication. 
Rather, we specify everything together in a single formula~$\varphi$, that can only talk about the current value being communicated, and the register values in the pre- and post-state.
Thus, $\varphi$ can describe values that are communicated between participants, in addition to register assignments and updates. 
For example, $\msgFromToNS{\procA}{\procB}{x \set{even(x) \land r_1' = r_1 + 2 \land r_2' = x}}$ describes $\procA$ sending $\procB$ an even number $x$, incrementing the value of register $r_1$ by 2, and storing the value of $x$ in register $r_2$. 
We formally specify the two-bidder protocol from \cref{fig:two-bidder-protocol} as a symbolic protocol in \cref{fig:two-bidder-symbolic} for demonstration purposes; note that the transition predicate $ISBN(y)$ from $q_1$ to $q_2$ is replaced with an equality. 
For readability and conciseness, we employ the following conventions from now on.
We treat communication variables as registers that are automatically assigned the communicated value, \eg
$\msgFromTo{\seller}{\buyerOne}{z \cond{z > 0}}$
should be understood as
$\msgFromTo{\seller}{\buyerOne}{x \cond{x > 0 \land z' = x}}$ for some fresh $x$.
Furthermore, if the communicated value is a constant $c$ and there is no need to store this value, we inline it and write
$\msgFromTo{\seller}{\buyerTwo}{\succMsg \cond{\top}}$
instead of
$\msgFromTo{\seller}{\buyerTwo}{x \cond{x=\succMsg}}$.
We may omit the condition $\top$,
turning
$\msgFromTo{\seller}{\buyerTwo}{\succMsg \cond{\top}}$
into
$\msgFromTo{\seller}{\buyerTwo}{\succMsg}$.

Symbolic protocols are specification-wise similar to symbolic register automata~\cite{DBLP:conf/cav/DAntoniFS019}, but allow more general patterns of register manipulation and do not a priori require formulas to come from an effective Boolean algebra.
Symbolic protocols can be seen as a finite description of an infinite-state LTS, whose concrete states consist of a control state along with an assignment for the register variables $R$. Transitions are concrete communication events that optionally modify register values. 
We formally define the concretization of a symbolic protocol below. 
\begin{definition}[Concretization of symbolic protocols]
	For a symbolic protocol $\SymProt=(S, R, \Delta, s_{0}, \rho_0, F)$, let $\mathcal{S}_{\SymProt}$ denote its concrete protocol.
	The set of states of $\mathcal{S}_{\SymProt}$ is $S \times (R \to \MsgVals)$.
	
	Transitions in $\mathcal{S}_{\SymProt}$ are defined as follows: 
	\[
	\infer
	{
		s_1 \xrightarrow{\msgFromToNS{\procA}{\procB}{x \set{\varphi}}} s_2 \in \Delta \\ 
		\rho_1\rho_2'[x \mapsto v] \models \varphi
	}
	{
		(s_1, \rho_1) \xrightarrow{\msgFromToNS{\procA}{\procB}{v}} (s_2, \rho_2)
	}
	\]
        Intuitively, the rule says that a symbolic transition from $s_1$ to $s_2$ can be instantiated to one from $(s_1, \rho_1)$ to $(s_2, \rho_2)$ on value $v$ when $v$ together with the register assignments in the pre- and post-state satisfy the transition constraint $\varphi$. Here, we use juxtaposition $\rho_1\rho_2'$ of register assignments to express their disjoint union. The assignment $\rho_2'$ is obtained from $\rho_2$ by replacing registers $r$ in the domain with their primed version in $R'$. 
	The initial state is defined as $(s_0, \rho_0)$. 
	A state $(s, \rho)$ in $\mathcal{S}_{\SymProt}$ is final when $s \in F$. 
\end{definition}

Thus, the concrete protocol $\mathcal{S}_{\SymProt}$ is a protocol over the alphabet $\AlphSyncSubscript$. 
The language of a symbolic protocol $\SymProt$ is defined as the language of its concretization $\mathcal{S}_{\SymProt}$. 
Consequently, a symbolic protocol is implementable if its concretization is implementable.
\section{Characterizing Protocol Implementability}
\label{sec:characterization}
We motivate our precise characterization of protocol implementability through examples of non-implementable protocols, and show that seemingly disparate sources of non-implementability share a unified semantic explanation. 
Recall the protocol~$\SymProt_1$ from \cref{sec:overview} with a receiver violation, depicted in \cref{fig:receive-violation-and-ok}. 
The infinite-state LTS $\SymProt_1$ contains the two concrete run prefixes depicted in
\cref{fig:receive-violation-run-2}, where the values of $x, y$ are $2,3$ and $1,3$ respectively.
\begin{figure}[b]
\begin{tabular}{rl}
(a) \hspace{-1em} &
\begin{tikzpicture}[sem] 
	\node[state, initial left=, initial text =](q0){};
	
	\node[state, right=0.4cm and 1.5cm of q0](q1){};
	\node[state, right=2.5cm of q1](q2){};
	
	\node[state, above right=0.1cm and 2cm of q2](q2l){};
	\node[state, right=1.5cm of q2l](q3l){};
	\node[state, right=1.5cm of q3l](q4l){};
	\node[state, accepting, right=1.5cm of q4l](q5l){};
	
	\path(q0) edge node[sloped, pos=0.5, above] {$\msgFromTo{\procD}{\procA}{2}$} (q1);
	\path(q1) edge node[sloped, pos=0.5, above] {$\msgFromTo{\procA}{\procC}{3}$} (q2);
	\path(q2) edge node[sloped, pos=0.5, above] {$\msgFromTo{\procA}{\procB}{\msgB}$} (q2l);
	\path(q2l) edge node[sloped, pos=0.5, above] {$\msgFromTo{\procA}{\procC}{\msgO}$} (q3l);
	\path(q3l) edge node[sloped, pos=0.5, above] {$\msgFromTo{\procB}{\procC}{\msgO}$} (q4l);
	\path(q4l) edge node[sloped, pos=0.5, above] {$\msgFromTo{\procC}{\procA}{\msgB}$} (q5l);
\end{tikzpicture}
 \\
\raisebox{4ex}{(b)} \hspace{-1em} &
\begin{tikzpicture}[sem] 
	\node[state, initial left=, initial text =](q0){};
	
	\node[state, right=0.4cm and 1.5cm of q0](q1){};
	\node[state, right=2.5cm of q1](q2){};
	
	\node[state, below right=0.1cm and 2cm of q2](q2r){};
	\node[state, right=1.5cm of q2r](q3r){};
	\node[state, right=1.5cm of q3r](q4r){};
	\node[state, accepting, right=1.5cm of q4r](q5r){};
	
	\path(q0) edge node[sloped, pos=0.5, above] {$\msgFromTo{\procD}{\procA}{1}$} (q1);
	\path(q1) edge node[sloped, pos=0.5, above] {$\msgFromTo{\procA}{\procC}{3}$} (q2);
	\path(q2) edge node[sloped, pos=0.5, below] {$\msgFromTo{\procA}{\procB}{\msgM}$} (q2r);
	\path(q2r) edge node[sloped, pos=0.5, below] {$\msgFromTo{\procB}{\procC}{\msgO}$} (q3r);
	\path(q3r) edge node[sloped, pos=0.5, below] {$\msgFromTo{\procA}{\procC}{\msgO}$} (q4r);
	\path(q4r) edge node[sloped, pos=0.5, below] {$\msgFromTo{\procC}{\procA}{\msgM}$} (q5r);
\end{tikzpicture}
 \end{tabular}
	\caption{
		Two concrete runs of $\SymProt_1$ (\cref{fig:receive-violation-and-ok}):
		(a) with $x = 2$ and $y = 3$ and
		(b) with $x = 1$ and $y = 3$.
		\label{fig:receive-violation-run-2}
	}
\end{figure}

Inspecting $\SymProt_1$'s specification reveals that the protocol expects $\procC$ to receive messages from $\procA$ and~$\procB$ in a different order depending on the branch that $\procB$ chooses to follow.
However, this expectation is unreasonable in a distributed setting.  
Between the two concrete runs,
$\procC$'s partial view of the protocol's behavior is the same: $\procC$ receives a value 3 from $\procA$, yet $\procC$ is expected to receive in $\procA, \procB$ order in one run, and receive in $\procB, \procA$ order in the other.

Recall the protocol $\SymProt_2$ from \cref{sec:overview} with a sender violation, depicted in
\cref{fig:send-violation}.
\begin{figure}[t]
\begin{minipage}[b]{.52\textwidth}
\begin{tikzpicture}[sem] 
	\node[state, initial left=, initial text =](q0){};
	
	\node[state, above right=0.3cm and 1.3cm of q0](q1){};
	\node[state, right=1.3cm of q1](q2){$q_1$};
	\node[state, right=1.3cm of q2](q3){};
	\node[state, accepting, right=1.3cm of q3](q4){};

	\path(q0) edge node[sloped, pos=0.5, above] {$\msgFromTo{\procD}{\procB}{\msgB}$} (q1);
	\path(q1) edge node[sloped, pos=0.5, above] {$\msgFromTo{\procB}{\procA}{4}$} (q2);
	\path(q2) edge node[sloped, pos=0.5, above] {$\msgFromTo{\procA}{\procB}{\msgO}$} (q3);
	\path(q3) edge node[sloped, pos=0.5, above] {$\msgFromTo{\procB}{\procC}{\msgM}$} (q4);

\end{tikzpicture}
 	\caption{
		A concrete run of $\SymProt_2$ (\cref{fig:send-violation}) with $\procD$ choosing the top branch.
		\label{fig:send-violation-run-1}
	}
\end{minipage}
\hfill
\begin{minipage}[b]{.45\textwidth}
\begin{tikzpicture}[sem] 
	\node[state, initial left=, initial text =](q0){};

	\node[state, below right=0.3cm and 1.3cm of q0](q5){};
	\node[state, right= 1.3cm of q5](q6){$q_2$};
	\node[state, accepting, above right=0.07cm and 1.3cm of q6](q6l){};

	\path(q0) edge node[sloped, pos=0.5, below] {$\msgFromTo{\procD}{\procB}{\msgM}$} (q5);
	\path(q5) edge node[below] {$\msgFromTo{\procB}{\procA}{4}$} (q6);
	\path(q6) edge node[sloped, pos=0.5, above] {$\msgFromTo{\procB}{\procC}{\msgB}$} (q6l);
\end{tikzpicture}
 	\caption{
		A concrete run of $\SymProt_2$ (\cref{fig:send-violation}) with \\ $\procD$ choosing the bottom branch and $x_2 = 4$.
		\label{fig:send-violation-run-2} 
	}
\end{minipage}
\end{figure}
Again inspecting $\SymProt_2$'s specification, the branching structure imposes the expectation that on the top branch, $\procA$ should send $\procB$ an $\msgO$ message, whereas on the bottom branch, $\procA$ should immediately terminate.
The two concrete runs in \cref{fig:send-violation-run-1} and \cref{fig:send-violation-run-2} again demonstrate that this expectation is unreasonable: $\procA$~receives the value 3 from $\procB$ in both runs, but in one run is expected to send a message, whereas in the other is expected to terminate. 

The non-implementability in the examples above can be attributed to insufficient local information about protocol control flow. 
This source of non-implementability is inherent to the expressive power of branching choice in protocol specifications, and is present even in finite protocols with more restricted choice constructs. 
While most existing works soundly detect insufficient local information through conservative projection algorithms \cite{DBLP:conf/popl/HondaYC08,DBLP:conf/sfm/CoppoDPY15,DBLP:journals/jlp/ToninhoY17,DBLP:conf/ecoop/ScalasDHY17}, \citet{DBLP:conf/cav/LiSWZ23} give a complete characterization.
To check implementability,
they
first obtain a candidate implementation by restricting the global protocol onto each participant's alphabet, and then determinizing the resulting finite state automaton.
Then,
they
check sound and complete conditions directly on the states of the candidate implementation.

Our first observation towards a precise characterization is that implementability can be checked directly on the global protocol specification, without synthesizing a candidate implementation upfront. 
This is especially important in the general case, when synthesizing a~candidate implementation is itself challenging and not always possible.
Our analysis of the protocols above shows that non-implementability can be blamed solely on the existence of certain states in the concrete LTS represented by the global protocol. 
In fact, we show in \S\ref{sec:symbolic} that the implementability check for global types by \citet{DBLP:conf/cav/LiSWZ23} can be made more efficient by forgoing the synthesis~step. 

Let us now turn our attention to a different source of non-implementability that is unique to the setting of dependent data refinements. 
Consider the following pair of symbolic protocols $\SymProt_3$ and $\SymProt_3'$, depicted in \cref{fig:unrealizable-local-type} and \cref{fig:realizable-local-type}. 
\begin{figure}[b]
\begin{minipage}[b]{.47\textwidth}
\begin{tikzpicture}[sem] 
	
	\node[state, initial left=, initial text =](s0){};
	\node[state, right=1.4cm of s0](s1){};
	\node[state, right=1.8cm of s1](s2){};
	\node[state, accepting, right=1.8cm of s2](s3){};

	\path(s0) edge node[sloped, pos=0.5, above] {$\msgFromTo{\procA}{\procB}{x \set{\top}}$} (s1);
	\path(s1) edge node[sloped, pos=0.5, above] {$\msgFromTo{\procB}{\procC}{y \set{y > x}}$} (s2);
	\path(s2) edge node[sloped, pos=0.5, above] {$\msgFromTo{\procC}{\procA}{z \set{z > x}}$} (s3);
\end{tikzpicture}
 	\caption{$\SymProt_3$: A non-implementable protocol with \\ dependent refinements.
	\label{fig:unimplementable-local-type}
	\label{fig:unrealizable-local-type}}
\end{minipage}
\hfill
\begin{minipage}[b]{.47\textwidth}
\begin{tikzpicture}[sem] 
	
	\node[state, initial left=, initial text =](s0){};
	\node[state, right=1.4cm of s0](s1){};
	\node[state, right=1.8cm of s1](s2){};
	\node[state, accepting, right=1.8cm of s2](s3){};

	\path(s0) edge node[sloped, pos=0.5, above] {$\msgFromTo{\procA}{\procB}{x \set{\top}}$} (s1);
	\path(s1) edge node[sloped, pos=0.5, above] {$\msgFromTo{\procB}{\procC}{y \set{y = x}}$} (s2);
	\path(s2) edge node[sloped, pos=0.5, above] {$\msgFromTo{\procC}{\procA}{z \set{z > x}}$} (s3);
\end{tikzpicture}
 	\caption{$\SymProt_3'$: An implementable protocol with \\ dependent refinements.
	\label{fig:implementable-local-type}
	\label{fig:realizable-local-type}}
\end{minipage}
\end{figure}

Non-implementability is again caused by insufficient local information, but this time with respect to message data rather than control flow: in fact, no branching choice appears in this pair of simple protocols. 
The problem instead arises in the fact that in both $\SymProt_3$ and $\SymProt_3'$, $\procC$ does not know the value of $x$. 
While an implementation for $\procC$ could produce a subset of $\SymProt_3$'s behaviors (\eg by sending $z$ such that $z > y$), or produce a superset of $\SymProt_3$'s behaviors (\eg by sending all values for $z$), no implementation can produce exactly the specified behaviors, as required by protocol fidelity. 
\citet{DBLP:journals/pacmpl/00020HNY20} address partial information of protocol variables by syntactically classifying whether a variable is known or unknown to a participant, and annotating the variables accordingly in the typing context: a variable is known to its sender and receiver, and unknown to all other participants. 
However, this syntactic analysis is itself insufficient, as demonstrated by these examples: both protocols yield the same classification of variables per participant, yet one is implementable and the other is not.

We instead turn to concrete runs of $\SymProt_3$ to find the source of non-implementability. 
Let us consider the concrete runs of $\SymProt_3$
depicted in
\cref{fig:unimplementable-local-type-runs}
, where the values of $x,y$ are $2,4$ and $3,4$~respectively.
\begin{figure}[t]
	\begin{tabular}{rl}
		(a) \hspace{-1em} &
\begin{tikzpicture}[sem] 
	
	\node[state, initial left=, initial text =](s0){};
	\node[state, right=1.4cm of s0](s1){};
	\node[state, right=1.8cm of s1](s2){};
	\node[state, accepting, right=1.8cm of s2](s3){};

	\path(s0) edge node[sloped, pos=0.5, above] {$\msgFromTo{\procA}{\procB}{2}$} (s1);
	\path(s1) edge node[sloped, pos=0.5, above] {$\msgFromTo{\procB}{\procC}{4}$} (s2);
	\path(s2) edge node[sloped, pos=0.5, above] {$\msgFromTo{\procC}{\procA}{3}$} (s3);
\end{tikzpicture}
 		\\
		(b) \hspace{-1em} &
\begin{tikzpicture}[sem] 
	
	\node[state, initial left=, initial text =](s0){};
	\node[state, right=1.4cm of s0](s1){};
	\node[state, right=1.8cm of s1](s2){};
	\node[state, accepting, right=1.8cm of s2](s3){};

	\path(s0) edge node[sloped, pos=0.5, above] {$\msgFromTo{\procA}{\procB}{3}$} (s1);
	\path(s1) edge node[sloped, pos=0.5, above] {$\msgFromTo{\procB}{\procC}{4}$} (s2);
	\path(s2) edge node[sloped, pos=0.5, above] {$\msgFromTo{\procC}{\procA}{4}$} (s3);
\end{tikzpicture}
 	\end{tabular}
	\caption{
		Two concrete runs of $\SymProt_3$ (\cref{fig:unimplementable-local-type}):
		(a) with $x = 2, y = 4, z = 3$ and
		(b) with $x = 3, y = 4, z = 4$.
		\label{fig:unimplementable-local-type-runs}
	}
\end{figure}

In this pair of runs, $\procC$ observes the same behaviors, namely receiving the value 4 from $\procB$. 
While $\SymProt_3$ also permits $\procC$ to send 4 to $\procA$ in the first run, sending 3 to $\procA$ in the second run constitutes a violation to the refinement predicate $z > x$, \ie $3 > 3$ is false.
Again, this presents a problem because the two run prefixes are indistinguishable to $\procC$. 
Observe that in this example, non-implementability can again be blamed solely on the existence of states in the global protocol. 

We formalize a participant's local information about the protocol using two variations on the standard notion of reachability.
Let $\mathcal{S} = (S, \AlphSyncSubscript, T, s_0, F)$ be a protocol and let $\procA \in \Procs$ be a participant. 
The standard notion defines $s'$ as reachable from $s$ in $\mathcal{S}$ on $w \in \AlphSyncSubscript^*$, denoted
$s \xrightarrow{w} \Kleenestar s'$, when there exists a sequence of transitions $s_1 \xrightarrow{l_1} s_2 \ldots s_{n-1} \xrightarrow{l_{n-1}} s_n$, such that $s_1 = s$, $s_n = s'$, $l_1 \ldots l_{n-1} = w$ and for each $1 \leq i < n$, it holds that $s_i \xrightarrow{l_i} s_{i+1} \in T$. 
We first define a notion of reachability that restricts the transitions to only the actions observable by a single participant. 
\paragraph{Participant-based Reachability.}
We say that $s \in S$ is \emph{reachable for $\procA$ on $u \in \AlphSync_\procA^*$} when there exists $w \in \AlphSyncSubscript^*$ such that 
$s_0 \xrightarrow{w} \Kleenestar s \in T$ and $w \wproj_{\AlphSync_{\procA}} = u$, 
which we denote $s_0 \xRightarrow[\procA]{u} \Kleenestar s$. 
We characterize \emph{simultaneously reachable} pairs of states for each participant using the notion of participant-based reachability.
\paragraph{Simultaneous Reachability.}
We say that $s_1, s_2 \in S$ are simultaneously reachable for participant $\procA$ on $u \in \AlphSync_\procA^*$, denoted
$s_0 \xRightarrow[\procA]{u} \Kleenestar s_1, s_2$, 
if there exist $w_1, w_2 \in \AlphSyncSubscript^*$ such that $s_0 \xrightarrow{w_1} \Kleenestar s_1 \in T, s_0 \xrightarrow{w_2} \Kleenestar s_2 \in T$ and $w_1 \wproj_{\AlphSync_\procA} = w_2 \wproj_{\AlphSync_\procA} = u$.  
Simultaneous reachability captures the notion of \emph{locally indistinguishable} states: to a participant, two states are locally indistinguishable if they are simultaneously reachable. 
Send Coherence requires that any message that can be sent from a state can also be sent from all other states that are locally indistinguishable to the sender. 
\begin{definition}[Send Coherence]
	\label{cond:scc} 
	A protocol
	$\mathcal{S} = (S, \AlphSyncSubscript, T, s_0, F)$  
	satisfies Send Coherence (SC) when for every $s_1 \xrightarrow{\msgFromToNS{\procA}{\procB}{\val}} s_2 \in T, s_1' \in S$:
	\[
		(\exists u \in \AlphSync_{\procA}^*.~s_0\xRightarrow[\procA]{u}\Kleenestar s_1, s_1')
		\implies 
		(\exists s_2' \in S.~s_1' \xRightarrow[\procA]{\msgFromToNS{\procA}{\procB}{\val}} \Kleenestar s_2') \enspace.
	\]
\end{definition}
Receive Coherence, on the other hand, requires that no message which can be received from a state can be received from any other state that is locally indistinguishable to the receiver. 
\begin{definition}[Receive Coherence]
	\label{cond:rcc} 
	A protocol
	$\mathcal{S} = (S, \AlphSyncSubscript, T, s_0, F)$  
	satisfies Receive Coherence (RC) when for every $s_1 \xrightarrow{\msgFromToNS{\procA}{\procB}{\val}} s_2, s_1' \xrightarrow{\msgFromToNS{\procC}{\procB}{\val}} s_2' \in T$:
	\[
		(\procC \neq \procA \, \land \,
		\exists u \in \AlphSync_{\procB}^*.~s_0 \xRightarrow[\procB]{u}\Kleenestar s_1, s_1')
		\!\implies\! 
		\forall w \in \pref (\lang(\mathcal{S}_{s_2'})).~
		w \wproj_{\Alphabet_{\procB}} \!\neq\! \emptystring \, \lor
		\, \MsgVals(w \wproj_{\snd{\procA}{\procB}{\_}}) \!\neq\!
		\MsgVals(w \wproj_{\rcv{\procA}{\procB}{\_}}) \cdot \val) \;.
	\]
\end{definition}
No Mixed Choice requires that roles cannot equivocate between sending and receiving in two locally indistinguishable states.
\begin{definition}[No Mixed Choice]
	\label{cond:nmc} 
	A protocol
	$\mathcal{S} = (S, \AlphSyncSubscript, T, s_0, F)$  
	satisfies No Mixed Choice (NMC) when for every $s_1 \xrightarrow{\msgFromToNS{\procA}{\procB}{\val}} s_2, s_1' \xrightarrow{\msgFromToNS{\procC}{\procA}{\val}} s_2' \in T$:
	$
	(\exists u \in \AlphSync_{\procA}^*.~s_0 \xRightarrow[\procA]{u}\Kleenestar s_1, s_1')
	\implies 
	\bot 
	$ \enspace.
\end{definition}
Our semantic characterization of protocol implementability is the conjunction of the above three conditions. 
In contrast to the syntactic analysis in \cite{DBLP:journals/pacmpl/00020HNY20}, our semantic approach is sound and complete.  \fs{FIX: citet?}
In contrast to the sound and complete approach in \cite{DBLP:conf/cav/LiSWZ23}, our implementability conditions do not rely on synthesizing an implementation upfront. \fs{FIX: citet?}
\begin{definition}[Coherence Conditions]
	A protocol
	satisfies Coherence Conditions (CC) when it satisfies Send Coherence, Receive Coherence and No Mixed Choice.
\end{definition}
The preciseness of \Characterization is stated as follows.
\begin{restatable}{theorem}{Preciseness of Coherence Conditions}
	\label{thm:soundness} 
	Let $\mathcal{S}$ be a protocol. 
	Then, $\mathcal{S}$ is implementable if and only if it satisfies CC. 
\end{restatable} 
In the next two sections, we illustrate the key steps for proving \cref{thm:soundness}. We refer the reader to \cref{app:characterization} for the complete proofs. 
\subsection{Soundness}
Soundness requires us to show that if a protocol satisfies \Characterization, then it is implementable. 
We begin by echoing the observation made in several prior works~\cite{DBLP:journals/tse/AlurEY03,  DBLP:conf/ecoop/Stutz23, DBLP:conf/cav/LiSWZ23} that for any global protocol, there exists a \emph{canonical} implementation consisting of one local implementation per participant. 
We formally define what it means for an implementation to be canonical in our setting below.
\begin{definition}[Canonical implementations]
	
	\label{def:local-language-property}
	We say a CLTS $\CLTS{T}$ is a \emph{canonical implementation} for a protocol $\mathcal{S} = (S, \AlphSyncSubscript, T, s_0, F)$
	if for every $\procA \in \Procs$, $T_\procA$ satisfies: \\
	\begin{inparaenum}[(i)]
		\item $\forall w \in \Alphabet_{\procA}^*.~w \in \lang(T_\procA) \Leftrightarrow w \in \lang(\mathcal{S}) \wproj_{\Alphabet_\procA}$, and
		\item $\pref(\lang(T_\procA)) = \pref(\lang(\mathcal{S}) \wproj_{\Alphabet_\procA})$.
	\end{inparaenum}
\end{definition}
We first prove that following fact about canonical implementations of protocols satisfying NMC, which states that the canonical implementations themselves do not exhibit mixed choice.
\begin{restatable}[No Mixed Choice]{lemma}{noMixedChoice}
	\label{cor:no-mixed-choice}
	Let $\mathcal{S}$ be a protocol satisfying NMC (\cref{cond:nmc}) and let $\CLTS{T}$ be a canonical implementation for $\mathcal{S}$. 
	Let $wx_1, wx_2 \in \pref(\lang(T_\procA))$ with $x_1 \neq x_2$ for some $\procA \in \Procs$.
	Then, $x_1 \in \Alphabet_!$ iff $x_2 \in \Alphabet_!$.
\end{restatable}
We choose the canonical implementation as our existential witness to show that any protocol satisfying \Characterization is implementable. 
By the definition of implementability (\cref{def:lts-implementability}), soundness amounts to showing the following three conditions: \\
\begin{inparaenum}[(a)]
	\item 
	\label{claim:global-lang-incl-clts-lang}
	$\lang(\mathcal{S}) \subseteq \lang(\CLTS{T})$, 
	\item 
	\label{claim:clts-lang-incl-global-lang}
	$\lang(\CLTS{T}) \subseteq \lang(\mathcal{S})$, and 
	\item 
	\label{claim:clts-deadlock-free}
	$\CLTS{T}$ is deadlock-free.
\end{inparaenum}
Condition (a) states that any canonical implementation recognizes at least the global protocol behaviors. This fact can be shown for any LTS and canonical CLTS, and does not rely on assumptions about determinism or sender-drivenness, nor assumptions about the LTS satisfying \Characterization.
\begin{restatable}[Canonical implementation language contains protocol language] {lemma}{languageInclusionLeft}
	Let $\mathcal{S}$ be an LTS and let $\CLTS{T}$ be a canonical implementation for $\mathcal{S}$.
	Then, $\lang(\mathcal{S}) \subseteq \lang(\CLTS{T})$. 
\end{restatable}
Condition (b), on the other hand, states that any behavior recognized by the canonical implementation is a global protocol behavior, in other words, that the canonical CLTS does not add behaviors. This is only true for protocols that satisfy \Characterization. 

Furthermore, the acceptance condition for infinite words in $\lang(\mathcal{S})$ differs from that in $\CLTS{T}$: the latter accepts all infinite traces, whereas the former requires to show that an infinite word $w$ satisfies $w \preceq_\interswap^\omega w'$ for some other infinite word $w' \in \lang(\mathcal{S})$. 
Therefore, showing prefix set inclusion is not sufficient, and we must reason about the finite and infinite case separately. 
\begin{restatable}[Global protocol language contains canonical implementation language]{lemma}{languageInclusionRight}
	Let $\mathcal{S}$ be a protocol satisfying \Characterization and let $\CLTS{T}$ be a canonical implementation for $\mathcal{S}$ such that for all $w \in \AlphAsyncSubscript^*$, if $w$ is a trace of $\CLTS{T}$, then $I(w) \neq \emptyset$.

	Then, $\lang(\CLTS{T}) \subseteq \lang(\mathcal{S})$. 
\end{restatable} 
Towards these ends, we adapt the key intermediate lemma from \cite{DBLP:conf/cav/LiSWZ23} to our setting, and show the inductive invariant that every trace in the canonical implementation of a protocol satisfying \Characterization satisfies \emph{intersection set non-emptiness}. 
Note that although our intermediate lemma statements are similar to those in~\cite{DBLP:conf/cav/LiSWZ23} in structure, \cite{DBLP:conf/cav/LiSWZ23} reasons about a particular implementation, namely the subset construction obtained from the global type, whereas our proofs reason about any canonical implementation of a global protocol that satisfies \Characterization. \fs{FIX: citet?}
As a result, the proof arguments differ significantly. 

We adapt the relevant definitions to our setting below. 
\begin{definition}[LTS intersection sets]
	\label{def:lts-intersection-sets}
	Let $\mathcal{S}$
	
	be an LTS.
	Let $\procA$ be a participant and
	$w \in \AlphAsyncSubscript^*$ be a word. 
	We define the set of possible runs $\globcomplocal{\mathcal{S}}{\procA}{w}$
	as all maximal runs of $\mathcal{S}$ that are consistent with $\procA$'s local view of $w$:
	\[
	\globcomplocal{\mathcal{S}}{\procA}{w}
	\is
	\set{
		\run
		\text{ is a maximal run of }
		\mathcal{S}
		\mid
		w \wproj_{\Alphabet_\procA} \preforder \SyncToAsync(\trace(\run)) \wproj _{\Alphabet_\procA}
	}
	\enspace .
	\]
	We denote the intersection of the possible run sets for all participants as
	$
	I^{\mathcal{S}}(w) 
	\is 
	\Inters_{\procA \in \Procs} \globcomplocal{\mathcal{S}}{\procA}{w}
	$.
\end{definition}
\begin{definition}[Unique splitting of a possible run]
	Let $\mathcal{S}$ be an LTS, $\procA$ a participant, and $w \in \AlphAsyncSubscript^*$ a word. 
	Let $\run$ be a run in $\globcomplocal{\mathcal{S}}{\procA}{w}$. 
	We define the longest prefix of $\run$ matching $w$:
	\[
	\alpha'
	\is
	\max
	\set{
		\run'
		\mid
		\run' \leq \run ~\wedge~
		\SyncToAsync(\trace(\run')) \wproj _{\Alphabet_\procA} \preforder w \wproj_{\Alphabet_\procA}
	}
	\enspace .
	\]
	If $\alpha' \neq \run$, we can split $\run$ into
	$
	\run = \alpha \cdot s \xrightarrow{l} s' \cdot \beta
	$
	where $\alpha' = \alpha \cdot s$. 
	, which we call the unique splitting of~$\run$ for $\procA$ matching $w$.
	Uniqueness follows from the maximality of $\alpha'$. 
\end{definition}

For example, the unique splitting of 
$\run = 
s_1 \xrightarrow{\msgFromToNS{\procA}{\procB}{\msgM}} s_2 \xrightarrow{\msgFromToNS{\procC}{\procB}{\msgB_1}} s_3
\xrightarrow{\msgFromToNS{\procC}{\procB}{\msgB_2}} s_4
\xrightarrow{\msgFromToNS{\procB}{\procA}{\msgO}} s_5$ 
for $\procA$ matching 
$w = \snd{\procC}{\procB}{\msgB_1}. \,\snd{\procA}{\procB}{\msgM}$ is 
$\alpha \cdot s_3 \xrightarrow{\msgFromToNS{\procC}{\procB}{\msgB_2}} s_4 \cdot \beta$, 
where $\alpha =  s_1 \xrightarrow{\msgFromToNS{\procA}{\procB}{\msgM}} s_2 \xrightarrow{\msgFromToNS{\procC}{\procB}{\msgB_1}} s_3$ and $\beta = s_4
\xrightarrow{\msgFromToNS{\procB}{\procA}{\msgO}} s_5$. 
Our intersection non-emptiness inductive invariant is stated below. 
The proof proceeds by induction on the length of a prefix $w$ of the canonical implementation, and case splits based on whether $w$ is extended by a send or receive action. 
\cref{lm:sndPrefixPreservation} and \cref{lm:rcvIntersectionSetEquality} provide a characterization for each case respectively. 
\begin{restatable}[Intersection set non-emptiness]{lemma}{INonEmpty}
	\label{lm:intersection-non-emptiness}
	Let $\mathcal{S}$ be a protocol satisfying \Characterization, and let $\CLTS{T}$ be a canonical implementation for $\mathcal{S}$. 
	Then, for every trace $w \in \AlphAsyncSubscript^*$ of $\CLTS{T}$, it holds that $I(w) \neq \emptyset$. 
\end{restatable}
\begin{restatable}[Receive events do not shrink intersection sets]{lemma}{rcvIntersectionEquality}
	\label{lm:rcvIntersectionSetEquality}
	Let $\mathcal{S}$ be a protocol satisfying \Characterization, and let $\CLTS{T}$ be a canonical implementation for $\mathcal{S}$.  
	Let $wx$ be a trace of $\CLTS{T}$ such that $x \in \Alphabet_?$. 
	Then, $I(w) = I(wx)$. 
\end{restatable}
\begin{restatable}[Send events preserve run prefixes]{lemma}{sndPreservation}
	\label{lm:sndPrefixPreservation}
	Let $\mathcal{S}$ be a protocol satisfying \Characterization and $\CLTS{T}$ be a canonical implementation for $\mathcal{S}$.  
	Let $wx$ be a trace of
	$\CLTS{T}$ such that $x \in \Alphabet_{\procA,!}$
	
	for some $\procA \in \Procs$.
	Let $\rho$ be a run in $I(w)$, and $\alpha \cdot s_{pre} \xrightarrow{l} s_{post} \cdot \beta$ be the unique splitting of $\rho$ for $\procA$ with respect to $w$.
	Then, there exists a run $\rho'$ in $I(wx)$ such that $\alpha \cdot s_{pre} \leq \rho'$.
\end{restatable}
Finally, we show that protocols that satisfy \Characterization and intersection set non-emptiness have deadlock-free canonical implementations. The proof follows immediately from the following lemma and the fact that CLTS are deterministic, and is thus omitted. 
\begin{restatable}[Channel compliance and intersection set non-emptiness implies prefix]{lemma}{ccIPrefix}
	\label{lm:cc-intersection-nonemptiness-implies-prefix}
	Let $\mathcal{S} = (S, \AlphSyncSubscript, T, s_0, F)$ be a protocol and let $w \in \AlphAsyncSubscript^*$ be a word such that

	(i) $w$ is channel-compliant, and 
	(ii) $I(w) \neq \emptyset$.
	Then, $w \in \pref(\lang(\mathcal{S}))$.
\end{restatable} 
\begin{lemma}[Canonical implementation deadlock freedom]
	Let $\mathcal{S} = (S, \AlphSyncSubscript, T, s_0, F)$ be a protocol satisfying \Characterization and let $\CLTS{T}$ be a canonical implementation for $\mathcal{S}$ such that for all $w \in \AlphAsyncSubscript^*$, if $w$ is a trace of $\CLTS{T}$, then $I(w) \neq \emptyset$. 
	Then, $\CLTS{T}$ is deadlock-free. 
\end{lemma}
Soundness thus follows from the three conditions above. 
\begin{restatable}[Soundness of \Characterization]{lemma}{soundnessThm}
	Let $\mathcal{S}$
	
	be a protocol.
	If $\mathcal{S}$ satisfies \Characterization, then $\mathcal{S}$ is implementable. 
\end{restatable}
\subsection{Completeness}
Completeness requires us to show that if a protocol is implementable, then it satisfies \Characterization. 
We prove completeness by modus tollens, and assume that a protocol $\mathcal{S}$ does not satisfy \Characterization. 
We negate SC, RC and NMC in turn: from the negation of SC we obtain a simultaneously reachable pair of states in~$\mathcal{S}$ such that a send event that is enabled in one is never enabled from the other. From the negation of RC there exists a simultaneously reachable pair of states in $\mathcal{S}$ such that a receive event that is enabled in one is also enabled in the other. From the negation of NMC we obtain a simultaneously reachable pair of transitions where a participant is the sender in one and the receiver in the other. 
We assume an arbitrary CLTS that implements $\mathcal{S}$, and use these witnesses to show that this CLTS must recognize a trace that is not a prefix in $\lang(\mathcal{S})$, thereby either violating protocol fidelity or deadlock freedom.
\begin{restatable}[Completeness]{lemma}{completenessThm}
	Let $\mathcal{S}$ be a protocol. 
	If $\mathcal{S}$ is implementable, then $\mathcal{S}$ satisfies \Characterization.
\end{restatable}
An immediate consequence of the soundness and completeness of \Characterization is the following fact about the special case of binary protocols, when $\card{\Procs} = 2$: 
\begin{lemma}
	Every binary protocol is implementable. 
\end{lemma}
In the binary case, participant-based reachability is equivalent to standard reachability, because both participants are involved in every synchronous communication. Because protocols are deterministic, there exist no two distinct states in a binary protocol that are simultaneously reachable for either participant, and thus \Characterization holds vacuously. 

\subsection{Synthesis} 

When proving soundness, we chose the canonical implementation as our witness to implementability. 
In other words, if a protocol satisfies \Characterization, then the canonical implementation implements it. 
When proving completeness, we showed that \emph{any} implementation would cause a violation to protocol fidelity or deadlock-freedom. 
In other words, if a protocol violates \Characterization, then no implementation exists. 
Having established that \Characterization precisely characterizes implementable protocols, we combine these observations to yield the following corollary: 
\begin{corollary}[Canonical implementation is all you need]
	A protocol is implementable if and only if the canonical implementation implements it. 
\end{corollary}
For an implementable protocol, this fact serves as a criterion for synthesizing implementations: any implementation that is canonical will suffice. 
For the general class of protocols, synthesis is undecidable. 
However, for many expressive fragments of protocols that still feature infinite data, \eg corresponding to symbolic finite automata~\cite{DBLP:conf/cav/DAntoniV17, DBLP:journals/corr/abs-2303-00924} and certain classes of timed and register automata~\cite{DBLP:journals/fmsd/BertrandSJK15, DBLP:journals/lmcs/ClementeLP22}, one can simply use off-the-shelf determinization algorithms to compute canonical implementations~\cite{DBLP:conf/icst/VeanesHT10, DBLP:conf/tacas/VeanesB12,  DBLP:journals/jlp/BertrandBBC18}. 
\section{Checking Implementability}
\label{sec:symbolic}
Having established that \Characterization is precise for protocol implementability, we next present sound and relatively complete algorithms to check \Characterization for several protocol classes. We start with the most general case of symbolic protocols before considering decidable classes of finite-state~protocols.
\subsection{Symbolic Protocols}
In the remainder of the section, we fix a symbolic protocol $\SymProt = (S, R, \Delta, s_{0}, \rho_0, F)$. 
We assume that the concretization of $\SymProt$ is a GCLTS (\cref{def:gclts}). 
Additionally, we define two copies of the symbolic protocol, denoted $\SymProt_1$ and $\SymProt_2$ that we will use in describing our symbolic implementability check. 
Each copy $\SymProt_i = (R_i, S, \Delta_i, \rho_i, s_0, F)$ with $i \in \set{1,2}$ is obtained from~$\SymProt$ by renaming each register $r$ to a fresh register $r_i$, each unique communication variable $x$ to $x_i$, and substituting the new register and communication variables into the transition constraints and initial register assignment accordingly; the control states remain the~same.

Because symbolic protocols describe concrete protocols with infinitely many states and transitions, implementability cannot be checked explicitly using our \Characterization characterization for protocols, \ie by iterating over all states and transitions.
Instead, we present symbolic conditions that are valid on the symbolic protocol if and only if its concrete protocol is~implementable. 
\begin{theorem}[Symbolic Implementability]
	$\SymProt$ is implementable if and only if it satisfies Symbolic Send Coherence, Symbolic Receive Coherence, and Symbolic No Mixed Choice.  
\end{theorem}

We now present these symbolic conditions, starting with Symbolic Send Coherence.

Send Coherence first requires us to characterize pairs of states in a protocol that are simultaneously reachable for each participant on some prefix in its local language.
In the symbolic setting, this amounts to the following: given a participant and a pair of control states $(s_1,s_2)$ in the symbolic protocol, characterize the register assignments for pairs of concrete states $(s_1,\rho_1)$, $(s_2,\rho_2)$ that are in the respective control states.
The predicate $\prodreach_{\procA}(s_1,\boldsymbol{r_1},s_2,\boldsymbol{r_2})$ describes this for each $\procA \in \Procs$ where $\boldsymbol{r_i}$ are vectors of the registers in $R_i$ obtained by ordering them according to some fixed total order. We define this predicate as a least fixpoint as follows.
\begin{definition}[Simultaneous reachability in product symbolic protocol]
	\label{def:prod-reach}
	Let $\procA \in \Procs$ be a participant and let $s_1, s_1', s_2, s_2' \in S$.
	Then, \fs{Added closing parentheses below. I'm pretty sure that's correct but remove if not please.}
	{ \small
	\begin{align*}
		&
		\prodreach_{\procA}(s_1', \boldsymbol{r_1'}, s_2', \boldsymbol{r_2'}) \is_\mu 
		(\,
		s_1' = s_0 \land s_2' = s_0 \land \boldsymbol{r_1'} = \rho_0 \land \boldsymbol{r_2'} = \rho_0
		\,)
		\\
		& \qquad
		\lor 
		(
			\bigvee_{\substack{
				(s_1, \,\msgFromToNS{\procC}{\procD}{x_1 \set{\varphi_1}}, \,s_1') \in \Delta_1 \\
				(s_2, \,\msgFromToNS{\procC}{\procD}{x_2 \set{\varphi_2}}, \,s_2') \in \Delta_2 \\
				\procA = \procC \lor \procA = \procD}}
		\hspace{-2ex}\exists x_1 x_2 \boldsymbol{r_1} \boldsymbol{r_2}.\, \prodreach_{\procA}(s_1, \boldsymbol{r_1}, s_2, \boldsymbol{r_2})
		\land 
		\varphi_1
		\land 
		\varphi_2
		\land 
		x_1 = x_2
		\,)
				\\
				& \qquad
				\lor 
				(
				\bigvee_{\substack{
						(s_1, \,\msgFromToNS{\procC}{\procD}{x_1 \set{\varphi_1}}, \,s_1') \in \Delta_1  
						\, \land \,
						\procA \neq \procC \land \procA \neq \procD}}
				\hspace{-4ex}\exists x_1 \boldsymbol{r_1}.\, \prodreach_{\procA}(s_1, \boldsymbol{r_1}, s_2', \boldsymbol{r_2'})
				\land 
				\varphi_1
				\,)
				\\
				& \qquad
				\lor 
				(
				\bigvee_{\substack{
						(s_2, \,\msgFromToNS{\procC}{\procD}{x_2 \set{\varphi_2}}, \,s_2') \in \Delta_2 
						\, \land \,
						\procA \neq \procC \land \procA \neq \procD}}
				\hspace{-4ex}\exists x_2 \boldsymbol{r_2}.\, \prodreach_{\procA}(s_1', \boldsymbol{r_1'}, s_2, \boldsymbol{r_2})
				\land 
				\varphi_2
				\,)
				\enspace.
	\end{align*}
	}
\end{definition}
The second top-level disjunct in the definition after the base case handles the cases where $\SymProt_1$ and $\SymProt_2$ synchronize on a common action involving $\procA$. The remaining two disjuncts correspond to the cases where either $\SymProt_1$ or $\SymProt_2$ follows an $\varepsilon$ transition.
Given a pair of simultaneously reachable states $(s_1,\rho_1)$, $(s_2,\rho_2)$ in $\procA$, Send Coherence now checks whether all values $x_1$ that can be sent to some $\procB$ in $(s_1,\rho_1)$ can also be sent from $(s_2,\rho_2)$, modulo following $\varepsilon$ transitions to reach the actual state where $\procA$ can send to $\procB$. We thus need to express $\varepsilon$-reachability. We formalize the dual: the predicate $\unreach^\emptystring_{\procA,\procB}(s_2, \boldsymbol{r_2}, x_1)$ expresses that $\procA$ \emph{cannot} reach any state where it may send $x_1$ to $\procB$, by following $\varepsilon$ transitions from symbolic state $(s_2,\boldsymbol{r_2})$. This is formulated as a greatest fixpoint as follows: 
\begin{definition}[$\varepsilon$-unreachability of $\procA\!$ sending $x$ to $\procB$]
	\label{def:epsilon-unreach-pair}
	For $\procA, \procB \in \Procs$ and $s \in S$, let
	{ \small
	\begin{align*}
		&\unreach^\emptystring_{\procA,\procB}(s, \boldsymbol{r}, x) \is_\nu 
		&
		( \bigwedge_{\substack{
				(s, \,\msgFromToNS{\procA}{\procB}{y \set{\varphi}}, \,s') \in \Delta}}
		\hspace{-3ex} \neg \varphi[x/y]\,) 
		&\land ( \bigwedge_{\substack{
				(s, \,\msgFromToNS{\procC}{\procE}{y \set{\varphi}}, \,s') \in \Delta \\
				\procA \neq \procC \land \procA \neq \procE}}
		\hspace{-3ex} \forall y~\boldsymbol{r'}.\, \varphi \Rightarrow \unreach^\emptystring_{\procA,\procB}(s', \boldsymbol{r'}, x) \,) \enspace.
	\end{align*}
	}
\end{definition}
The first conjunct checks that whenever $\procA$ reaches a state with an outgoing send transition to $\procB$, it cannot send the value $x$ because the transition constraint $\varphi$ is not satisfied. The second conjunct checks that every outgoing $\varepsilon$ transition is either disabled ($\neg \varphi$ holds) or following the transition does not reach an appropriate send state.

We combine the auxiliary predicates into our Symbolic Send Coherence condition. 
\begin{definition}[Symbolic Send Coherence] 
	\label{cond:sym-send-coherence}
	A symbolic protocol $\SymProt$  satisfies Symbolic Send Coherence when for each transition $s_1 \xrightarrow{\msgFromToNS{\procA}{\procB}{x_1 \set{\varphi_1}}} s_1' \in \Delta_1$ and state  $s_2 \in S$, the following is valid:
	{ \small
	\[
		\prodreach_\procA(s_1, \boldsymbol{r_1}, s_2, \boldsymbol{r_2})~\land~\varphi_1
		\land 
		\unreach^\emptystring_{\procA,\procB}(s_2, \boldsymbol{r_2}, x_1)
		\implies 
		\bot \enspace.
	\]
	}
\end{definition}
A keen reader may have noticed that because the symbolic characterization of Send Coherence involves a greatest fixpoint, it is a liveness property. Thus, proving Send Coherence, in general, involves a termination argument. To see this, consider the two protocols shown in \cref{fig:send-coherence-reachability-important-1,fig:send-coherence-reachability-important-2}. Consider the pair of states $(q_1,[c\mapsto0])$ and $(q_3,[c\mapsto0])$ which are simultaneously reachable for $\procC$ in both protocols. The send transition for $\procC$ enabled in $q_1$ needs to be matched with a corresponding send transition in an $\emptystring$-reachable state from $q_3$. The only candidate states for this match in both protocols are those at control state $q_4$. These states are reachable from $q_3$ if and only if the loop in $q_3$ terminates, which it does in \cref{fig:send-coherence-reachability-important-1} but not in \cref{fig:send-coherence-reachability-important-2}. 
\begin{figure}
  \begin{minipage}[b]{.45\textwidth}
\resizebox{0.95\textwidth}{!}{
\begin{tikzpicture}[sem, node distance=2em and 1em]
    \node[state, initial above, initial text = ] (q0) {$q_{0}$};
    \node[state, below left = of q0] (q1) {$q_{1}$};
    \node[state, accepting, below = of q1] (q2) {$q_{2}$};
    \node[state, below right = of q0] (q3) {$q_{3}$};
    \node[state, below = of q3] (q4) {$q_{4}$};
    \node[state, accepting, below = of q4] (q5) {$q_{5}$};
    \path (q0) edge node[left,yshift=.9ex] {$\msgFromTo{\procA}{\procB}{b
    \left\{ \hspace{-1.5ex}
    \begin{array}{c}
    b=0 \\
    \land \, c' = 0
    \end{array}
    \hspace{-1.5ex} \right\}
    }$} (q1);
    \path (q1) edge node[left] {$\msgFromTo{\procC}{\procA}{x \cond{\top}}$} (q2);
    \path (q0) edge node[right,yshift=.9ex] {$\msgFromTo{\procA}{\procB}{b
    \left\{ \hspace{-1.5ex}
    \begin{array}{c}
    b > 0 \\
    \land \, c' = 0
    \end{array}
    \hspace{-1.5ex} \right\}
    }$} (q3);
    \path (q3) edge[loop right] node[right] {$\msgFromTo{\procA}{\procB}{c
    \left\{ \hspace{-1.5ex}
    \begin{array}{c}
    c < b \\
    \land \, c' = c+1
    \end{array}
    \hspace{-1.5ex} \right\}
    }$
    } (q3);
    \path (q3) edge node[right] {$\msgFromTo{\procA}{\procB}{\mathsf{exit} \cond{c \geq b}}$} (q4);
    \path (q4) edge node[right] {$\msgFromTo{\procC}{\procA}{x \cond{\top}}$} (q5);
\end{tikzpicture}
  }
 \caption{
 Example where states $q_1$ and $q_3$ satisfy \\ Send Coherence for $\procC$.
 \label{fig:send-coherence-reachability-important-1}
 }
\end{minipage}
\begin{minipage}[b]{0.45\textwidth}
\resizebox{0.95\textwidth}{!}{
\begin{tikzpicture}[sem, node distance=2em and 1em]
    \node[state, initial above, initial text = ] (q0) {$q_{0}$};
    \node[state, below left = of q0] (q1) {$q_{1}$};
    \node[state, accepting, below = of q1] (q2) {$q_{2}$};
    \node[state, below right = of q0] (q3) {$q_{3}$};
    \node[state, below = of q3] (q4) {$q_{4}$};
    \node[state, accepting, below = of q4] (q5) {$q_{5}$};
    \path (q0) edge node[left,yshift=.9ex] {$\msgFromTo{\procA}{\procB}{b
    \left\{ \hspace{-1.5ex}
    \begin{array}{c}
    b=0 \\
    \land \, c' = 0
    \end{array}
    \hspace{-1.5ex} \right\}
    }$} (q1);
    \path (q1) edge node[left] {$\msgFromTo{\procC}{\procA}{x \cond{\top}}$} (q2);
    \path (q0) edge node[right,yshift=.9ex] {$\msgFromTo{\procA}{\procB}{b
    \left\{ \hspace{-1.5ex}
    \begin{array}{c}
    b > 0 \\
    \land \, c' = 0
    \end{array}
    \hspace{-1.5ex} \right\}
    }$} (q3);
    \path (q3) edge[loop right] node[right] {$\msgFromTo{\procA}{\procB}{c
    \left\{ \hspace{-1.5ex}
    \begin{array}{c}
    c \geq 0 \\
    \land \, c' = c+1
    \end{array}
    \hspace{-1.5ex} \right\}
    }$
    } (q3);
    \path (q3) edge node[right] {$\msgFromTo{\procA}{\procB}{\mathsf{exit} \cond{c < 0}}$} (q4);
    \path (q4) edge node[right] {$\msgFromTo{\procC}{\procA}{x \cond{\top}}$} (q5);
\end{tikzpicture}
  }
 \caption{
 Example where states $q_1$ and $q_3$ violate \\ Send Coherence for $\procC$.
 \label{fig:send-coherence-reachability-important-2}
 }
\end{minipage}
\end{figure}

Receive Coherence is conditioned on two simultaneously reachable states  $(s_1,\boldsymbol{r_1})$ and $(s_2,\boldsymbol{r_2})$ for a participant $\procB$. It checks that if $\procB$ can receive $x$ from $\procA$ in the first state, $\procB$ cannot also receive $x$ as the first message from $\procA$ in the second state, in which it can also receive from a different participant $\procC$, unless $\procA$ sending $x$ causally depends on $\procB$ first receiving from $\procC$. We thus need to define a predicate that captures whether $x_1$ may be available as the first message from $\procB$ to $\procA$, while tracking causal dependencies.
We introduce a family of predicates $\avail_{\procA, \procB, \blockedset}(x_1, s_2, \boldsymbol{r_2})$ for this purpose. Here, $\blockedset$ is used to track the causal dependencies. $\blockedset$ tracks the set of participants that are blocked from sending a message because their send action causally depends on $\procB$ first receiving from~$\procC$. The predicates are defined as the least fixpoint of the following mutually recursive definition. \fs{Added $\lor$ and parentheses in next definition.}
\begin{definition}[Symbolic Availability]
	\label{def:sym-availability}
	\begin{align*}
	\small
		\avail_{\procA, \procB, \blockedset}(x_1, s, \boldsymbol{r})
		\is_\mu 
		&\phantom{\lor j}(
		\bigvee_{\substack{(s, \,\msgFromToNS{\procC}{\procE}{x \set{\varphi}}, \,s') \in \Delta \\ \procC \in \blockedset \\ \procC \neq \procA \lor \procE \neq \procB}}
		\hspace{-2ex}
		\exists x~\boldsymbol{r'}.\, \avail_{\procA, \procB, \blockedset \cup \set{\procE}} (x_1, s', \boldsymbol{r'})
		\land 
		\varphi
		\, )  \\
		& \lor(
		\bigvee_{\substack{(s, \,\msgFromToNS{\procC}{\procE}{x \set{\varphi}}, \,s') \in \Delta \\ \procC \notin \blockedset \\ \procC \neq \procA \lor \procE \neq \procB}}
		\hspace{-2ex}
		\exists x~\boldsymbol{r'}.\, \avail_{\procA, \procB, \blockedset}(x_1, s', \boldsymbol{r'})
		\land 
		\varphi
		\, ) 
		\lor(
		\bigvee_{\substack{(s, \,\msgFromToNS{\procA}{\procB}{x \set{\varphi}}, \,s') \in \Delta \\ \procA \notin \blockedset}}
		\hspace{-2ex}
		 \varphi[x_1/x]
		 \, ) \enspace.
	\end{align*}
\end{definition}
The last disjunct in the definition handles the cases where the message $x_1$ from $\procA$ is immediately available to be received by $\procB$ in symbolic state $(s,\boldsymbol{r})$ and $\procA$ has not been blocked from sending. The other two disjuncts handle the cases when $x_1$ becomes available after some other message exchange between $\procC$ and $\procE$. Here, if $\procC$ is blocked, then $\procE$ also becomes blocked since it depends on $\procC$ sending before it can receive (the first disjunct). Otherwise, no participant is added to the blocked set (the second disjunct).

With the available message predicate in place, we can now define Symbolic Receive Coherence.
\begin{definition}[Symbolic Receive Coherence] 
	\label{cond:sym-receive-coherence}
	A symbolic protocol $\SymProt$ satisfies Symbolic Receive Coherence when for every pair of transitions $s_1 \xrightarrow{\msgFromToNS{\procA}{\procB}{x_1 \set{\varphi_1}}} s_1' \in \Delta_1
	$ and $
	s_2 \xrightarrow{\msgFromToNS{\procC}{\procB}{x_2 \set{\varphi_2}}} s_2' \in \Delta_2$ with $\procA \neq \procC$:
	\[
	\prodreach_\procB(s_1, \boldsymbol{r_1}, s_2, \boldsymbol{r_2})~\land~\varphi_1~\land~\varphi_2
	~\land~
	\avail_{\procA, \procB, \set{\procB}}(x_1, s_2', \boldsymbol{r_2'})
	\implies 
	\bot \enspace.
	\]
\end{definition}
Finally, No Mixed Choice is conditioned on two simultaneously reachable states $(s_1,\boldsymbol{r_1})$ and $(s_2,\boldsymbol{r_2})$ with outgoing send and receive transitions for a participant $\procA$. 
\begin{definition}[Symbolic No Mixed Choice] 
	\label{cond:sym-no-mixed-choice}
	A symbolic protocol $\SymProt$  satisfies Symbolic No Mixed Choice when for every pair of transitions $s_1 \xrightarrow{\msgFromToNS{\procA}{\procB}{x_1 \set{\varphi_1}}} s_1' \in \Delta_1
	$ and $
	s_2 \xrightarrow{\msgFromToNS{\procC}{\procA}{x_2 \set{\varphi_2}}} s_2' \in \Delta_2$:
	\[
	\prodreach_\procA(s_1, \boldsymbol{r_1}, s_2, \boldsymbol{r_2})~\land~\varphi_1~\land~\varphi_2
	\implies 
	\bot \enspace.
	\]
\end{definition}
We conclude this section with a discussion of how to check GCLTS assumptions, namely sink finality, sender-driven choice, and deadlock-freedom, on a symbolic protocol. 
Sink finality can be checked directly by examining the syntax of the symbolic protocol. 
Sender-driven choice without determinism can likewise be checked directly on the states of the symbolic protocol. 
Determinism and deadlock-freedom are undecidable in general but can both be reduced to reachability. 
Thus, both our Symbolic Coherence Conditions and GCLTS assumptions can be discharged using off-the-shelf $\mu$CLP solvers. We leave such an implementation to future work. 

We next apply our framework to decidable fragments of symbolic protocols, some of which have previously been studied in the literature.
\subsection{Finite Protocols}
We first consider finite protocols. Let $\mathcal{S} = (S, \AlphSyncSubscript,  T, s_0, F)$ be a protocol with finite $S$ and $T$. Because $S$ and $T$ are finite, we can transform \Characterization into an imperative algorithm (see \cref{alg:characterization-for-global-types}) and use it to check implementability directly. For checking Receive Coherence, we need to decide the predicate $\avail_{\procA,\procB,\{\procB\}}(\val,s)$, which is defined like the symbolic availability predicate $\avail_{\procA,\procB,\{\procB\}}(x,s,\boldsymbol{r})$, except on protocols instead of symbolic protocols. 
\begin{algorithm}[t]
	\caption{Check \Characterization for finite protocols	\label{alg:characterization-finite}\label{alg:characterization-for-global-types}}
	{ \footnotesize
	\begin{algorithmic}
		\CommentLine{Let LTS $\mathcal{S} = (S, \AlphSyncSubscript, T, s_0, F)$}
		\CommentLine{Checking Send Coherence}
		\For{$s_1 \xrightarrow{\msgFromToNS{\procA}{\procB}{\val}} s_2 \in T$}
			\For{$s \neq s_1 \in S$} 
			\If{$\lang(S, \Gamma_\procA \dunion \set{\emptystring}, T_\procA, s_0, \set{s}) \inters \lang(S, \Gamma_\procA \dunion \set{\emptystring}, T_\procA, s_0, \set{s_1}) \neq \emptyset$}
				\State $b \gets \bot$
				\For{$s_3 \xrightarrow{\msgFromToNS{\procA}{\procB}{\val}} s_4 \in T$}
          
          $b \gets b \lor \left( s \xRightarrow[\procA]{\emptystring}\Kleenestar s_3 \right) $
				\EndFor
				\If{$\neg b$} $\textbf{return}~\bot$
				\EndIf
			\EndIf
			\EndFor
		\EndFor
		
		\CommentLine{Checking Receive Coherence}
		\For{$s_1 \xrightarrow{\msgFromToNS{\procA}{\procB}{\val}} s_2,  s_3 \xrightarrow{\msgFromToNS{\procC}{\procB}{\val}} s_4 \in T, s_1 \neq s_2, \procA \neq \procC$}
			\If{$\lang(S, \Gamma_\procB \dunion \set{\emptystring}, T_\procB, s_0, \set{s_1}) \inters \lang(S, \Gamma_\procB \dunion \set{\emptystring}, T_\procB, s_0, \set{s_3}) \neq \emptyset$}
        \If{$\avail_{\procA,\procB,\set{\procB}}(\val, s_4)$} $\textbf{return}~\bot$
				\EndIf
			\EndIf
		\EndFor
		
		\CommentLine{Checking No Mixed Choice}
		\For{$s_1 \xrightarrow{\msgFromToNS{\procA}{\procB}{\val}} s_2,  s_3 \xrightarrow{\msgFromToNS{\procC}{\procA}{\val}} s_4 \in T, s_1 \neq s_2$}
		\If{$\lang(S, \Gamma_\procB \dunion \set{\emptystring}, T_\procB, s_0, \set{s_1}) \inters \lang(S, \Gamma_\procB \dunion \set{\emptystring}, T_\procB, s_0, \set{s_3}) \neq \emptyset$} $\textbf{return}~\bot$
		\EndIf
		\EndFor
		\State $\textbf{return}~\top$
	\end{algorithmic}
	}
\end{algorithm}

It is easy to see that Send Coherence and No Mixed Choice can be checked in time polynomial in the size of $\mathcal{S}$. However, the inclusion of $\avail_{\procA,\procB,\{\procB\}}(\val,s)$ as a subroutine for checking Receive Coherence yields the following complexity result. 
\newcommand{\nextMsg}{\textit{next}}
\newcommand{\lastMsg}{\textit{last}}
\newcommand{\topMsg}{\textit{top}}
\begin{theorem}
	Implementability of finite protocols is co-NP-complete.
\end{theorem}
\begin{proof}
	To see that implementability is in co-NP, observe that violations of Send Coherence and No Mixed Choice can be checked in NP, by guessing a participant $\procA$ and a pair of states $s_1, s_2$ that satisfy the respective preconditions, and verifying simultaneous reachability of $s_1$ and $s_2$ for $\procA$. For Send Coherence, we guess an additional state $s_3$ with an outgoing transition labeled with $\msgFromToNS{\procA}{\procB}{\val}$, and check $\epsilon$-reachability from $s_1$ to $s_3$. 
	For Receive Coherence, $\avail_{\procA,\procB,\{\procB\}}(\val,s_2)$ can be checked in NP by guessing a simple path in $\mathcal{S}$ from $s_2$ to some state $s'$ with an outgoing transition labeled with $\msgFromToNS{\procA}{\procB}{\val}$. We then evaluate $\avail_{\procA,\procB,\{\procB\}}(\val,s_2)$ along that path, which can be done in polynomial time. We can restrict ourselves to simple paths because the blocked set $\blockedset$ monotonically increases when traversing a path in $\mathcal{S}$. Moreover, $\avail_{\procA,\procB,\{\procB\}}(\val,s_2)$ is antitone in the blocked set.
	
  We show NP-hardness of non-implementability via a reduction from the 3-SAT problem.
  Assume a 3-SAT instance $\varphi = C_1 \land \ldots \land C_k$. Let $x_1,\dots,x_n$ be the variables occurring in $\varphi$ and let $L_{ij}$ be the $j$th literal of clause~$C_i$, with $1 \leq i \leq k$ and $1 \leq j \leq 3$. 
  
  We construct a protocol $\mathcal{S}_\varphi$ over participants $\Procs = \{\procA,\procB, \procC, \roleFmt{x_1},\roleFmt{\overline{x}_1}, \dots,\roleFmt{x_n},\roleFmt{\overline{x}_n}\}$, such that $\varphi$ is satisfiable iff $\mathcal{S}_\varphi$ is implementable. 
  In particular, we ensure that $\mathcal{S}_\varphi$ is implementable iff $\avail_{\procA,\procB,\{\procB\}}(\val,s)$ does not hold for some state $s$ in $\mathcal{S}_\varphi$. 
  The protocol $\mathcal{S}_\varphi$ is constructed from the following subprotocols: 
  \begin{enumerate}
  \item Define a protocol $\mathcal{S}_X$ representing a truth assignment to variables $x_i$  with states $s_1,\dots,s_{n+1}$ as follows: for every $1 \leq i \leq n$ there are two paths of four transitions each between $s_i$ and $s_{i+1}$. The paths consist of transitions labeled with $\msgFromToNS{\procC}{\roleFmt{x_i}}{\bot}$, $\msgFromToNS{\procC}{\roleFmt{\overline{x}_i}}{\top}$, $\msgFromToNS{\procC}{\procB}{\val_{x_i}}$, $\msgFromToNS{\procB}{\roleFmt{x_i}}{\val}$, 
  and 
 $\msgFromToNS{\procC}{\roleFmt{\overline{x}_i}}{\bot}$, $\msgFromToNS{\procC}{\roleFmt{x_i}}{\top}$, $\msgFromToNS{\procC}{\procB}{\val_{\overline{x}_i}}$, $\msgFromToNS{\procB}{\roleFmt{\overline{x}_i}}{\val}$, 
  respectively.
  \item Define a protocol $\mathcal{S}_C$ representing the clauses $C_i$ with states $t_{1},\dots,t_{k+1}$ as follows. For each $1 \leq i \leq k$ there are three paths of three transitions between each $t_i$ and $t_{i+1}$, one for each $1 \leq j \leq 3$, labeled with 
  $\msgFromToNS{\procC}{s}{\val_j}$, $\msgFromToNS{\procC}{\procA}{\val_r}$, $\msgFromToNS{s}{\procA}{\val}$, where $s=\roleFmt{x}$ if $L_{ij}=x$ and $s=\roleFmt{\overline{x}}$ if $L_{ij} = \neg x$ for $x \in \{x_1,\dots,x_n\}$. 
 
  \item Define a protocol $\mathcal{S}_F$ with two states $q_f'$ and $q_f$ and a single transition from $q_f'$ to $q_f$ labeled with $\msgFromToNS{\procA}{\procB}{\val}$.
  \item Define a protocol $\mathcal{S}_T$ with five states $q_1, \ldots, q_5$, and two paths from $q_1$, respectively
  $q_1 \xrightarrow{\msgFromToNS{\procC}{\procA}{\val_1}} q_2 \xrightarrow{\msgFromToNS{\procC}{\procB}{\val}} q_3$ and 
  $q_1 \xrightarrow{\msgFromToNS{\procC}{\procA}{\val_2}} q_4 \xrightarrow{\msgFromToNS{\procA}{\procB}{\val}} q_5$. 
  \end{enumerate}
  
   We merge all of the above protocols to obtain $\mathcal{S}_\varphi$ by identifying the state $q_3$ with $s_1$, $s_{n+1}$ with $t_{1}$ and $t_{k+1}$ with $q_f'$. The initial state is $q_1$ and the final states are $\set{q_5, q_f}$. 
    
  Observe that the size of $\mathcal{S}_\varphi$ is linear in the size of $\varphi$. Moreover, it is easy to check that $\mathcal{S}_\varphi$ is indeed a GCLTS: 
  all choices are sender-driven and deterministic, and 
  final states are the only states with no outgoing transitions, 
  yielding sink-finality and deadlock-freedom. 
  
  We first establish that $\avail_{\procA,\procB,\{\procB\}}(\val,q_3)$ holds in $\mathcal{S}_\varphi$ iff $\varphi$ is satisfiable. Observe that the blocked set $\blockedset$ computed by $\avail_{\procA,\procB,\{\procB\}}(\val,q_3)$ along a path between $s_1$ and $s_{n+1}$ contains for each variable $x_i$ either $\roleFmt{x_i}$ or $\roleFmt{\overline{x}_i}$.
  The blocked set $\blockedset$ thus encodes a truth assignment $\rho_\blockedset$ for the $x_i$'s where $\rho_\blockedset(x_i)=\top$ iff $\roleFmt{x_i} \not\in \blockedset$.
  By construction of $\mathcal{S}_X$, for every truth assignment $\rho$, there exists a path between $s_1$ and $s_{n+1}$ such that $\rho=\rho_\blockedset$ for the blocked set $\blockedset$ computed along that path.
  The paths between states $t_{i}$ and $t_{i+1}$ in subprotocol $\mathcal{S}_C$ allow $\procA$ to proceed and not be blocked if one of the paths has a participant not in $\blockedset$, i.e. $C_i$ is satisfied by $\rho_\blockedset$.
  Thus, a path from $s_{n+1} = t_{1}$ to $t_{k+1} = q_f'$ adds $\procA$ to $\blockedset$ at $t_i$ iff $\rho_\blockedset$ does not satisfy at least one of the clauses $C_i$.
  Therefore, $\val$ is available in $q_3$ iff there exists a $\blockedset$ such that $\rho_\blockedset$ satisfies $\varphi$.
  
  It remains to show that $\mathcal{S}_\varphi$ is non-implementable iff $\avail_{\procA,\procB,\{\procB\}}(\val,q_3)$ holds in $\mathcal{S}_\varphi$. 
  We argue that all participants except $\procB$ have sufficient local information about the control flow of the protocol to behave accordingly. 
  Participant $\procC$ dictates the control flow at every branching point of the protocol, and thus is implementable. 
  Participants $\roleFmt{x_1}, \roleFmt{\overline{x}_1}, \ldots \roleFmt{x_n}, \roleFmt{\overline{x}_n}$ learn the control flow via receiving messages from participant $\procC$, whose labels uniquely determine their next actions: receiving $\top$ means inaction, receiving $\bot$ means receive a further message from $\procB$, and receiving $m$ means send a message encoding its own variable name to $\procA$. 
  Participant $\procA$ is likewise informed by $\procC$ about the control flow, and only sends $\val$ to $\procB$ upon either receiving $\val_2$ or $\topMsg$ from $\procC$. Upon receiving $\procC$'s choice of disjunct for each clause, it anticipates a message from the participant encoding that disjunct. 
  
  Participant $\procB$, on the other hand, is not informed by $\procC$ about $\procC$'s initial choice at $G_{x_1}$, and can locally choose between receptions from $\procA$ or $\procC$. In the case that $\avail_{\procA,\procB,\{\procB\}}(\val,q_3)$ holds, there exists a path from $\overline{G}$ to $\GG_\varphi$ in which $\procA$ is not blocked. Thus, the message from $\procA$ can be asynchronously reordered to arrive in $\procB$'s channel such that both receptions are enabled, and $\procB$ may violate implementability by receiving the message from $\procA$ out of order.  
  If $\avail_{\procA,\procB,\{\procB\}}(\val,q_3)$ does not hold, only one reception is enabled, which uniquely informs $\procB$ about $\procC$'s choice. In the case that the reception from $\procA$ is enabled, $\procB$ terminates, otherwise it receives messages from $\procC$ encoding participants to send further messages to, and terminates upon receiving the final message from $\procA$. 
  Thus, $\mathcal{S}_\varphi$ is non-implementable iff $\procB$ violates Receive Coherence for the transitions 
  $q_2 \xrightarrow{\msgFromToNS{\procC}{\procB}{\val}} q_3$ and 
  $q_4 \xrightarrow{\msgFromToNS{\procA}{\procB}{\val}} q_5$, i.e.\
  $\avail_{\procA,\procB,\{\procB\}}(\val,q_3)$ does not hold. 
  
  We obtain that $\mathcal{S}_\varphi$ is non-implementable iff $\avail_{\procA,\procB,\{\procB\}}(\val,q_3)$ holds in $\mathcal{S}_\varphi$ iff $\varphi$ is satisfiable.
\end{proof}
Implementability for global multiparty session types was shown in~\cite{DBLP:conf/cav/LiSWZ23} to be in PSPACE, with the matching lower bound corrected in \cite{DBLP:journals/corr/abs-2305-17079}. 
We show that, in fact, the same 3-SAT reduction can be adapted to show co-NP-completeness of implementability for global multiparty session types. 
\paragraph{Global multiparty session types}
 Global types for MSTs 
\cite{DBLP:conf/cav/LiSWZ23} are defined by the grammar:
\begin{grammar}
	G \is
	0
	\mid \sum_{i ∈ I} \msgFromTo{\procA}{\procB_{i}}{\val_i.G_i}
	\mid \mu t. \; G
	\mid t
\end{grammar}\\[-3ex]
where $\procA, \procB_i$ range over $\Procs$, $\val_i$ over a finite set $\MsgVals$, and $t$ over a set of recursion variables. 
The semantics of a global type $\GG$ are defined using a finite state machine
$\semglobalsync(\GG) = (Q_{\GG}, \AlphSync_{sync} \uplus \set{\emptystring}, δ_{\GG}, q_{0, \GG}, F_{\GG})$
where 
	$Q_{\GG}$ is the set of all syntactic subterms in $\GG$ together with the term $0$,
	$δ_{\GG}$ is the smallest set containing
	$(\sum_{i ∈ I} \msgFromTo{\procA}{\procB_i}{\val_i.G_i}, \msgFromTo{\procA}{\procB_i}{\val_i}, G_i)$ for each $i ∈ I$,
	as well as $(μ t. G', ε, G')$ and $(t, ε, μ t. G')$ for each subterm~$\mu t.G'$, 
	$q_{0, \GG} = \GG$ and
	$F_{\GG} = \set{0}$.
	
Each branch of a choice is assumed to be distinct: 
$∀ i,j ∈ I.\, i≠j ⇒ (\procB_{i},\val_i) ≠ (\procB_{j},\val_j)$, 
and the sender and receiver of an atomic action is assumed to be distinct: 
$∀ i ∈ I.\, \procA ≠ \procB_i$. 
Recursion is guarded: 
in $μ t. \, G$, there is at least one message between $μt$ and each $t$ in $G$.

Each $\emptystring$-transition in $\semglobal(\GG)$ is the only transition from the state it originates from.
This makes removing them easy, yielding a protocol $\mathcal{S_\GG} = (Q_\GG, \AlphSyncSubscript, \delta'_\GG, q_{0,\GG}, F_\GG)$, where $\delta'_\GG$ contains only transitions labeled with $l \in \AlphSyncSubscript$.
It is easy to verify that $\mathcal{S_\GG}$ is indeed a GCLTS. 
\begin{restatable}[]{lemma}{globalTypeImplementability}								
	Implementability of global types is co-NP-complete.  
\end{restatable} 
The proof can be found in \cref{app:symbolic}. 
\subsection{Symbolic Finite Protocols} 
Finally, we study symbolic representations of finite protocols. More precisely, we consider the fragment of symbolic protocols where $\MsgVals$ is the set of Booleans and all transition constraints~$\varphi$ are given by propositional formulas. 
We show that for this class of symbolic protocols, the implementability problem is PSPACE-complete. 
\begin{theorem}
	\label{thm:bitvector-global-types-pspace-complete}
	Implementability of symbolic finite protocols is PSPACE-complete. 
\end{theorem}
\begin{proof}[Proof sketch]
To show that implementability is in PSPACE, we show that a witness to the negation of \Characterization can be checked in nondeterministic polynomial space. This follows by a reduction to the reachability problem for extended finite state machines, which is in PSPACE~\cite{DBLP:conf/tacas/GodefroidY13}. By Savitch's Theorem, it follows that the negation of \Characterization is in PSPACE. Because PSPACE is closed under complement and \Characterization precisely characterizes implementability, it follows that implementability is in~PSPACE. 	
 
We show PSPACE-hardness of the implementability problem by a reduction from the PSPACE-hard problem of deciding reachability  for 1-safe Petri nets~\cite{DBLP:journals/eik/EsparzaN94}.
	Let $(N, M_0)$ be a 1-safe Petri net, with $N = (S, T, F)$. 
	
	Let $M$ be a marking of $N$. 
        We construct a symbolic protocol that is implementable iff $N$ does not reach $M$. For ease of exposition, we present this symbolic protocol as a symbolic dependent global type $\GG_N$ with the understanding that the encoding of $\GG_N$ as a symbolic protocol is clear.
	
	We first describe the construction of $\GG_N$.
The outermost structure of $\GG_N$ consists of a participant~$\procC$ communicating a choice between two branches to $\procD$
       where the bottom branch solely consists of $\procA$ sending $l$ to $\procB$:
$
               \GG_N \is
				(
                \msgFromTo{\procC}{\procD}{m_1 \set{\top}}.\, G_t 
                +
                \msgFromTo{\procC}{\procD}{m_2 \set{\top}}.\, \msgFromTo{\procA}{\procB}{l \set{\top}}.\, 0 
                )
$.
       Since $\procA$ is not informed about the choice of the branch taken by $\procD$, it will have to be able to match this send transition in every run that follows the continuation $G_t$ of the top branch. We will construct $G_t$ such that this match is possible iff $M$ is reachable in $N$.
	In $G_t$, participants $\procC$ and $\procD$ enter a loop that simulates $N$:
	\[ \small
		G_t \is
			\mu s [v \is M_0]. \, +
			\begin{cases}
				\sum_{t \in T} \,
				\msgFromTo{\procC}{\procD}{m_t \set{v \Rightarrow t^-}}
				.\, s [ v \is ((v \land \neg t^-) \lor t^+) ]
				\\
				\msgFromTo{\procC}{\procD}{\mathit{restart} \set{\top}}. \, s [ v \is M_0 ] \\
				\msgFromTo{\procC}{\procD}{\mathit{reach}_M \set{v = M}}. \,
				\msgFromTo{\procA}{\procB}{l \set{\top}}. \, 0 \\
			\end{cases}
	\]
	The loop variable $v$ is a $\card{S}$-length bitvector that tracks the current marking of the net. It is initialized to $M_0$.
	Inside the loop, $\procC$ has the following choices. First, it may pick any transition $t \in T$ of the net and send an $m_t$ message to $\procD$, provided the transition is enabled for firing (i.e., the input places of $t$ all contain a token: $v \Rightarrow t^-$). After this communication, $v$ is updated according to the fired transition $t$. 
        The last branch of the choice in the loop is enabled if $v$ is equal to $M$. Here, $\procC$ can send $\mathit{reach}_M$ to~$\procD$, which gives $\procA$ the opportunity to send the $l$ message to $\procB$, allowing it to match the send transition from the lower branch in the top level choice of $G_N$.
        Finally, the middle branch allows $\procC$ to abort the simulation at any point and start over. This ensures that if the simulation ever reaches a dead state due to firing a transition that would render $M$ unreachable, it can recover by starting again from $M_0$. Thus, for all states of the simulator, $\procA$ has an $\varepsilon$ path from that state to a state where it can send $l$ to $\procB$ iff $M$ is reachable from $M_0$ in $N$. The only other sender is $\procC$ which makes all choices and, hence, never reaches two different states along the same prefix trace, thus satisfying Send Coherence trivially. It follows that Send Coherence for $\procA$ holds iff $M$ is reachable from $M_0$ in $N$. To see that Receive Coherence holds, observe that no participant receives messages from two different senders. No Mixed Choice similarly holds trivially. 
        $G_N$ is deadlock-free because the branch in the loop of $G_t$ where $\procC$ sends the $\mathit{restart}$ message is always enabled. Moreover, it is easy to see that $G_N$ is deterministic because each branch of a choice sends a different message value.
        In summary, $G_N$ is a GCLTS that is implementable iff $N$ reaches $M$.	The size of $\GG_N$ is linear in the size of $N$, so we obtain the desired reduction.

\end{proof}
\section{Related Work}
\label{sec:related-work}
\begin{table}[b]
	\small
	\centering
	\scalebox{0.9}{
	\begin{tabular}{|c|c|c|c|c|c}
		\hline
                \makecell{\bf Paper} &
		\makecell{\bf Communication\\\bf paradigm} &
		\makecell{\bf Branching\\\bf restrictions} &
		\makecell{\bf History\\\bf sensitivity} &
		\makecell{\bf Characterization} 
		
		\\
		\hline\hline
		\cite{DBLP:conf/concur/BocchiHTY10} &
		asynchronous &
		\makecell{directed choice} &
		required &
		\makecell{incomplete} 
		
		\\
		\hline
		\cite{DBLP:conf/tgc/BocchiDY12} &
		asynchronous &
		\makecell{directed choice} &
		required &
		\makecell{incomplete} 
		
		\\
		\hline
		\cite{DBLP:journals/jlp/ToninhoY17} &
		synchronous &
		\makecell{directed choice} &
		required &
		\makecell{incomplete} 
		
		\\
		\hline
		\cite{DBLP:journals/pacmpl/00020HNY20} &
		synchronous &
		\makecell{directed choice} &
		required &
		\makecell{incomplete} 
		
		\\
		\hline
		\cite{DBLP:conf/ecoop/GheriLSTY22} &
		synchronous &
		\makecell{well-sequencedness} & 
		required &
		\makecell{unknown} 
		
		\\
		\hline
		\makecell{this work} &
		asynchronous &
		\makecell{sender-driven choice} &
		not required &
		\makecell{relatively complete} 
		
		\\
		\hline 
	\end{tabular}
	}
	\vspace{1ex}
	\caption{Comparison of related work (in chronological order)} 
	\vspace{-4ex}
	\label{tab:comparison}
\end{table}
\cref{tab:comparison} summarizes the most closely related works that address the implementability problem of communication protocols with data refinements. 
We discuss these works in terms of key expressive features and completeness of characterization. 
\paragraph{Expressivity.} 
All existing works in \cref{tab:comparison} effectively require \emph{history-sensitivity}, which means that a
``predicate guaranteed by a [participant $\procA$] can only contain those interaction variables that [$\procA$] knows'' \cite{DBLP:conf/concur/BocchiHTY10}, see also \cite[Def.\,2]{DBLP:conf/tgc/BocchiDY12}.
As discussed in \S\ref{sec:characterization}, syntactic approaches to analyzing variable knowledge is overly conservative, and as a result no prior work can handle protocols such as the example in \cref{fig:implementable-local-type}. 
In a similar vein, \citet{DBLP:journals/pacmpl/00020HNY20} impose the syntactic restriction that all participants in a loop must be able to update all loop registers, which rules out loops like the one in the two-bidder protocol (\cref{fig:two-bidder-protocol}).
Furthermore, all prior works except for
\cite{DBLP:conf/ecoop/GheriLSTY22}
employ the directed choice restriction, which is strictly less general than sender-driven choice. \fs{FIX: citet?}
Many of these works also feature separate constructs for selecting branches and sending data.
In our symbolic protocols, this is not necessary because selecting branches can be modeled with equality predicates, as demonstrated by \cref{fig:two-bidder-symbolic}.
\citet{DBLP:conf/ecoop/GheriLSTY22} generalize choreography automata, which are finite-state LTSs with communciation events as transition labels but without final states.
One major difference between our work and theirs lies in the treatment of interleavings.
Unlike our protocol semantics, which are closed under the indistinguishability relation $\interswap$, inspired by Lamport's happened-before relation, choreography automata languages do not include any interleavings not present in the language. 
Setting aside asynchronous traces, the protocol 
$
\msgFromTo{\procA}{\procB}{\val}.\,
\msgFromTo{\procC}{\procD}{\val}.\, 0
$
in our setting would need to be represented as 
$
\msgFromTo{\procA}{\procB}{\val}.\,
\msgFromTo{\procC}{\procD}{\val}.\, 0
\; + \;
\msgFromTo{\procC}{\procD}{\val}.\,
\msgFromTo{\procA}{\procB}{\val}.\, 0
$ 
in their setting, and the protocol 
$
\mu t.\,
\msgFromTo{\procA}{\procB}{\val}.\,
\msgFromTo{\procC}{\procD}{\val}.\,
t
$
does not admit a representation as a choreography automaton. 
The branching behaviors are restricted with a well-sequencedness condition \cite[Def.\,3.2]{DBLP:conf/ecoop/GheriLSTY22}, a condition that has since been refined because it was shown to be flawed \cite{DBLP:journals/lmcs/FinkelL23}. 
\citet{DBLP:journals/corr/abs-2107-03984} showed that well-formedness conditions on synchronous choreography automata do not generalize soundly to the asynchronous~setting. 
Asynchronous communication is more challenging to analyze in general because it easily gives rise to infinite-state systems. 
\citet{NODBLPthesiszhou} conjectures that the framework in
\cite{DBLP:journals/pacmpl/00020HNY20} ``can be extended to support asynchronous communication'', \fs{FIX: citet?}
but does not conjecture if and how the projection operator would change. 
Due to directed choice, the same projection operator may remain sound under asynchronous semantics, because it rules out protocols where participants have a choice to receive from different senders. 
However, it will also likely inherit the same sources of incompleteness present in the synchronous~setting. 
In contrast to all aforementioned works, 
several works~\cite{DBLP:conf/coordination/DagninoGD21,DBLP:journals/corr/abs-2203-12876,DBLP:journals/fuin/CastellaniDG24} allow to specify send and receive events separately with ``deconfined'' global types. 
Deconfined global types are specified as a parallel composition of local processes, and then checked for desirable correctness properties, which were shown to be undecidable \cite{DBLP:conf/coordination/DagninoGD21}.
\paragraph{Completeness.} 
Implementability is a thoroughly-studied problem in the high-level message sequence chart (HMSC) literature. 
HMSCs are a standardized formalism for describing communication protocols in industry~\cite{z120-standard} and are well-studied in academia
~\cite{
	DBLP:conf/sdl/MauwR97,
	DBLP:conf/ac/GenestMP03,
	DBLP:conf/acsd/GenestM05,
	DBLP:conf/concur/GazagnaireGHTY07,
	DBLP:journals/tosem/RoychoudhuryGS12
}.
In the HMSC setting, implementability is called safe realizability, and is defined with respect to the implementation model of  communicating finite state machines \cite{DBLP:journals/jacm/BrandZ83}. 
Similar to our setting, a canonical implementation exists for any HMSC~\cite[Thm.\,13]{DBLP:journals/tse/AlurEY03}; unlike our setting, it is always computable. 
Therefore, existing work has focused less on synthesis and more on checking implementability. 
Despite having only finite states and data, HMSC implementability was shown to be undecidable in general \cite{DBLP:journals/tcs/Lohrey03}. 
Various fragments have since been identified in which the problem regains decidability.
\citet{DBLP:journals/tcs/Lohrey03} showed implementability to be EXPSPACE-complete for bounded HMSCs~\cite{DBLP:conf/concur/AlurY99,DBLP:conf/mfcs/MuschollP99} and globally-cooperative HMSCs ~\cite{DBLP:conf/stacs/Morin02,DBLP:journals/jcss/GenestMSZ06}.
These fragments restrict the communication topology of loops to be strongly and weakly connected respectively.
For HMSCs where every two consecutive communications share a participant, implementability was shown to be PSPACE-complete~\cite{DBLP:journals/tcs/Lohrey03}. 
In contrast, works that study comparably expressive protocol fragments to ours often sidestep the implementability question. 
Instead, implementability is addressed in the form of syntactic well-formedness conditions, as mentioned above, or indirectly through synthesis. 
None of the prior works attempted to show completeness; it was later shown in \cite{DBLP:conf/ecoop/Stutz23,DBLP:conf/cav/LiSWZ23} that all but \citet{DBLP:conf/ecoop/GheriLSTY22} are incomplete. \fs{FIX: citet?}
Several works \cite{DBLP:conf/concur/BocchiHTY10,DBLP:conf/tgc/BocchiDY12,DBLP:journals/jlp/ToninhoY17,DBLP:journals/pacmpl/00020HNY20} synthesize local implementations using the ``classical'' projection from multiparty session types. 
One kind of merge operator, called the plain merge, allows only the two participants in a choice to exhibit different behavior on each branch, a condition which is breached by our two-bidder protocol
(\cref{fig:two-bidder-protocol}).
\citet{DBLP:journals/pacmpl/00020HNY20} proves the soundness of projection with plain merge, but implements a more permissive variant called full merge in the toolchain. 
However, the projected local types are not guaranteed to be implementable:
both \cref{fig:unimplementable-local-type} and \cref{fig:implementable-local-type} are projectable in \cite{DBLP:journals/pacmpl/00020HNY20}. \fs{FIX: citet?}
Thus, the implementability problem is deferred to local types. 
Our results show that synthesis is ``as possible as'' the determinization of the non-deterministic underlying automata fragment. 
This means that implementations can be synthesized even for expressive classes of protocols that correspond to \eg symbolic finite automata~\cite{DBLP:conf/cav/DAntoniV17, DBLP:journals/corr/abs-2303-00924} and certain classes of timed and register automata~\cite{DBLP:journals/fmsd/BertrandSJK15, DBLP:journals/lmcs/ClementeLP22} due to the existence of off-the-shelf determinization algorithms for these classes~\cite{DBLP:conf/icst/VeanesHT10, DBLP:conf/tacas/VeanesB12,  DBLP:journals/jlp/BertrandBBC18}. 
\citet{DBLP:journals/pacmpl/ScalasY19} check safety properties of collections of local types by encoding the properties as $\mu$-calculus formulas and then model checking the typing context against the specification. They focus primarily on finite-state typing contexts under synchronous semantics, and thus all properties in their setting are decidable. For the asynchronous setting, only three sound approximations of safety are proposed, one of which bounds channel sizes and thus falls back into the finite-state setting. 
\smallskip
Next, we discuss further related works on choreographic programming and binary session types. 
\paragraph{Choreographic programming.}
Choreographic programming~\cite{DBLP:journals/tcs/Cruz-FilipeM20,DBLP:conf/ecoop/GiallorenzoMPRS21,DBLP:journals/pacmpl/HirschG22} describes global message-passing behaviors as programs rather than protocols, and therefore incorporate many more programming language features that are abstracted away in our model, such as computation and mutable state, in addition to features that our model cannot express, such as higher-order computations and delegation. 
Endpoint projection for choreographic programs, which shares a theoretical basis with multiparty session type projection, then generates individual, executable programs for each participant. 
The question of implementability, though undecidable in the presence of such expressivity, remains relevant to the soundness of endpoint projections. 
We discuss three approaches to endpoint projection. 
Pirouette~\cite{DBLP:journals/corr/abs-2111-03484} requires the programmer to specify explicit synchronization messages to ensure that ``different locations stay in lock-step with each other'', and conservatively rejects programs that are underspecified in this regard. Pirouette provides a mechanized proof of deadlock freedom for endpoint projections in Coq. 
Note that the claims of soundness and completeness in~\cite{DBLP:journals/corr/abs-2111-03484} are not with respect to implementability, but with respect to the translation via endpoint projection. \fs{FIX: citet?}
HasChor~\cite{DBLP:journals/corr/abs-2303-00924} rules out non-implementability by automatically incorporating location broadcasts when a choice is made. 
No formal correctness claims are made in \cite{DBLP:journals/corr/abs-2303-00924}. \fs{FIX: citet?}
\citet{DBLP:conf/esop/JongmansB22} allow if- and while- statements to be annotated with a conjunction of conditional choices for each participant, which expresses decentralized decision-making in protocols. 
They show that their endpoint projection for well-formed choreographies guarantees deadlock freedom and functional correctness. 
All aforementioned choreographic programming works assume a synchronous network. 
\paragraph{Binary session types with refinements.}
Finally, we briefly mention work on binary session types with refinements and data dependencies.  
In the binary setting, implementability is a less interesting problem due to the inherent duality between the two protocol participants; the distinction between global and local types is no longer meaningful. 
\citet{DBLP:conf/nfm/GriffithG13} refine binary sessions with basic data types, and shows decidability of the subtyping problem. 
\citet{DBLP:conf/esop/GommerstadtJP18} applies a similar type system for runtime monitoring of binary communication.
\citet{DBLP:journals/pacmpl/ThiemannV20} propose a label-dependent binary session type framework which allows the subsequent behavior of the protocol to depend on previous labels, which are drawn from a finite set. 
\citet{DBLP:conf/concur/DasP20} study the undecidable problem of local type equality, and provide a sound approximate algorithm. 
Das et al. \cite{DBLP:journals/lmcs/DasP22, DBLP:conf/csfw/DasB0PS21} further apply binary session types with refinements to resource analysis of blockchain smart contracts and amortized cost analysis. 
Actris~\cite{DBLP:journals/pacmpl/HinrichsenBK20} embeds binary session types into the Iris framework~\cite{DBLP:journals/jfp/JungKJBBD18}. 
The framework assumes asynchronous communication with FIFO channels, and can verify programs that combine message-passing concurrency and shared-memory concurrency. 
Actris has been extended with session type subtyping (Actris 2.0~\cite{DBLP:journals/lmcs/HinrichsenBK22}) and with linearity to prove both preservation and progress (LinearActris~\cite{DBLP:journals/pacmpl/JacobsHK24}).
Multris~\cite{Multris} is an extension of Actris in Iris to the multiparty setting. 
The message-passing layer of Multris is more restricted than Actris: Multris assumes synchronous communication and prohibits choice over channels: choices can only be made about message values between a given sender and receiver.
Multris takes a bottom-up approach~\cite{DBLP:journals/pacmpl/ScalasY19} to correctness: given a collection of local types, the type system checks that they can be safely combined. 
Multris guarantees protocol fidelity but not progress.

	\fs{Camera-ready: Add acks back}
	\subsubsection*{Acknowledgements}
	
	This work is supported in parts by 
	the \grantsponsor{GS100000001}{National Science Foundation}{http://dx.doi.org/10.13039/100000001} under the grant agreement~\grantnum{GS100000001}{2304758} 
	and 
	by the Luxembourg National Research Fund (FNR) under the grant agreement C22/IS/17238244/AVVA. 
	We thank the anonymous OOPSLA reviewers for their comments which improved the paper, and for identifying an erroneous claim in an earlier draft related to the complexity analysis of MST implementability. 
	
	\phantomsection\label{paper-last-page}
	\clearpage

	\bibliography{biblio,dblp}
	\clearpage
	\appendix
\section{Additional Material for Section~\ref{sec:characterization}}
\label{app:characterization}
\noMixedChoice*
\begin{proof}
	Suppose by contradiction that $x_1 \in \Alphabet_?$ and $x_2 \in \Alphabet_!$. 
	Let $\run_1$ be a run in $\mathcal{S}$ such that $wx_1 \leq \SyncToAsync(\trace(\run_1)) \wproj_{\Alphabet_{\procA}}$. 
	Let $\alpha_1 \cdot s_1 \xrightarrow{l_1} s_1' \cdot \beta_1$ be the unique splitting of $\run$ for $\procA$ with respect to $w$. 
	Then, $\procA$ is the receiver in $l_1$ and $\SyncToAsync(\trace(\alpha_1 \cdot s_1)) \wproj_{\Alphabet_{\procA}} = w$. 
	Let $\run_2$ be a run in $\mathcal{S}$ such that $wx_2 \leq \SyncToAsync(\trace(\run_2)) \wproj_{\Alphabet_{\procA}}$.
	Let $\alpha_2 \cdot s_2 \xrightarrow{l_2} s_2' \cdot \beta_2$ be the unique splitting of $\run_2$ for $\procA$ with respect to $w$. 
	Then, $\procA$ is the sender in $l_2$ and $\SyncToAsync(\trace(\alpha_2 \cdot s_2)) \wproj_{\Alphabet_{\procA}} = w$. 
	If $s_1 = s_2$, then we find a violation to the assumption that $\mathcal{S}$ is sender-driven. 
	Hence, $s_1 \neq s_2$ and we can instantiate NMC (\cref{cond:nmc}) with $s_2 \xrightarrow{l_2} s_2'$,  $s_1$ and $w$ to obtain a contradiction. 
\end{proof}
\ccIPrefix*
\begin{proof}
	Let $\run$ be a run in $I(w)$, and let 
	$w' = \SyncToAsync(\trace(\run)) \in \lang(\mathcal{S})$.
	In the case that $I(w)$ contains finite runs, we can pick a finite $\run$. Otherwise, $\run$ is infinite. 
	We reason about each case in turn. 
	
	\paragraph{Case: $\run$ is a finite run.}
	In the case that $\run$ is a finite run, to show that 
	$w \in \pref(\lang(\mathcal{S}))$
	we need to show the existence of a $w'' \in \lang(\mathcal{S})$ such that $w \leq w''$. 
	We construct such a $w''$ by constructing a $u$ such that in $wu$, all participants have completed their actions in $\run$, and furthermore $wu$ is channel-compliant. 
	Then, because $w'$ is channel-compliant by construction, and for all participants $\procA \in \Procs$, it holds that $wu \wproj_{\Alphabet_{\procA}} = w' \wproj_{\Alphabet_{\procA}}$, by \cite[Lemma 23]{DBLP:conf/concur/MajumdarMSZ21} it follows that $wu \interswap w'$, and thus $wu \in \lang(\mathcal{S})$. 
	
	For each participant $\procA \in \Procs$, let $y_\procA$ be defined such that 
	$w \wproj_{\Alphabet_{\procA}} \cdot y_\procA = w' \wproj_{\Alphabet_{\procA}}$. 
	We construct $u$ from the $y_\procA$ for each participant, starting with $u = \emptystring$. 
	If there exists some participant in $\Procs$ such that $y_\procA[0] \in \Alphabet_{\procA,!}$, append $y_\procA$ to $u$ and update $y_\procA$. 
	If not, for all participants $\procA\in \Procs$, $y_\procA[0] \in \Alphabet_{\procA,?}$. 
	Each symbol $y_\procA[0]$ for all participants appears in $v$. 
	Let $i_\procA$ denote for each participant the index in $w'$ such that $w'[i] = y_\procA[0]$.
	Let $\procC$ be the participant with the minimum index $i_\procC$: append $y_\procC$ to $u$ and update $y_\procC$. 
	Termination is guaranteed by the strictly decreasing measure of $\sum_{\procA \in \Procs} |y_\procA|$. 
	Furthermore, it is clear that upon termination, for all participants $\procA \in \Procs$, 
	$wu \wproj_{\Alphabet_{\procA}} = w' \wproj_{\Alphabet_{\procA}}$. 
	
	We argue that $wu$ satisfies the inductive invariant of channel compliancy. 
	In the case where $u$ is extended with a send action, channel compliancy is trivially re-established. 
	In the receive case, channel compliancy is re-established by the fact that the append order for receive actions follows that in $v$, which is channel-compliant by construction.
	
	\paragraph{Case: $\run$ is an infinite run.}
	In the case that $\run$ is a infinite run, to show that 
	$w \in \pref(\lang(\mathcal{S}))$
	we likewise need to show the existence of a $w'' \in \lang(\mathcal{S})$ such that $w \leq w''$. 
	Like before, we construct a $u$ and show that $wu \in \lang(\mathcal{S})$. 
	However, unlike before, we cannot rely on the fact that $wu \interswap w'$ to show that $wu \in \lang(\mathcal{S})$, because $w'$ is an infinite word and \cite[Lemma 23]{DBLP:conf/concur/MajumdarMSZ21} applies only to finite words. 
	Instead, we must prove that $wu \in \lang(\mathcal{S})$ by the definition of infinite word membership in $\lang(\mathcal{S})$, namely: 
	$wu \preceq_\interswap^\omega w'$. 
	By the definition of $ \preceq_\interswap^\omega$, it further suffices to show that: 
	\[
	\forall v \leq wu,~
	\exists v' \leq w', u' \in \AlphAsyncSubscript^*.~
	vu' \interswap v' \enspace.
	\] 
	
	For each participant $\procA \in \Procs$, let $\run_\procA$ be defined as the largest prefix of $\run$ such that 
	$\SyncToAsync(\trace(\run_\procA)) \wproj_{\Alphabet_\procA} = w \wproj_{\Alphabet_\procA}$. 
	Let $\procD$ be the participant with the maximum $|\run_\procD|$ in $\Procs$. 
	Clearly, $\run_\procD \leq \run$. 
	Let $\beta$ be defined such that $\run = \run_\procD \cdot \beta$. 
	We split the construction of $u$ into two parts: let $u = u_1u_2$. 
	We construct $u_1$ as above, by appending uncompleted actions in $\run_\procD$, ordering send events before receive events, and further ordering receive events by the order in which they appear in $\run_\procD$. 
	Then, upon termination, $wu_1$ is channel-compliant and satisfies for all $\procA \in \Procs$, 
	$wu_1 \wproj_{\Alphabet_{\procA}} = \SyncToAsync(\trace(\run_\procD)) \wproj_{\Alphabet_{\procA}}$. 
	Let $u_2 = \SyncToAsync(\trace(\beta))$. 
	
	We now show that $wu_1u_2 \preceq_\interswap^\omega w'$. 
	
	Let $v$ be an arbitrary prefix of $wu_1u_2$. 
	If $v \leq wu_1$, we pick $v' = \SyncToAsync(\trace(\run_\procD)) \leq w'$ and $u' \in \AlphAsyncSubscript*$ to be such that $vu' = wu_1$. 
	Otherwise, if $wu_1 < v$, let $\run'$ be defined as the smallest prefix of $\run$ such that for all participants $\procA \in \Procs$, 
	$v \wproj_{\Alphabet_{\procA}} = \SyncToAsync(\trace(\run')) \wproj_{\Alphabet_{\procA}}$. 
	We pick $v' = \SyncToAsync(\trace(\run'))$. 
	Because $v$ is channel-compliant, we can repeat the reasoning in the finite case to extend $v$ with $u'$ and apply \cite[Lemma 23]{DBLP:conf/concur/MajumdarMSZ21} to conclude that $vu' \interswap v'$. 
\end{proof}
\languageInclusionLeft* 
\begin{proof}
	Let $w$ be a word in $\lang(\mathcal{S})$. 
	Prior to case splitting on whether $w$ is a finite or infinite word, we establish a claim that is used in both cases.
	
	\paragraph{Claim 1.} $\pref(\lang(\mathcal{S})) \subseteq \pref(\lang(\CLTS{T}))$. \fs{FIX: the use of paragraph gives more spacing here (medskip?) as there is after the end of the proof usually; added smallskip below}
	Let $w \in \pref(\lang(\mathcal{S}))$. 
	We prove that $w \in \pref(\lang(\CLTS{T}))$ via structural induction on $w$. 
	The base case, $w = \emptystring$, is trivial.
	For the inductive step, let 
	$wx \in \text{pref}(\lang(\mathcal{S}))$. 
	From the induction hypothesis, 
	$w \in \text{pref}(\lang\CLTS{T})$. 
	It suffices to show that the transition labeled with $x$ is enabled for the active participant in~$x$.
	Let $(\vec{s},\xi)$ denote the $\CLTS{T}$ configuration reached on $w$. 
	In the case that $x \in \Alphabet_!$, let $x = \snd{\procA}{\procB}{\val}$. 
	The existence of an outgoing transition $\xrightarrow{\snd{\procA}{\procB}{\val}}$ from $\vec{s}_\procA$ follows from the fact that $\pref(\lang(\mathcal{S})) \wproj_{\Alphabet_\procA} \subseteq \pref(\lang(T_\procA))$ (\cref{def:local-language-property}). 
	The fact that $wx \in \text{pref}(\lang\CLTS{T})$ follows immediately from this and the fact that send transitions in a CLTS are always enabled.
	In the case that $x \in \Alphabet_?$, let $x = \rcv{\procB}{\procA}{\val}$.
	We obtain an outgoing transition $\xrightarrow{\rcv{\procB}{\procA}{\val}}$ from $\vec{s}_\procA$ analogously.
	We additionally need to show that $\xi(\procB,\procA)$ contains $m$ at the head. 
	This follows from the fact that $w$ is channel-compliant (\cref{prop:clts-channel-compliant}) and the induction hypothesis. 
	This concludes our proof of prefix set inclusion. 
	\textit{End Proof of Claim 1.}
	\smallskip
	
	\paragraph{Case:} $w \in \AlphAsyncSubscript^*$.  
	In the finite case, it remains to show that $\CLTS{T}$ reaches a final configuration on $w$. 
	From the canonicity of $\CLTS{T}$, it holds that all states in $\vec{s}$ are final. From the fact that all finite words in $\lang(\mathcal{S})$ contain matching receive events, all channels in $\xi$ are empty. \fs{unclear where the ``fact'' comes from}
	
	\paragraph{Case:} $w \in \AlphAsyncSubscript^\omega$. 
	The infinite case when $w \in \AlphAsyncSubscript^\omega$ is immediate from Claim 1. 
\end{proof}
\languageInclusionRight*
\begin{proof}
	Let $w \in \lang\CLTS{T} $.
	We again case split on whether $w$ is a finite or infinite word. 
	
	\paragraph{Case:} $w \in \AlphAsync^*$. 
	First, we establish that $w$ is terminated. 
	Let $(\vec{s}, \xi)$ be the $\CLTS{T}$ configuration reached on $w$. 
	Because $w$ is a finite, maximal word in $\lang(\CLTS{T})$, it holds that all states in $\vec{s}$ are final, and all channels in $\xi$ are empty. Therefore, no receive transitions are enabled from $(\vec{s}, \xi)$. 
	We argue that no send transitions are enabled from $(\vec{s}, \xi)$ either. 
	Suppose by contradiction that there exists an outgoing transition $\vec{s}_\procA \xrightarrow{\snd{\procA}{\procB}{\val}} s' \in T_\procA$ for participant $\procA$. 
	Then, $w \wproj_{\Alphabet_{\procA}} \cdot \snd{\procA}{\procB}{\val} \in \pref(\lang(T_\procA))$, and by the canonicity of $T_\procA$, 
	$w \wproj_{\Alphabet_{\procA}} \cdot \snd{\procA}{\procB}{\val} \in \pref(\lang(\mathcal{S})) \wproj_{\Alphabet_{\procA}}$.
	Then, there exists a maximal run $\run'$ in $\mathcal{S}$ such that 
	$w \wproj_{\Alphabet_{\procA}} \cdot \snd{\procA}{\procB}{\val}  \leq \SyncToAsync(\trace(\run')) \wproj_{\Alphabet_{\procA}}$. 
	Furthermore, there exists a finite, maximal run $\run_{fin}$ in $\mathcal{S}$ such that 
	$w \wproj_{\Alphabet_{\procA}} = \SyncToAsync(\trace(\run)) \wproj_{\Alphabet_{\procA}}$. 
	Let $s_{fin}$ be the last state in $\run_{fin}$. By assumption, $s_{fin} \in F$. 
	Let $\alpha \cdot s_1 \xrightarrow{\msgFromToNS{\procA}{\procB}{\val}} s_2 \cdot \beta$ be the unique splitting of $\run'$ for $\procA$ with respect to $w$. 
	Then, $s_1$ and $s_{fin}$ are simultaneously reachable for $\procA$ on prefix $w \wproj_{\Alphabet_{\procA}}$. 
	From SC, there exists a $s_2'$ such that $s_{fin} \xRightarrow[\procA]{\msgFromToNS{\procA}{\procB}{\val}} \Kleenestar s_2'$. 
	We find a contradiction to the assumption that final states in $\mathcal{S}$ do not have outgoing transitions. 
	
	Next, we show that for every $\run \in I(w)$, it holds that for every $\procA \in \Procs$, 
	$w \wproj_{\Alphabet_{\procA}} = \SyncToAsync(\trace(\run)) \wproj_{\Alphabet_{\procA}}$. 
	This implies that there exist no infinite runs in $I(w)$. 
	Suppose by contradiction that there exists a run $\run \in I(w)$ and a non-empty set of participants $\mathcal{Q}$ such that for every $\procC \in \mathcal{Q}$, it holds that
	$
	w \wproj_{\Alphabet_\procC} <
	\bigl(
	\SyncToAsync(\trace(\run))
	\bigr)
	\wproj_{\Alphabet_\procC}
	$ (*).

	Given a participant $\procA$, let $\run_\procA$ denote the largest prefix of $\run$ that contains $\procA$'s local view of $w$. Formally,
	\[
	\run_\procA = max\{\run'~|~\run' \leq \run ~\land~ 
	\SyncToAsync(\trace(\run'))
	\wproj_{\Alphabet_\procA} = w\wproj_{\Alphabet_\procA}
	\}\enspace.
	\]
	Note that due to maximality, the next transition in $\run$ after $\run_\procA$ must have $\procA$ as its active participant. 
	Let $\procB$ be the participant in $\mathcal{S}$ for whom $\run_\procB$ is the smallest.
	From the canonicity of $T_\procB$ and (*), it follows that $\vec{s}_{\procB}$ has outgoing transitions.
	If $\vec{s}_{\procB}$ has outgoing send transitions, then we reach a contradiction to the fact that $w$ is terminated. 
	If $\vec{s}_{\procB}$ has outgoing receive transitions, it must be the case that the next transition in $\run$ after $\run_\procB$ is of the form $\msgFromTo{\procA}{\procB}{\val}$ for some $\procA$ and $\val$.
	From the fact that $\procB$ is the participant with the smallest $\run_\procB$, we know that $\run_\procB < \run_\procA$, and from the FIFO property of CLTS channels it follows that $\val$ is in $\xi(\procA,\procB)$. 
	Then, the receive transition is enabled for $\procB$, and we again reach a contradiction to the fact that $w$ is terminated. 
	
	Thus, we can pick any finite run $\run \in I(w)$ which is maximal by definition, and invoke \cite[Lemma 23]{DBLP:conf/concur/MajumdarMSZ21} to conclude that $\SyncToAsync(\trace(\run)) \interswap w$, and thus $w \in \lang(\mathcal{S})$. 
	
	\paragraph{Case:} $w \in \AlphAsync^\infty$. 
	By the semantics of $\lang(\mathcal{S})$, to show $w \in \lang(\mathcal{S})$ it suffices to show:
	\[
	\exists w' \in \Alphabet^\omega.~
	w' \in \SyncToAsync(\lang(\mathcal{S}))
	\land 
	w \preceq_\interswap^\omega w' \enspace.
	\]
	
	\paragraph{Claim.} $\Inters_{u \leq w} I(u)$ contains an infinite run. 
	
	First, we show that there exists an infinite run in $\mathcal{S}$. 
	We apply König's Lemma to an infinite tree where each vertex corresponds to a finite run. 
	We obtain the vertex set from the intersection sets of $w$'s prefixes; each prefix ``contributes'' a set of finite runs.
	Formally, for each prefix $u \leq w$, let $V_u$ be defined as: 
	\[
	V_u \is \Union_{\run_u \in I(u)} \text{min}\{\run' \mid \run' \leq \run_u \land \forall \procA \in \Procs.~u \wproj_{\Alphabet_{\procA}} \leq \SyncToAsync(\trace(\run')) \wproj_{\Alphabet_{\procA}} \} \enspace.
	\]
	By the assumption that $I(u) \neq \emptyset$, $V_u$ is guaranteed to be non-empty. 
	We construct a tree $\mathcal{T}_w(V,E)$ with 
	$V \is \Union_{u \leq w} V_u$ and 
	$E \is \{(\run_1, \run_2) \mid \run_1 \leq \run_2\}$. 
	The tree is rooted in the empty run, which is included in $V$ by the prefix $\emptystring$. 
	$V$ is infinite because there are infinitely many prefixes of $w$. 
	$\mathcal{T}_w$ is finitely branching due to the fact that $\mathcal{S}$ is deterministic: while there can be infinitely many transitions from a given state in $S$, there are only finitely many transitions from a given state in $S$ on a particular transition label. In fact, there is only a single transition. 
	Therefore, we can apply König's Lemma to obtain a ray in $\mathcal{T}_w$ representing an infinite run in $\mathcal{S}$. 
	
	Let $\run'$ be such an infinite run. 
	We now show that $\run' \in \Inters_{u \leq w} I(u)$. 
	Let $v$ be a prefix of $w$. 
	To show that $\run' \in I(v)$, it suffices to show that one of the vertices in $V_v$ lies on $\run'$. In other words, 
	\[
	V_v \inters \{v \mid v \in \run'\} \neq \emptyset \enspace.
	\]
	Assume by contradiction that $\run'$ passes through none of the vertices in $V_v$. 
	Then, for any $u' \geq u$, because intersection sets are monotonically decreasing, it must be the case that $\run'$ passes through none of the vertices in $V_u'$.
	Therefore, $\run'$ can only pass through vertices in $V_u''$, where $u'' \leq u$. 
	However, the set $\Union_{u'' \leq u} V_u''$ has finite cardinality. 
	We reach a contradiction, concluding our proof of the above claim. 
	
	Let $\run' \in \Inters_{u \leq w} I(u)$, and let 
	$w' = \SyncToAsync(\trace(\run'))$. 
	It is clear that 
	$w' \in \AlphAsyncSubscript^\omega$ and $w' \in \SyncToAsync(\lang(\mathcal{S}))$.
	It remains to show that 
	$w \preceq_\interswap^\omega w'$. 
	By the definition of $ \preceq_\interswap^\omega$, it further suffices to show that: 
	\[
	\forall u \leq w,~
	\exists u' \leq w', v \in \AlphAsyncSubscript^*.~
	uv \interswap u' \enspace.
	\]
	Let $u$ be an arbitrary prefix of $w$. 
	Because by definition $\run' \in I(u)$, it holds that 
	$u \wproj_{\Alphabet_\procA} \leq \SyncToAsync(\trace(\run')) \wproj_{\Alphabet_\procA}$.
	
	For each participant $\procA \in \Procs$, let $\run'_\procA$ be defined as the largest prefix of $\run'$ such that 
	$\SyncToAsync(\trace(\run'_\procA)) \wproj_{\Alphabet_\procA} = u \wproj_{\Alphabet_\procA}$. 
	Such a run is well-defined by the fact that $u$ is a prefix of an infinite word $w$, and there exists a longer prefix $v$ such that $u \leq v$ and 
	$v \wproj_{\Alphabet_\procA} \leq \SyncToAsync(\trace(\run')) \wproj_{\Alphabet_\procA}$.
	
	Let $\procD$ be the participant with the maximum $|\run'_\procD|$ in $\Procs$. 
	Let $u' = \SyncToAsync(\trace(\run'_\procD))$. 
	Clearly, $u' \leq w'$. 
	Because $u'$ is $\SyncToAsync(\trace(\run'_\procD))$ for the participant with the longest $\run'_\procD$, it holds for all participants $\procA \in \Procs$ that 
	$u \wproj_{\Alphabet_\procA} \leq u' \wproj_{\Alphabet_\procA}$. 
	Then, there must exist $y_\procA \in \Alphabet_\procA^*$ such that 
	\[
	u \wproj_{\Alphabet_\procA} \cdot y_\procA = u' \wproj_{\Alphabet_\procA}\enspace.
	\]
	Let $y_\procA$ be defined in this way for each participant. 
	We construct $v \in \AlphAsyncSubscript^*$ such that $uv \interswap u'$. 
	Let $v$ be initialized with $\emptystring$.
	If there exists some participant in $\Procs$ such that $y_\procA[0] \in \Alphabet_{\procA,!}$, append $y_\procA$ to $v$ and update $y_\procA$. 
	If not, for all participants $\procA\in \Procs$, $y_\procA[0] \in \Alphabet_{\procA,?}$. 
	Each symbol $y_\procA[0]$ for all participants appears in $u'$. 
	Let $i_\procA$ denote for each participant the index in $u'$ such that $u'[i] = y_\procA[0]$.
	Let $\procC$ be the participant with the minimum index $i_\procC$. 
	Append $y_\procC$ to $v$ and update $y_\procC$. 
	Termination is guaranteed by the strictly decreasing measure of $\sum_{\procA \in \Procs} |y_\procA|$. 
	
	We argue that $uv$ satisfies the inductive invariant of channel compliancy. 
	In the case where $v$ is extended with a send action, channel compliancy is trivially re-established. 
	In the receive case, channel compliancy is re-established by the fact that the append order for receive actions follows that in $u'$, which is channel-compliant by construction.
	We conclude that $uv \interswap u'$ by applying \cite[Lemma~22]{DBLP:conf/concur/MajumdarMSZ21}.  
	
\end{proof} 
\INonEmpty*
\begin{proof}
	We prove the claim by induction on the length of $w$.
	
	\paragraph{\textbf{Base Case.}} $w = \emptystring$.
	The trace $w = \emptystring$ is trivially consistent with all maximal runs, and $I(w)$ therefore contains all maximal runs. 
	By assumption, $\mathcal{S}$ contains at least one maximal run. 
	Thus, $I(w)$ is non-empty.
	\efl{Technically we assume here that the LTS is non-empty. Empty LTS are trivially implementable.}
	
	\paragraph{\textbf{Induction Step.}} Let $wx$ be an extension of $w$ by $x \in \AlphAsyncSubscript$.
	
	The induction hypothesis states that 
	$I(w) \neq \emptyset$. 
	To re-establish the induction hypothesis, we need to show
	$I(wx) \neq \emptyset$. 
	We proceed by case analysis on whether $x$ is a receive or send event.
	
	\paragraph{Send Case.} Let $x$ = $\snd{\procA}{\procB}{\val}$.
	By \cref{lm:sndPrefixPreservation}, there exists a run in $I(wx)$ that shares a prefix with a run in $I(w)$. $I(wx) \neq \emptyset$ again follows immediately.
	
	\paragraph{Receive Case.} Let $x$ = $\rcv{\procB}{\procA}{\val}$.
	By \cref{lm:rcvIntersectionSetEquality}, $I(wx) = I(w)$. $I(wx) \neq \emptyset$ follows trivially from the induction hypothesis and this equality.
\end{proof}
\begin{proposition}[CLTS traces are channel-compliant]
	\label{prop:clts-channel-compliant}
	Let $\CLTS{T}$ be a CLTS, and let $w \in \AlphAsyncSubscript^*$ be a trace of $\CLTS{T}$. 
	Let $(\vec{s}, \xi)$ be the $\CLTS{T}$ configuration reached on $w$. 
	Then, $w$ is channel-compliant, and for every pair of participants $\procA \neq \procB \in \Procs$, 
	$\MsgVals(w \wproj_{\snd{\procA}{\procB}{\hole}}) = \MsgVals(w \wproj_{\rcv{\procA}{\procB}{\hole}}) \cdot \xi(\procA, \procB)$. 
\end{proposition}
The proof of the same proposition for communicating state machines can be generalized directly to CLTSs, and thus we refer the reader to \cite[Lemma 19]{DBLP:conf/concur/MajumdarMSZ21}.
\rcvIntersectionEquality*
\begin{proof}
	Let $x$ = $\rcv{\procB}{\procA}{\val}$.
	Because $wx$ is a trace of $\CLTS{T}$, there exists a run
	$(\vec{s}_0, \xi_0) \xrightarrow{w} \Kleenestar (\vec{s}, \xi) \xrightarrow{x} (\pvec{s}',\xi')$
	such that
	$m$ is at the head of $\xi(\procB,\procA)$. 
	
	We assume that $I(w)$ is non-empty; if $I(w)$ is empty then $I(wx)$ is trivially empty.
	
	To show $I(w) = I(wx)$, let $\run \in I(w)$ and we show that $\run \in I(wx)$.
	Recall that
	$
	I(wx)
	$ is defined as $
	\Inters_{\procC \in \Procs} \globcomplocal{\mathcal{S}}{\procC}{wx}
	$.
	
	Because
	$\globcomplocal{\mathcal{S}}{\procC}{wx} =
	\globcomplocal{\mathcal{S}}{\procC}{w}$
	for every $\procC \in \Procs$ with $\procC \neq \procA$,
	it suffices to show that
	$\run \in \globcomplocal{\mathcal{S}}{\procA}{wx}$
	to show $\run \in I(wx)$.

	We proceed via proof by contradiction so let
	$\run \notin \globcomplocal{\mathcal{S}}{\procA}{wx}$ for $\run \in I(w)$. 
	
	Let $\alpha \cdot s_{\mathit{pre}} \xrightarrow{l} s_{\mathit{post}} \cdot \beta$ be the unique splitting of $\run$ for $\procA$ matching $w$. 
	By definition of unique splittings, $\procA$ is the active participant in $l$. 
	Because 
	$\run \notin \globcomplocal{\mathcal{S}}{\procA}{wx}$, it follows that 
	$l \neq \msgFromTo{\procB}{\procA}{\val}$. 
	By \cref{cor:no-mixed-choice}, $\procA$ is the receiver in $l$, and $l$ is of the form $\msgFromTo{\procC}{\procA}{\val'}$, 
	where $\procC \neq \procB$ or $\val' \neq \val$. 
	
	Before performing case analysis, we first establish a claim that is used in both cases. 
	Let $\run_\procA$ denote the largest prefix of $\run$ that is consistent with $w$ for $\procA$.
	Formally,
	$
	\run_\procA =
	\max
	\set{
		\run \mid \run \leq \run ~\land~ \bigl(
		\SyncToAsync(\trace(\run))
		\bigr)
		\wproj_{\Alphabet_\procA} \preforder  w\wproj_{\Alphabet_\procA}
	}
	$.
	Let $\run_\procB$ be defined analogously.
	It is clear that 
	$\run_\procA = \alpha \cdot s_{pre}$.

	\textit{Claim I. } $\run_\procB > \run_\procA$.
	
	From \cref{prop:clts-channel-compliant}, 
	$\MsgVals(w \wproj_{\snd{\procB}{\procA}{\hole}}) = 
	\MsgVals(w \wproj_{\rcv{\procB}{\procA}{\hole}}) \cdot \xi(\procB,\procA)$.
	Because $\run_\procA = \alpha \cdot s_{pre}$, it follows that
	$\MsgVals(w \wproj_{\rcv{\procB}{\procA}{\hole}}) = 
	\MsgVals(\SyncToAsync(\trace(\alpha \cdot s_{pre})) \wproj_{\rcv{\procB}{\procA}{\hole}})$. 
	
	Because $\val$ is at the head of $\xi(\procB, \procA)$ by assumption, there exists $u \in \MsgVals^*$ such that 
	$\MsgVals(w \wproj_{\snd{\procB}{\procA}{\hole}}) = 
	\MsgVals(\SyncToAsync(\trace(\alpha \cdot s_{pre}))) \wproj_{\rcv{\procB}{\procA}{\hole}} \cdot \val \cdot u$. 
	Thus, 
	$\MsgVals(w \wproj_{\snd{\procB}{\procA}{\hole}}) > 
	\MsgVals(\SyncToAsync(\trace(\alpha \cdot s_{pre}))) \wproj_{\rcv{\procB}{\procA}{\hole}}$ and 
	$\run_\procB > \run_\procA$ follows. 
	\textit{End Proof of Claim I.}
	
	\noindent
	\textit{Case:} $\procC = \procB$ and $\val' \neq \val$.
	
	We discharge this case by showing a contradiction to the assumption that $m$ is at the head of $\xi(\procB,\procA)$. 
	Because $\alpha \cdot s_{pre} \leq \run_\procA$ and $\run_\procA < \run_\procB$ from Claim I, it must be the case that
	$\alpha \cdot s_{pre} \xrightarrow{l} s_{post} \leq  \run_\procB$
	and
	$\snd{\procB}{\procA}{\val'}$ is in $w \wproj_{\Alphabet_{\procB}}$. 
	From \cref{prop:clts-channel-compliant}, it follows that 
	$\mathcal{V}(w \wproj_{\snd{\procB}{\procA}{\hole}}) = \MsgVals(w \wproj_{\rcv{\procB}{\procA}{\hole}}) \cdot \val' \cdot u'$ and
	$\xi(\procB,\procA)  =  \val' \cdot u'$, i.e. $m'$ is at the head of $\xi(\procB, \procA)$.
	We find a contradiction to the assumption that $m' \neq m$. 
	
	\noindent
	\textit{Case:} $\procC \neq \procB$.
	
	We discharge this case by showing a contradiction to RC.
	First, we establish the existence of a transition $s_1 \xrightarrow{\msgFromToNS{\procB}{\procA}{\val}} s_2 \in T$ such that $s_1 \neq s_{pre}$ and $s_1$ is reachable by $\procA$ on $\SyncToAsync^{-1}(w \wproj_{\Alphabet_{\procA}})$.
	\FS{Minor: at first, I thought that ``transition $\in T$'' does not ``type'' but it does. The reason was that $T_\procA$ is part of the CLTS and $T$ are the transitions of the sender-driven LTS. Maybe change one?}
	By the assumption that $wx$ is a trace of $\CLTS{T}$, 
	it follows that $wx \wproj_{\Alphabet_{\procA}}$ is a prefix of $\lang(T_\procA)$. 
	By the canonicity of $\CLTS{T}$, it holds that 
	$\pref(\lang(T_\procA)) \subseteq \pref(\lang(\mathcal{S}) \wproj_{\Alphabet_{\procA}})$, and thus
	$wx \wproj_{\Alphabet_{\procA}} \in \pref(\lang(\mathcal{S}) \wproj_{\Alphabet_{\procA}})$.
	Thus, there exists a maximal run $\run'$ in $\mathcal{S}$ such that 
	$wx \wproj_{\Alphabet_{\procA}} \leq \SyncToAsync(\trace(\run')) \wproj_{\Alphabet_{\procA}}$
	and 
	$s_1 \xrightarrow{\msgFromToNS{\procB}{\procA}{\val}} s_2 \in \run'$. 
	Because $\mathcal{S}$ is sender-driven, there does not exist a state $s \in S$ with two outgoing transition labels with different senders. 
	Therefore, $s_1 \neq s_{pre}$.
	
	By the fact that 
	$\alpha \cdot s_{pre} \xrightarrow{l} s_{post} \cdot \beta$ is the unique splitting of $\run$ for $\procA$ matching~$w$, it holds that $s_{pre}$ is also reachable by $\procA$ on $\SyncToAsync^{-1}(w \wproj_{\Alphabet_{\procA}})$.
	
	We instantiate RC with 
	$s_1 \xrightarrow{\msgFromToNS{\procB}{\procA}{\val}} s_2$, $s_{pre} \xrightarrow{\msgFromToNS{\procC}{\procA}{\val}} s_{post}$ and $\SyncToAsync^{-1}(w \wproj_{\Alphabet_{\procA}})$ 
	to obtain:
	\[
	\neg (\exists v \in \pref (\lang_{s_{post}}).~
	v \wproj_{\Alphabet_\procA} = \emptystring \land 
	\MsgVals(v \wproj_{\snd{\procB}{\procA}{\_}}) =
	\MsgVals(v \wproj_{\rcv{\procB}{\procA}{\_}}) \cdot \val) \enspace.
	\]
	We show, on the contrary, that
	\[
	\exists v \in \pref (\lang_{s_{post}}).~
	v \wproj_{\Alphabet_\procA} = \emptystring \land 
	\MsgVals(v \wproj_{\snd{\procB}{\procA}{\_}}) =
	\MsgVals(v \wproj_{\rcv{\procB}{\procA}{\_}}) \cdot \val \enspace.
	\]
	It is clear that $s_{post} \cdot \beta$ is a maximal run in $\mathcal{S}_{s_{post}}$.  
	By \cref{lm:cc-intersection-nonemptiness-implies-prefix}, to show that a witness $v \in \pref(\lang_{s_{post}})$, it suffices to show that $v$ is channel-compliant and furthermore, that for all participants $\procD \in \Procs$, 
	$v \wproj_{\Alphabet_{\procD}} \leq \SyncToAsync(\trace(s_{post} \cdot \beta)) \wproj_{\Alphabet_{\procD}}$. 
	
	Recall that $w$ is a trace of $\CLTS{T}$ and is thus channel-compliant. 
	Intuitively, we obtain a witness for $v$ by deleting from $w$ symbols that belong to $\SyncToAsync(\trace(\alpha \cdot s_{pre} \xrightarrow{l} s_{post}))$. 
	Formally, let $v$ be initialized to $w$ and let $l_1 \ldots l_n = \trace(\alpha \cdot s_{pre} \xrightarrow{l} s_{post})$. 
	For each $i \in \set{1, \ldots, n}$, let $l_i \is \msgFromTo{\procA_i}{\procB_i}{\val_i}$.
	We check whether $\snd{\procA_i}{\procB_i}{\val_i} \leq w \wproj_{\Alphabet_{\procA_i}}$, and if so, we delete the symbol $\snd{\procA_i}{\procB_i}{\val_i}$ from $w$. 
	We then check whether $\rcv{\procA_i}{\procB_i}{\val_i} \leq w \wproj_{\Alphabet_{\procB_i}}$, and again delete the symbol if so. 
	Note that due to the channel-compliancy of $v$, either both symbols are deleted, or only the send action is deleted.  
	We argue that the inductive invariant of channel-compliancy is satisfied: if a matching pair of send and receive actions are found in $v$ and deleted, each of $\MsgVals(v \wproj_{\rcv{\procA_i}{\procB_i}{\_}})$ and $\MsgVals(v \wproj_{\snd{\procA_i}{\procB_i}{\_}})$ lose their head message, and $\MsgVals(v \wproj_{\rcv{\procA_i}{\procB_i}{\_}}) \leq \MsgVals(v \wproj_{\snd{\procA_i}{\procB_i}{\_}})$ continues to hold; if only the send action is found and deleted, then it must be the case that $\MsgVals(v \wproj_{\rcv{\procA_i}{\procB_i}{\_}}) = \emptystring$ and the invariant is trivially re-established. 
	Thus, we establish that upon termination, $v$ is channel-compliant. 
	Furthermore, it holds that $s_{post} \cdot \beta \in I^{\mathcal{S}_{s_{post}}}(v)$. 
	\efl{Add sentence for why $\procA$ has no actions left}
	
	Recall that 
	$\MsgVals(w \wproj_{\snd{\procB}{\procA}{\_}}) =
	\MsgVals(w \wproj_{\rcv{\procB}{\procA}{\_}}) \cdot \val$.
	It remains to show that 
	$\MsgVals(v \wproj_{\snd{\procB}{\procA}{\_}}) =
	\MsgVals(v \wproj_{\rcv{\procB}{\procA}{\_}}) \cdot \val$.
	This holds from the fact that $\alpha \cdot s_{pre} = \run_\procA < \run_\procB$, which means that any labels of the form $\msgFromTo{\procB}{\procA}{\hole}$ in $l_1 \ldots l_n$ must find and delete a matching pair of send and receive actions in~$v$, thus preserving the above equality.

\end{proof}
\sndPreservation*
\begin{proof}
	Let $x$ = $\snd{\procA}{\procB}{\val}$.
	We prove the claim by induction on the length of $w$.
	
	\paragraph{\textbf{Base Case.}} $w = \emptystring$.
	By definition, $I(\emptystring)$ contains all maximal runs in $\mathcal{S}$. 
	Then, $\SyncToAsync(\trace(\alpha \cdot s_{pre})) \wproj_{\Alphabet_\procA} = \emptystring$ and it holds that 
	$s_0 \xRightarrow[\procA]{\emptystring} \Kleenestar s_{pre}$. 
	We argue that there exists $s_1 \in S$ such that $s_0 \xRightarrow[\procA]{\emptystring} \Kleenestar s_1$. 
	From the canonicity of $\CLTS{T}$ and the fact that $x \in \pref(\lang(T_\procA))$, it follows that 
	$x \in \pref(\lang(\mathcal{S}) \wproj_{\Alphabet_{\procA}})$.
	Thus, there exists $w \in \lang(\mathcal{S})$ such that $x \leq w \wproj_{\Alphabet_{\procA}}$, and consequently there exists a run $\run'$ such that $x \leq \SyncToAsync(\trace(\run')) \wproj_{\Alphabet_{\procA}}$.
	The unique splitting of $\run'$ for $\procA$ with respect to $\emptystring$ gives us a candidate for $s_1$. 
	By \cref{cond:scc}, there exists a $s_2$ such that $s_1 \xRightarrow[\procA]{l} \Kleenestar s_2$. 
	By the assumption that every run in $\mathcal{S}$ extends to a maximal run, there exists a maximal run in $I(x)$. 
	
	\paragraph{\textbf{Induction Step.}} Let $wx$ be an extension of $w$ by $x \in \Alphabet_{\procA, !}$.
	To re-establish the induction hypothesis, we need to show the existence of a run $\bar \run$ in $I(wx)$ such that 
	$\alpha \cdot s_{pre} \leq \bar \run$.
	Since $\procA$ is the active participant in $x$, it holds for any $\procC \neq \procA$ that
	$\globcomplocal{\mathcal{S}}{\procC}{w} = \globcomplocal{\mathcal{S}}{\procC}{wx}$.
	Therefore, to prove the existential claim, it suffices to construct a run $\bar \run$ that satisfies:
	\begin{enumerate}
		\item \label{claim:soundness-snd-case-in-extension-run-set}
		$\bar \run \in \globcomplocal{\mathcal{S}}{\procA}{wx}$,
		\item \label{claim:soundness-snd-case-in-original-run-set}
		$\bar \run \in I(w)$, and 
		\item \label{claim:soundness-snd-case-correct-prefix}
		$\alpha \cdot s_{pre} \leq \bar \run$.
	\end{enumerate}
	
	In the case that $l \wproj_{\Alphabet_{\procA}} = x$, we \fs{Minor: $l$ is quite ``far away'' here.}
	are done:
	Property~\ref{claim:soundness-snd-case-correct-prefix} and \ref{claim:soundness-snd-case-in-original-run-set} hold by construction, and Property~\ref{claim:soundness-snd-case-in-extension-run-set} holds by the definition of possible run sets.
	
	In the case that $l \wproj_{\Alphabet_{\procA}} \neq x$, we show the existence of a different continuation such that the resulting run satisfies all three conditions. 
	
	First, we establish that $\procA$ is the sender in $l$. 
	By definition of unique splitting, we know that $\procA$ is active in $l$.
	Assume towards a contradiction that $\procA$ is the receiver in $l$.
	Then, $l$ is of the form $\msgFromTo{\procB}{\procA}{\val}$. 
	Because $\alpha \cdot s_{pre} \xrightarrow{\msgFromToNS{\procB}{\procA}{\val}} s_{post} \cdot \beta$ is a maximal run in $\mathcal{S}$, 
	we have that
	$(w \cdot \rcv{\procB}{\procA}{\val})\wproj_{\Alphabet_{\procA}} \in \pref(\lang(\mathcal{S})\wproj_{\Alphabet_{\procA}})$. 
	By the canonicity of $\CLTS{T}$, 
	it holds that
	$\pref(\lang(\mathcal{S}) \wproj_{\Alphabet_{\procA}}) \subseteq \pref(\lang(T_\procA))$,
	and therefore
	$(w \cdot \rcv{\procB}{\procA}{\val})\wproj_{\Alphabet_{\procA}} \in \pref(\lang(T_\procA))$. 
	By assumption that $wx$ is a trace of $\CLTS{T}$, it holds that 
	$(wx) \wproj_{\Alphabet_{\procA}} \in \pref(\lang(T_\procA))$.
	From the fact that $\rcv{\procB}{\procA}{\val} \in \Alphabet_{\procA,?}$ and $x \in \Alphabet_{\procA,!}$, we find a contradiction to \cref{cor:no-mixed-choice}. 
	Therefore, $l$ must be of the form $\msgFromTo{\procA}{\procB'}{\val'}$, with $\procB' \neq \procB$ or $\val' \neq \val$.

	By assumption that $wx$ is a trace of $\CLTS{T}$, it holds that $wx \wproj_{\Alphabet_{\procA}} \in \pref(\lang(T_\procA))$.
	By the canonicity of $\CLTS{T}$ (\cref{def:local-language-property}(ii)), 
	we have
	$\pref(\lang(T_\procA)) \subseteq \pref((\lang(\mathcal{S}) \wproj_{\Alphabet_{\procA}}))$
	and hence,
	$wx \wproj_{\Alphabet_{\procA}} \in \pref(\lang(\mathcal{S}) \wproj_{\Alphabet_{\procA}})$.
	Thus, there exists $v \in \lang(\mathcal{S})$ such that
	$wx \wproj_{\Alphabet_{\procA}} \leq v \wproj_{\Alphabet_{\procA}}$, and consequently there exists a run $\run'$ such that $wx \wproj_{\Alphabet_{\procA}} \leq \SyncToAsync(\trace(\run')) \wproj_{\Alphabet_{\procA}}$.
	The unique splitting of $\run'$ for $\procA$ with respect to $w$ gives us a transition 
	$s_1 \xrightarrow{\msgFromToNS{\procA}{\procB}{\val}} s_2 \in T$. 
	
	If $s_1 = s_{pre}$, then $\alpha \cdot s_{pre} \xrightarrow{\msgFromToNS{\procA}{\procB}{\val}} s_2$ is a run in $\mathcal{S}$. 
	Otherwise, we instantiate SC (\cref{cond:scc}) with $s_1 \xrightarrow{\msgFromToNS{\procA}{\procB}{\val}} s_2$, $s_{pre}$ and the witness $w \wproj_{\Alphabet_{\procA}}$.
	Then, there exists $s'$ such that 
	$s_{pre} \xRightarrow[\procA]{\msgFromToNS{\procA}{\procB}{\val}} \Kleenestar s'$. 
	We argue that, in fact, 
	$s_{pre} \xrightarrow{\msgFromToNS{\procA}{\procB}{\val}} s' \in T$. 
	This follows from the fact established above that $\procA$ is the sender in $l$, and that $s_{pre} \xrightarrow{l} s_{post} \in T$. 
	By the assumption that $\mathcal{S}$ is sender driven, there does not exist a state with outgoing transitions that do not share a sender. 
	Therefore, 
	$\alpha \cdot s_{pre} \xrightarrow{\msgFromToNS{\procA}{\procB}{\val}} s'$ is a run in $\mathcal{S}$. 
	
	Either way, we have found a run that thus far satisfies Property~\ref{claim:soundness-snd-case-in-extension-run-set} and \ref{claim:soundness-snd-case-correct-prefix} regardless of its choice of maximal suffix. 
	Let $\alpha \cdot s_{pre} \xrightarrow{\msgFromToNS{\procA}{\procB}{\val}} \bar{s'}$ be a run in $\mathcal{S}$. 
	Then, for all choices of $\bar \beta$ such that $\alpha \cdot s_{pre} \xrightarrow{\msgFromToNS{\procA}{\procB}{\val}} \bar{s'} \cdot \bar \beta$ 
	is a maximal run, 
	both
	$wx \wproj_{\Alphabet_\procA} \leq  \SyncToAsync(\trace(\alpha \cdot s_{pre} \xrightarrow{\msgFromToNS{\procA}{\procB}{\val}} \bar{s'}))$ 
	and $\alpha \cdot s_{pre} \leq \alpha \cdot s_{pre} \xrightarrow{\msgFromToNS{\procA}{\procB}{\val}} \bar{s'} \cdot \bar \beta$ hold.
	
	Property~\ref{claim:soundness-snd-case-in-original-run-set}, however, requires that the projection of $w$ onto each participant is consistent with $\bar \run$, and this cannot be ensured by the prefix alone.
	
	We construct the remainder of $\bar \run$ by picking an arbitrary maximal suffix to form a candidate run, and iteratively performing suffix replacements on the candidate run until it lands in $I(w)$.
	Let $\bar \beta$ be a run suffix such that 
	$\alpha \cdot s_{pre} \xrightarrow{\msgFromToNS{\procA}{\procB}{\val}} \bar{s'} \cdot \bar \beta$ 
	is a maximal run in $\mathcal{S}$.
	Let $\run_c$ denote this candidate run. 
	
	If $\rho_c \in I(w)$, we are done.
	Otherwise, $\run_c \notin I(w)$ and there exists a non-empty set of processes $\mathcal{Q} \subseteq \Procs$ such that for each $\procC \in \mathcal{Q}$,
	\begin{align}\label{eq:not-prefix-of-rho-c}
		w \wproj_{\Alphabet_\procC} \nleq \SyncToAsync(\trace(\run_c)) \wproj _{\Alphabet_\procC}\enspace.
	\end{align}
	
	By the fact that $\rho \in I(w)$,
	\begin{align}\label{eq:prefix-of-rho}
		w \wproj_{\Alphabet_\procC} \leq \SyncToAsync(\trace(\run)) \wproj _{\Alphabet_\procC}\enspace.
	\end{align}
	We can rewrite (\ref{eq:not-prefix-of-rho-c}) and (\ref{eq:prefix-of-rho}) above to make explicit their shared prefix $\alpha \cdot s_{pre}$:
	\begin{align}
		\label{eq:not-prefix-of-rho-c-expanded}
		w \wproj_{\Alphabet_{\procC}} &\nleq \SyncToAsync(\trace(
		\alpha \cdot s_{pre} \xrightarrow{\msgFromToNS{\procA}{\procB}{\val}} \bar{s'} \cdot \bar{\beta})) \wproj _{\Alphabet_\procC} \\
		w \wproj_{\Alphabet_{\procC}} &\leq \SyncToAsync(\trace(\alpha \cdot s_{pre}  \xrightarrow{\msgFromToNS{\procA}{\procB'}{\val'}} s_{post} \cdot \beta)) \wproj _{\Alphabet_\procC}\enspace.
		\label{eq:prefix-of-rho-expanded}
	\end{align}
	
	We can further rewrite (\ref{eq:not-prefix-of-rho-c-expanded}) and (\ref{eq:prefix-of-rho-expanded}) to make explicit their point of disagreement: 
	\begin{align}
		\label{eq:not-prefix-of-rho-c-shared}
		w \wproj_{\Alphabet_\procC} &\nleq (\SyncToAsync(\trace(\alpha \cdot s_{pre})).~\snd{\procA}{\procB}{\val}.~\rcv{\procA}{\procB}{\val}.~ \SyncToAsync(\trace(\bar \beta))) \wproj _{\Alphabet_\procC} \\
		w \wproj_{\Alphabet_\procC} &\leq (\SyncToAsync(\trace(\alpha \cdot s_{pre})).~\snd{\procA}{\procB'}{\val'}.~\rcv{\procA}{\procB'}{\val'}.~\SyncToAsync(\trace(\beta))) \wproj _{\Alphabet_\procC}                                        
		\label{eq:prefix-of-rho-shared}
	\end{align}
	
	It is clear that in order for both \ref{eq:not-prefix-of-rho-c-shared} and \ref{eq:prefix-of-rho-shared} to hold, it must be the case that 
	$\SyncToAsync(\trace(\alpha \cdot s_{pre})) \wproj_{\Alphabet_{\procC}} < w \wproj_{\Alphabet_{\procC}}$.

	We formalize the point of disagreement between $w \wproj_{\Alphabet_\procC}$ and $\run_c$ using an index~$i_\procC$ representing the position of the first disagreeing symbol in $\trace(\run_c)$:
	\[
	i_\procC \is \text{max}\{i \mid  \SyncToAsync(\trace(\run_c)[0..i-1]) \wproj_{\Alphabet_{\procC}} \leq w \wproj_{\Alphabet_{\procC}} \}\enspace.
	\]
	By the maximality of $i_\procC$, it holds that $\procC$ is the active participant in $\trace(\run_c)[i_\procC]$. 
	By the fact that 
	$\SyncToAsync(\trace(\alpha \cdot s_{pre})) \wproj_{\Alphabet_{\procC}} < w \wproj_{\Alphabet_{\procC}}$ we know that 
	\[
	i_\procC > \card{\trace(\alpha \cdot s_{pre})} \enspace.
	\]
	
	We identify the participant in $\mathcal{Q}$ with the \textit{earliest disagreement} in $\SyncToAsync(\trace(\run_c))$: let $\bar{\procC}$ be the participant in $\mathcal{Q}$ with the smallest $i_{\bar{\procC}}$.
	If two participants that share the same smallest index, then by the fact that both participants are active in $\trace(\run_c)[i_{\bar\procC}]$, it must be the case that one is the sender and one is the receiver: we pick the sender to be $\bar \procC$.  
	Let $y_{\bar \procC}$ denote $\SyncToAsync(\trace(\run_c[i_{\bar \procC}])) \wproj_{\Alphabet_{\bar \procC}}$.
	
	\paragraph{Claim I} $y_{\bar \procC}$ is a send event.
	
	Assume by contradiction that $y_{\bar \procC}$ is a receive event.
	We identify the symbol in $w$ that disagrees with $y_{\bar \procC}$: let $w'$ be the largest prefix of $w$ such that $w' \wproj_{\Alphabet_{\bar \procC}} \leq \SyncToAsync(\trace(\run_c)) \wproj_{\Alphabet_{\bar \procC}}$.
	By definition, $w' \wproj_{\Alphabet_{\bar \procC}} = \SyncToAsync(\trace(\run_c)[0..i_{\bar \procC}-1]) \wproj_{\Alphabet_{\bar \procC}}$.
	Let $z$ be the next symbol following $w'$ in $w$; then $w'z \leq w$ and $z \in \Alphabet_{\bar \procC}$ with $z \neq y_{\bar \procC}$.
	Furthermore, by No Mixed Choice (\ref{cor:no-mixed-choice}) we have that $z \in \Alphabet_{\bar \procC, ?}$.
	
	By assumption,
	$w'z \nleq \SyncToAsync(\trace(\run_c)[0..i_{\bar \procC}])$.
	Therefore, any run with a trace that begins with $\run_c[0..i_{\bar \procC}]$ cannot be contained in $\globcomplocal{\mathcal{S}}{\bar \procC}{w'z}$, or consequently in $I(w'z)$.
	We show however, that $I(w'z)$ must contain some runs that begin with $\run_c[0..i_{\bar \procC}]$.
	From Lemma \ref{lm:rcvIntersectionSetEquality} for traces $w'$ and $w' z$, we obtain that
	$I(w') = I(w'z)$.
	Therefore, it suffices to show that $I(w')$ contains runs that begin with $\run_c[0..i_{\bar \procC}]$.

	\paragraph{Claim II} $\forall w'' \leq w'.\, I(w'')$ contains runs that begin with $\run_c[0..i_{\bar \procC}]$. 
	
	We prove the claim via induction on $w'$. 
	
	The base case is trivial from the fact that $I(\emptystring)$ contains all maximal runs.
	
	For the inductive step, let $w''y \leq w'$.
	
	In the case that $y \in \Alphabet_?$, we know $I(w''y) = I(w'')$ from Lemma \ref{lm:rcvIntersectionSetEquality} and the witness from $I(w'')$ can be reused.
	
	In the case that $y \in \Alphabet_!$, let $\procD$ be the active participant of $y$ and let $\run'$ be a run in $I(w'')$ beginning with $\run_c[0..i_{\bar \procC}]$ given by the inner induction hypothesis. 
	Let $\alpha' \cdot s_3 \xrightarrow{l'} s_4 \cdot \beta'$ be the unique splitting of $\run'$ for $\procD$ with respect to $w''$.
	If $\SyncToAsync(l') \wproj_{\Alphabet_\procD} = y$, then $\rho'$ can be used as the witness.
	Otherwise, $\SyncToAsync(l') \wproj_{\Alphabet_\procD} \neq y$, and $\rho' \notin \globcomplocal{\mathcal{S}}{\procD}{w''y}$.

	The outer induction hypothesis holds for all prefixes of $w$: we instantiate it with $w''$ and $y$ to obtain:
	\[
	\exists~\rho'' \in I(w''y).~\alpha' \cdot s_3 \leq \rho''\enspace.
	\]
	Let $i_\procD$ be defined as before; it follows that $\run'[i_\procD] = s_3$.
	It must be the case that $i_\procD > i_{\bar \procC}$: if $i_\procD \leq i_{\bar \procC}$, because $\run_c$ and $\run'$ share a prefix $\run_c[0..i_{\bar \procC}]$ and $w''y \leq w$, $\procD$~would be the earliest disagreeing participant instead of $\bar \procC$.
	
	Because $i_\procD > i_{\bar \procC}$, 
	$\run_c[0..i_{\bar \procC}] = \run'[0..i_{\bar \procC}] \leq \run'[0..i_\procD]$.
	Because $\run'[0..i_\procD] = \alpha' \cdot s_3 \leq \run''$,
	it follows from prefix transitivity that $\run_c[0..i_{\bar \procC}] \leq \run''$,
	thus re-establishing the induction hypothesis for $w''y$ with $\run''$ as a witness run that begins with $\run_c[0..i_{\bar \procC}]$.
	
	\FS{I found it a bit hard to follow this structurally but I think this is similar to an earlier proof?}
	This concludes our proof that $I(w')$ contains runs that begin with $\run_c[0..i_{\bar \procC}]$, and in turn our proof by contradiction that $y_{\bar \procC}$ must be a send event.

	Having established that $l_{i_{\bar \procC}}$ is a send event for $\bar \procC$, we can now reason from the canonicity of $\CLTS{T}$ and SC and conclude that there exists an outgoing transition from $\run_c[i_{\bar \procC}]$ and a maximal suffix such that the resulting run no longer disagrees with $w \wproj_{\Alphabet_{\bar \procC}}$. 
	The reasoning is identical to that which is used to construct our candidate run $\run_c$, and is thus omitted. 
	We update our candidate run $\run_c$ with the correct transition label and maximal suffix, update the set of states $\mathcal{Q} \in \Procs$ to the new set of participants that disagree with the new candidate run, and repeat the construction above on the new candidate run until $\mathcal{Q}$ is empty.
	
	Termination is guaranteed in at most $|w|$ steps by the fact that the number of symbols in $w$ that agree with the candidate run up to $i_{\bar \procC}$ must increase.
	
	Upon termination, the resulting $\run_c$ serves as our witness for $\bar \run$ and $\bar \run$ thus satisfies the final remaining property \ref{claim:soundness-snd-case-correct-prefix}: $\bar \run \in I(w)$. 
	This concludes our proof by induction of the prefix-preservation of send transitions.
\end{proof}
\completenessThm* 
\begin{proof}
	Let communicating LTS $\CLTS{B}$ implement $\mathcal{S}$. \fs{Minor: $B$ is the natural choice if the CLTS use $A$ usually. In contrast to earlier proofs, like the one in CAV, we do not need $T$ here so this could be used, couldn't it?}
	Specifically, we contradict protocol fidelity, and show that $\lang(\mathcal{S}) \neq \lang(\CLTS{B})$ by constructing a witness $v_0$ satisfying:
	\begin{enumerate}[(a)]
		\item \label{cond:word-is-trace}
		$v_0$ is a trace of $\CLTS{B}$, and
		\item \label{cond:word-I-set-empty}
		$I(v_0) = \emptyset$.
	\end{enumerate}

	The reasoning for the sufficiency of the above two conditions is as follows. 
	To prove the inequality of the two languages, it suffices to prove the inequality of their respective prefix sets, i.e. 
	\[
	\pref(\lang(\mathcal{S})) \neq \pref(\lang(\CLTS{B})) \enspace.
	\]
	Specifically, we show the existence of a $v \in \AlphAsyncSubscript^*$ such that
	\begin{align*}
		v &\in \{ u \mid u \leq w \land w \in \lang(\CLTS{B})\}
		~\land \\
		v &\notin \{ u \mid u \leq w \land w \in \lang(\mathcal{S})\} \enspace.
	\end{align*}
	Because $\CLTS{B}$ is deadlock-free by assumption, every trace can be extended to a maximal trace. 
	Therefore, every trace $v \in \AlphAsyncSubscript^*$ of $\CLTS{B}$ is a member of the prefix set of $\CLTS{B}$, \ie
	\[
	\exists ~(\vec{s},\xi).~(\vec{s}_0, \xi_0) \xrightarrow{v} \mathrel{\vphantom{\to}^*} (\vec{s}, \xi) 
	\implies v \in \{ u \mid u \leq w \land w \in \lang(\CSM{B})\}\enspace.
	\] 
	For any $w \in \lang(\mathcal{S})$, it holds that $I(w) \neq \emptyset$.

	Because $I(\hole)$ is monotonically decreasing, if $I(w)$ is non-empty then for any $v \leq w$, $I(v)$ is non-empty.
	By the following, to show that a word $v$ is not a member of the prefix set of $\lang(\mathcal{S})$ it suffices to show that $I(v)$ is empty:
	\[
	\forall v \in \AlphAsync^*.~
	I(v) = \emptyset \implies \forall w. ~v \leq w \implies w \notin \lang(\mathcal{S})\enspace.
	\]
	
	\myparagraph{Send Coherence. }  
	Assume that SC does not hold for some transition 
	$s_1 \xrightarrow{\msgFromToNS{\procA}{\procB}{\val}} s_2 \in T$. 
	The negation of SC says that there exists a simultaneously reachable state with no post-state reachable on $\msgFromTo{\procA}{\procB}{\val}$. 
	Formally, let $s \in S$ be a state with $s \neq s_1$ and $u \in \Alphabet^*_\procA$ be a word such that $s_0 \xRightarrow[\procA]{u} \Kleenestar s_1, s$. Then, there does not exist $s' \in S$ such that $s \xRightarrow[\procA]{\msgFromToNS{\procA}{\procB}{\val}} \Kleenestar s'$.
	
	Because $s_0 \xRightarrow[\procA]{u} \Kleenestar s$, there exists a run $\alpha \cdot s$ such that $\SyncToAsync(\trace(\alpha \cdot s)) \wproj_{\Alphabet_{\procA}} = u$. 
	
	Let $\bar w$ be $\SyncToAsync(\trace(\alpha \cdot s))$. 
	Let $\bar w \cdot \snd{\procA}{\procB}{\val}$ be our witness $v_0$; we show that $v_0$ satisfies \ref{cond:word-is-trace} and \ref{cond:word-I-set-empty}. 
	
	Because $\CLTS{B}$ implements $\mathcal{S}$, $\bar w$ is a trace of $\CLTS{B}$ and there exists a configuration $(\vec{t},\xi)$ of $\CLTS{B}$ such that
	$(\vec{t}_0,\xi_0) \xrightarrow{\bar w}\mathrel{\vphantom{\to}^*} (\vec{t},\xi)$.
	Because \mbox{$s_0 \xRightarrow[\procA]{u} \Kleenestar s_1$}, there again exists a run $\alpha_1 \cdot s_1$ such that $\SyncToAsync(\trace(\alpha_1 \cdot s_1)) \wproj_{\Alphabet_{\procA}} = u$.
	Thus, $\SyncToAsync(\trace(\alpha_1 \cdot s_1 \xrightarrow{\msgFromToNS{\procA}{\procB}{\val}} s_2))$ is a prefix of $\lang(\mathcal{S})$ and consequently, 
	$\SyncToAsync(\trace(\alpha_1 \cdot s_1 \xrightarrow{\msgFromToNS{\procA}{\procB}{\val}} s_2)) \wproj_{\Alphabet_{\procA}}$ is a prefix of $\lang(B_\procA)$. 
	In other words, $u \cdot \snd{\procA}{\procB}{\val}$ is a prefix of $\lang(B_\procA)$. 
	Because $B_\procA$ is deterministic, there exists an outgoing transition from $\vec{s}_\procA$ labeled with $\snd{\procA}{\procB}{\val}$.
	Because send transitions are always enabled in a communicating LTS, 
	$\bar w \cdot \snd{\procA}{\procB}{\val}$ is a trace of $\CLTS{B}$. 
	Thus, \ref{cond:word-is-trace} is established for $v_0$. 
	
	It remains to show that $v_0$ satisfies \ref{cond:word-I-set-empty}, namely $I(\bar w \cdot \snd{\procA}{\procB}{\val}) = \emptyset$.
	
	\textit{Claim.} All runs in $I(\bar w)$ begin with $\alpha \cdot s$.
	
	\textit{Proof of Claim.}

	This claim follows from the fact that $\mathcal{S}$ is deterministic and sender-driven. 
	Assume by contradiction that $\run' \in I(\bar w)$ and $\run'$ does not begin with $\alpha \cdot s$. 
	Because $\alpha \cdot s \neq \run'$, and $\mathcal{S}$ is deterministic, 
	$\trace(\alpha \cdot s) \neq \trace(\run')$. 
	Let $l = \trace(\alpha \cdot s)$ and let $l' = \trace(\run')$. 
	Moreover, let $\bar l$ be the largest common prefix of $l$ and $l'$. \fs{minor: $l$ is a single element usually, isn't it?}
	From the assumption that $\mathcal{S}$ is sender-driven, the first divergence between the traces of any two runs must correspond to a send action by some participant. 
	Let $\procA'$ be the sender in the first divergence between $l$ and $l'$.
	Because $\run' \in \globcomplocal{\mathcal{S}}{\procA'}{\bar w}$, it holds that
	$\bar w \wproj_{\Alphabet_{\procA'}} \leq \SyncToAsync(\trace(\rho')) \wproj_{\Alphabet_{\procA'}}$.
	We can rewrite the inequality as
	$\SyncToAsync(l) \wproj_{\Alphabet_{\procA'}} \leq \SyncToAsync(l') \wproj_{\Alphabet_{\procA'}}$.
	
	Because $\bar l$ is the largest common prefix shared by $l$ and $l'$, $\SyncToAsync(l) \wproj_{\Alphabet_{\procA'}}$ and $\SyncToAsync(l') \wproj_{\Alphabet_{\procA'}}$ are respectively of the form 
	$\bar l \wproj_{\Alphabet_{\procA'}} \cdot \snd{\procA'}{\procB_i}{\val'_i} \cdot z'$ and 
	$\bar l \wproj_{\Alphabet_{\procA'}} \cdot \snd{\procA'}{\procB_j}{\val'_j} \cdot y'$, with 
	$\procB_i \neq \procB_j$ or $\val'_i \neq \val'_j$.
	From this and 
	$\bar l \wproj_{\Alphabet_{\procA'}} \cdot \snd{\procA'}{\procB_i}{\val'_i} \cdot z' \leq 
	\bar l \wproj_{\Alphabet_{\procA'}} \cdot \snd{\procA'}{\procB_j}{\val'_j} \cdot y'$, we arrive at a contradiction.

	\textit{End Proof of Claim.}
	
	Because $I(\hole)$ is monotonically decreasing, $I(v_0) \subseteq I(\bar w)$. 
	With \textit{Claim}, every run in $I(v_0)$ begins with $\alpha \cdot s$.
	From the negation of SC, there does not exist $s' \in S$ such that $s \xRightarrow[\procA]{\msgFromToNS{\procA}{\procB}{\val}} \Kleenestar s'$, and thus there does not exist a maximal run $\bar \run \in \mathcal{S}$ such that 
	$v_0 \wproj_{\Alphabet_{\procA}} \leq \SyncToAsync(\trace(\bar \run)) \wproj_{\Alphabet_{\procA}}$. 
	
	Therefore, 
	$\globcomplocal{\mathcal{S}}{\procA}{\bar w \cdot \snd{\procA}{\procB}{\val}} = \emptyset$, and 
	$I(\bar w \cdot \snd{\procA}{\procB}{\val}) = \emptyset$
	follows.
	
	This concludes our proof by contradiction for the necessity of SC. 
	
	\myparagraph{Receive Coherence.}
	Assume that RC does not hold for a pair of transitions 
	$s_1 \xrightarrow{\msgFromToNS{\procA}{\procB}{\val}} s_2, s \xrightarrow{\msgFromToNS{\procC}{\procB}{\val}} s' \in T$. 
	
	Then, $s \neq s_1$, $\procC \neq \procA$ and let $u \in \Alphabet_{\procB}^*$ be a word such that 
	$s_0 \xRightarrow[\procB]{u} \Kleenestar s_1, s$. 
	Furthermore there exists 
	$w \in \pref(\lang(\mathcal{S}_{s'}))$ with
	$w \wproj_{\Alphabet_{\procB}} = \emptystring \land 
	\MsgVals(w \wproj_{\snd{\procA}{\procB}{\_}}) =
	\MsgVals(w \wproj_{\rcv{\procA}{\procB}{\_}}) \cdot \val$.

	Because 
	$s_0 \xRightarrow[\procB]{u} \Kleenestar s_1, s$ and 
	$s \xrightarrow{\msgFromToNS{\procC}{\procB}{\val}} s'$, 
	there exists a run $\alpha \cdot s \xrightarrow{\msgFromToNS{\procC}{\procB}{\val}} s'$ such that 
	$\SyncToAsync(\trace(\alpha \cdot s)) \wproj_{\Alphabet_{\procB}} = u$. 
	
	Let $\SyncToAsync(\trace(\alpha \cdot s)) \cdot \snd{\procC}{\procB}{\val} \cdot w \cdot \rcv{\procA}{\procB}{\val}$ be our witness $v_0$; we show that $v_0$ satisfies \ref{cond:word-is-trace} and \ref{cond:word-I-set-empty}. 
	
	First, we show that $v_0$ is a trace of $\CLTS{B}$. 
	We reason about each extension of $v_0$ in turn, starting with $\SyncToAsync(\trace(\alpha \cdot s))$. 
	It is clear that $\SyncToAsync(\trace(\alpha \cdot s))$ is a trace of $\CLTS{B}$: this follows immediately from the assumption that $\CLTS{B}$ implements $\mathcal{S}$. 
	Let $(\vec{s}, \xi)$ be the $\CLTS{B}$ configuration reached on $\SyncToAsync(\trace(\alpha \cdot s))$:
	\[
	(\vec{s_0}, \xi_0) \xrightarrow{\SyncToAsync(\trace(\alpha \cdot s))} \Kleenestar (\vec{s}, \xi)
	\] 
	
	Next, we reason about the extension $\snd{\procC}{\procB}{\val} \cdot w$ together. 
	We first establish that $\snd{\procC}{\procB}{\val} \cdot w \in \pref(\lang(\mathcal{S}_s))$.
	Because $w \in \pref(\lang(\mathcal{S}_{s'}))$, there exists a maximal run $s' \cdot \beta$ such that $s' \cdot \beta \in I(w)$. \fs{FIX: parametrise $I(w)$?}
	Observe that $s \xrightarrow{\msgFromToNS{\procC}{\procB}{\val}} s' \cdot \beta \in I(\snd{\procC}{\procB}{\val} \cdot w)$ and that $\snd{\procC}{\procB}{\val} \cdot w$ remains channel-compliant due to the assumption that $w \wproj_{\Alphabet_{\procB}} = \emptystring$. \fs{``remains'' refers here to $w$, does it? It would help if this was more explicit.}
	Thus, by \cref{lm:cc-intersection-nonemptiness-implies-prefix} it holds that $\snd{\procC}{\procB}{\val} \cdot w \in \pref(\lang(\mathcal{S}_s))$.
	Therefore, $\SyncToAsync(\trace(\alpha \cdot s)) \cdot \snd{\procC}{\procB}{\val} \cdot w \in \pref(\lang(\mathcal{S}))$, and by the assumption that $\CLTS{B}$ implements $\mathcal{S}$, 
	$\SyncToAsync(\trace(\alpha \cdot s)) \cdot \snd{\procC}{\procB}{\val} \cdot w$ is a trace of $\CLTS{B}$: 
	\[
	(\vec{s_0}, \xi_0) \xrightarrow{\SyncToAsync(\trace(\alpha \cdot s))} \Kleenestar (\vec{s}, \xi) \xrightarrow{\snd{\procC}{\procB}{\val} \cdot w} (\pvec{s}', \xi')
	\]
	Finally, we reason about the extension $\rcv{\procA}{\procB}{\val}$. 
	We show that there exists a $\CLTS{B}$ configuration $(\pvec{s}'', \xi'')$ such that
	$(\pvec{s}', \xi') \xrightarrow{\rcv{\procA}{\procB}{\val}} (\pvec{s}'', \xi')$.
	To do so, we need to show that
	\begin{enumerate}[label=(\arabic*)]
		\item \label{cond:outgoing-1}
		there exists an outgoing transition labeled with $\rcv{\procA}{\procB}{\val}$ from $\pvec{s}'_\procB$, and
		\item \label{cond:channels-2}
		$\xi'(\procA,\procB) = \val \cdot u'$, with $u' \in \MsgVals^*$.
	\end{enumerate}
	
	We know that
	$s_0 \xRightarrow[\procB]{u} \Kleenestar s_1$ and 
	$s_1 \xrightarrow{\msgFromToNS{\procA}{\procB}{\val}} s_2$, 
	so there exists a run $\alpha_1 \cdot s_2$ such that
	$\SyncToAsync(\trace(\alpha_1 \cdot s_2)) \wproj_{\Alphabet_{\procB}} = u \cdot \rcv{\procA}{\procB}{\val}$. 
	Because 
	$\SyncToAsync(\trace(\alpha_1 \cdot s_2)) \wproj_{\Alphabet_{\procB}} \in \pref(\lang(\mathcal{S})) \wproj_{\Alphabet_{\procB}}$ and $\CLTS{B}$ implements $\mathcal{S}$,
	it follows that 
	$u \cdot \rcv{\procA}{\procB}{\val} \in \pref(\lang(B_\procB))$.
	Let $t \in Q_\procB$ be the state reached on $u$ in $B_\procB$.
	The state $t$ is unique since $B_\procB$ is deterministic.
	Because $u \cdot \rcv{\procA}{\procB}{\val}$ is a prefix in $B_\procB$, there exists a transition
	$t \xrightarrow{\rcv{\procA}{\procB}{\val}} t_1 \in \delta_\procB$. 
	It holds that $(\SyncToAsync(\trace(\alpha \cdot s)) \cdot \snd{\procC}{\procB}{\val} \cdot w) \wproj_{\Alphabet_{\procB}} = u$, so it follows that $\vec{s'}_\procB = t$ and there exists an outgoing transition from $\vec{s'}_\procB$ labeled with $\rcv{\procA}{\procB}{\val}$. This establishes \ref{cond:outgoing-1}.
	
	\ref{cond:channels-2} is established from the fact that send actions are immediately followed by their matching receive action in $\SyncToAsync(\trace(\alpha \cdot s))$, and therefore all channels in $\xi$ are empty, including $\xi(\procA,\procB)$.
	Because $\snd{\procC}{\procB}{\val}$ does not concern $\xi(\procA,\procB)$, $\val$ remains the first unmatched send action from $\procA$ to $\procB$ in $\SyncToAsync(\trace(\alpha \cdot s)) \cdot \snd{\procC}{\procB}{\val} \cdot w$, \fs{What about $w$? Add that ``$w \wproj_{\Alphabet_\procB} = \emptystring$ by assumption and thus there cannot be any receive event for $\procB$ in $w$''?}
	and thus $\val$ is at the head of channel $\xi'(\procA,\procB)$:
	\[
	(\vec{s_0}, \xi_0) \xrightarrow{\SyncToAsync(\trace(\alpha \cdot s))} \Kleenestar (\vec{s}, \xi) \xrightarrow{\snd{\procC}{\procB}{\val} \cdot w} (\vec{s'}, \xi') \xrightarrow{\rcv{\procA}{\procB}{\val}} (\vec{s''}, \xi'')
	\enspace .
	\]
	
	This concludes our proof of \ref{cond:word-is-trace}.
	
	Next, we argue that $I(v_0) = \emptyset$. This claim follows trivially from the observation that every run in $I(v_0)$ must begin with $\alpha \cdot s \xrightarrow{\msgFromToNS{\procC}{\procB}{\val}} s'$, \fs{FIX: I believe it but I do not see where it comes from directly.}
	and therefore $v_0$ must satisfy
	$v_0 \wproj_{\Alphabet_{\procB}} \leq u \cdot \rcv{\procC}{\procB}{\val}$, yet 
	$v_0 \wproj_{\Alphabet_{\procB}} = u \cdot \rcv{\procA}{\procB}{\val}$ and we find a contradiction. 
	
	\myparagraph{No Mixed Choice. }  
	Assume that NMC does not hold for a pair of transitions 
	$s_1 \xrightarrow{\msgFromToNS{\procA}{\procB}{\val}} s_2, s \xrightarrow{\msgFromToNS{\procC}{\procA}{\val}} s' \in T$. 
	The negation of NMC says that $s_1$ and $s$ are simultaneously reachable.
	Let $u \in \Alphabet_{\procB}^*$ be a word such that 
	$s_0 \xRightarrow[\procB]{u} \Kleenestar s_1, s$. 
	
	Because $s_0 \xRightarrow[\procA]{u} \Kleenestar s$, there exists a run $\alpha \cdot s$ such that $\SyncToAsync(\trace(\alpha \cdot s)) \wproj_{\Alphabet_{\procA}} = u$. 
	
	Let $\bar w$ be $\SyncToAsync(\trace(\alpha \cdot s))$. 
	Let $\bar w \cdot \snd{\procC}{\procA}{\val} \cdot \snd{\procA}{\procB}{\val}$ be our witness $v_0$; we show that $v_0$ satisfies \ref{cond:word-is-trace} and \ref{cond:word-I-set-empty}. 
	
	The reasoning is similar to that for the witness constructed for Send Coherence Condition, and is thus omitted. 
\end{proof}
\section{Additional Material for Section~\ref{sec:symbolic}}
\label{app:symbolic}
\globalTypeImplementability*
\begin{proof}
	The arguments for co-NP membership of implementability for global types are identical to those for general finite protocols, and are thus omitted. 
	
	As in the proof of Theorem 5.8, we show NP-hardness of non-implementability via a reduction from the 3-SAT problem.
	Assume a 3-SAT instance $\varphi = C_1 \land \ldots \land C_k$. Let $x_1,\dots,x_n$ be the variables occurring in $\varphi$ and let $L_{ij}$ be the $j$th literal of clause~$C_i$, with $1 \leq i \leq k$ and $1 \leq j \leq 3$. 
	
	We construct a global type $\GG_\varphi$ over participants $\Procs = \{\procA,\procB, \procC, \roleFmt{x_1},\roleFmt{\overline{x}_1}, \dots,\roleFmt{x_n},\roleFmt{\overline{x}_n}\}$, such that $\varphi$ is satisfiable iff $\GG_\varphi$ is implementable. 
	In particular, we ensure that $\GG_\varphi$ is implementable iff $\avail_{\procA,\procB,\{\procB\}}(\val,G')$ does not hold for some subterm $G'$ in $\GG_\varphi$. 
	
	The construction idea for $\GG_\varphi$ is identical to that for $\mathcal{S}_\varphi$ from Theorem 5.8, but with several modifications to yield a tree-shaped protocol which corresponds to a global type. 
	First, for each branching state from which $\procC$ selects variables or clauses, represented as $\mu t$ terms, we introduce a new branch that acts as a forward edge connecting to the next branching state. 
	Because branches in a global type can only join at a single state via recursion variables, and recursion variables must appear in scope of their $\mu t$ terms, variable and clause selection proceeds by recursing ``backwards'' towards the top-level global type. Due to this reversal of traversal order, the initial choice by $\procC$ and the message exchange $\msgFromToNS{\procA}{\procB}{\val}$ potentially violating Receive Coherence swap places in the protocol. 
	The construction of global type $\GG_\varphi$ is detailed below:

	\begin{enumerate}
		\item Define for every variable $x_i$ with $2 < i < n$ a global type $G_{x_i}$ representing a truth assignment to variable $x_i$ as follows: 
		\[\small
		G_{x_i} \is 
		\mu t_{x_i}.~
		+ \;
		\begin{cases}
			\msgFromToNS{\procC}{\roleFmt{x_i}}{\bot}. \,
			\msgFromToNS{\procC}{\roleFmt{\overline{x}_i}}{\top}. \,
			\msgFromToNS{\procC}{\procB}{\val_{x_i}}. \,
			\msgFromToNS{\procB}{\roleFmt{x_i}}{\val}. \,
			t_{x_{i+1}} \, 
			\\
			\msgFromToNS{\procC}{\roleFmt{\overline{x}_i}}{\bot}. \,
			\msgFromToNS{\procC}{\roleFmt{x_i}}{\top}. \,
			\msgFromToNS{\procC}{\procB}{\val_{\overline{x}_i}}. \,  \msgFromToNS{\procB}{\roleFmt{\overline{x}_i}}{\val}. \,
			t_{x_{i+1}}\, 
			\\ 
			\msgFromToNS{\procC}{\roleFmt{x_i}}{next}. \,
			\msgFromToNS{\procC}{\roleFmt{\overline{x}_i}}{next}. \,
			\msgFromToNS{\procC}{\procB}{next}. \,
			G_{x_{i-1}} \, 
		\end{cases}
		\]
		For $x_2$ and $x_n$, the construction is modified as follows.  
		For $G_{x_n}$, the recursion variable in the first and second branches is replaced with $t_{C_1}$. 
		For $G_{x_2}$, the following is added before $G_{x_1}$ in the third branch:  
		\[\small
		\msgFromToNS{\procC}{\procB}{\lastMsg}. \,
		\msgFromToNS{\procC}{\procA}{\lastMsg}. \,
		\msgFromToNS{\procC}{\roleFmt{x_1}}{\lastMsg}. \,
		\msgFromToNS{\procC}{\roleFmt{\overline{x}_1}}{\lastMsg}. \,
		\msgFromToNS{\procB}{\procA}{\val}. \,
		\msgFromToNS{\procB}{\roleFmt{x_1}}{\val}. \,
		\msgFromToNS{\procB}{\roleFmt{\overline{x}_1}}{\val}. \,
		\]
		
		\item Define for every clause $C_i = L_{i1} \lor L_{i2} \lor L_{i3}$ with $2 \leq i < k$ a global type $G_{C_i}$ as follows, where $x_{ij}$ is defined as $\roleFmt{x}$ if $L_{ij}=x$ and $\roleFmt{\overline{x}}$ if $L_{ij} = \neg x$: 
		\[\small
		G_{C_i} \is 
		\mu t_{C_i}.~
		+ \;
		\begin{cases}
			\Sigma_{j = 1..3} \,
			\msgFromToNS{\procC}{x_{ij}}{\val}. \,
			\msgFromToNS{\procC}{\procA}{\val_{x_{ij}}}. \,
			\msgFromToNS{x_{ij}}{\procA}{\val}. \,
			t_{C_{i+1}} \, 
			\\
			\msgFromToNS{\procC}{x_{i1}}{\nextMsg}. \,
			\msgFromToNS{\procC}{x_{i2}}{\nextMsg}. \,
			\msgFromToNS{\procC}{x_{i3}}{\nextMsg}. \,
			\msgFromToNS{\procC}{\procA}{\nextMsg}. \,
			G_{C_{i-1}} \, 
		\end{cases}
		\]
		For $C_1$ and $C_k$, the construction is modified as follows.

		For $G_{C_1}$, the last branch continues with $G_{x_n}$. 
		For $G_{C_k}$, the recursion variable in the first three branches is replaced with $t$. 
		
		\item Define $G_{x_1}$ for variable $x_1$ as follows: 
		\[\small
		G_{x_1} \is 
		+ \;
		\begin{cases}
			\msgFromToNS{\procC}{\procA}{\val_1}. \,
			\msgFromToNS{\procC}{\procB}{\val}. \, 
			\overline{G}
			\\
			\msgFromToNS{\procC}{\procA}{\val_2}. \,
			\msgFromToNS{\procA}{\procB}{\val}. \,
			0 \, 
		\end{cases}
		\\
		\overline{G} \is 
		+ \;
		\begin{cases}
			\msgFromToNS{\procC}{\roleFmt{x_1}}{\bot}. \, 
			\msgFromToNS{\procC}{\roleFmt{\overline{x}_1}}{\top}. \, 
			\msgFromToNS{\procC}{\procB}{\val_{x_1}}. \, 
			\msgFromToNS{\procB}{\roleFmt{x_1}}{\val}. \,
			t_{x_2} \, 
			\\
			\msgFromToNS{\procC}{\roleFmt{\overline{x}_1}}{\bot}. \,
			\msgFromToNS{\procC}{\roleFmt{x_1}}{\top}. \,
			\msgFromToNS{\procC}{\procB}{\val_{\overline{x}_1}}. \,  \msgFromToNS{\procB}{\roleFmt{\overline{x}_1}}{\val}. \,
			t_{x_2}\, 
		\end{cases}
		\]
	\end{enumerate}
	
	The global type $\GG_\varphi$ is thus defined as: 
	\[
	\GG_\varphi \is 
	\mu t.~\msgFromToNS{\procC}{\procB}{\topMsg}. \, 
	\msgFromToNS{\procA}{\procB}{\val}. \,
	G_{C_k}
	\]
	
	Observe that $\GG_\varphi$ is linear in the size of $\varphi$. 
	
	We first establish that $\avail_{\procA,\procB,\{\procB\}}(\val,\overline{G})$ holds in $\GG_\varphi$ iff $\varphi$ is satisfiable. Observe that the $G_{x_i}$'s contain two branches that recurse ``backwards'' to the previous $G_{x_{i+1}}$, and one branch that proceeds ``forwards'' towards $G_{x_1}$. Each time a backward branch is taken, either $\roleFmt{x_i}$ or $\roleFmt{\overline{x}_i}$ is added to the blocked set $\blockedset$ along the path. Forward branches do not change the blocked set, as participant $\procB$ does not send messages in them. Thus, the path computed by $\avail_{\procA,\procB,\{\procB\}}(\val,\overline{G})$ from $\overline{G}$ to $G_{C_1}$ must contain for each variable $x_i$ either $\roleFmt{x_i}$ or $\roleFmt{\overline{x}_i}$. 
	The blocked set $\blockedset$ thus encodes a truth assignment $\rho_\blockedset$ for the $x_i$'s where $\rho_\blockedset(x_i)=\top$ iff $\roleFmt{x_i} \not\in \blockedset$.
	By construction of $G_{x_i}$, for every truth assignment $\rho$, there exists at least one path between $\overline{G}$ and $G_{C_1}$ such that $\rho=\rho_\blockedset$ for the blocked set $\blockedset$ computed along that path. 
	
	The $G_{C_i}$ terms allow $\procA$ to proceed backwards towards $\GG_\varphi$ by selecting a branch whose participant $x$ is not in $\blockedset$, i.e. $C_i$ is satisfied by $\rho_\blockedset$.
	Thus, a path from $G_{C_1}$ to $\GG_\varphi$ adds $\procA$ to $\blockedset$ at $t_i$ iff $\rho_\blockedset$ does not satisfy at least one of the clauses $C_i$.
	Therefore, $\val$ is available in $\overline{G}$ iff there exists a $\blockedset$ such that $\rho_\blockedset$ satisfies $\varphi$.
	
	The reasoning that $\GG_\varphi$ is implementable iff $\avail_{\procA,\procB,\{\procB\}}(\val,\overline{G})$ does not hold again follows that for $\mathcal{S}_\varphi$, and below we only discuss new behavior introduced by the structural changes to $\mathcal{S}_\varphi$. 
	
	Participant $\procC$ still dictates the control flow in the global type, but now additionally sends $\nextMsg$ messages to inform participants in the branch when a forward edge is taken, $\lastMsg$ messages to inform $\procA, \procB, \roleFmt{x_1}$ and $\roleFmt{\overline{x}_1}$ when the last forward edge is taken, and $\topMsg$ to $\procB$ to inform $\procB$ to receive $\val$ from $\procB$.  
	Receiving $\nextMsg$ messages means inaction for all other participants. 
	Receiving $\lastMsg$ prompts $\procB$ to send a message to $\procA, \roleFmt{x_1}$ and $\roleFmt{\overline{x}_1}$, which they anticipate by receiving $\lastMsg$ first from $\procC$. 
	
	As before, the only potential source of non-implementability lies in participant $\procB$, who can violate Receive Coherence for transitions labeled with 
	$\msgFromToNS{\procC}{\procB}{\val}$ and 
	$\msgFromToNS{\procA}{\procB}{\val}$ in $G_{x_1}$ when 
	$\avail_{\procA,\procB,\{\procB\}}(\val,\overline{G})$ does not hold, and the message from $\procA$ can be received out of order. 
	
	We obtain that $\GG_\varphi$ is non-implementable iff $\avail_{\procA,\procB,\{\procB\}}(\val,\overline{G})$ holds in $\GG_\varphi$ iff $\varphi$ is satisfiable.
\end{proof}
 \end{document}